\documentclass[a4paper,USenglish,cleveref, autoref, thm-restate]{lipics-v2021} 

\hideLIPIcs
\usepackage{alphalph}

\usepackage{mathtools}
\usepackage{cite}
\usepackage{amsmath, hyperref, url}
\usepackage{graphicx}
\usepackage{macros}
\usepackage{todonotes}
\usepackage{xspace}

\newcommand\upe{\textup{\textsc{Upward Planar Drawing Extension}}\xspace}
\newcommand\upef{\textup{\textsc{Upward Planar Drawing Extension in a Face}}\xspace}
\newcommand\upeshort{\textsc{\textup{UPE}}\xspace}
\newcommand\upefshort{\textsc{\textup{UPEF}}\xspace}
\newcommand\upefshorto{\textsc{\textup{UPEF-Outer}}\xspace}

\nolinenumbers
\hideLIPIcs

\renewenvironment{claimproof}{\begin{proof}[{\upshape{\noindent\underline{Proof of the Claim}}}]}{\end{proof}}

\newcommand{\bigoh}{\mathcal O}
\newcommand{\III}{\mathcal I}
\usepackage{amssymb, complexity}
\usepackage{boxedminipage}

\title{A Fixed-Parameter Algorithm for Extending Upward Planar Drawings 
}

\EventEditors{Maarten L\"{o}ffler and Silvia Miksch}
\EventNoEds{2}
\EventLongTitle{34th International Symposium on Graph Drawing and Network Visualization (GD 2026)}
\EventShortTitle{GD 2026}
\EventAcronym{GD}
\EventYear{2026}
\EventDate{August 19--21, 2026}
\EventLocation{Ontario, Canada}
\EventLogo{}
\SeriesVolume{396}
\ArticleNo{4}

\titlerunning{A Fixed-Parameter Algorithm for Extending Upward Planar Drawings}

\author{Vera Chekan}{Humboldt-Universität zu Berlin, Germany}{vera.chekan@hu-berlin.de}{https://orcid.org/0000-0002-6165-1566}{}

\author{Robert Ganian}{Algorithms and Complexity Group, TU Wien, Austria}{rganian@ac.tuwien.ac.at}{https://orcid.org/0000-0002-7762-8045}{Austrian Science Fund (FWF Project 10.55776/Y1329) and Vienna Science and Technology Fund (WWTF Project 10.47379/ICT22029)}

\author{Viktoriia Korchemna}{University of California, Santa Barbara}{vkorchemna@ucsb.edu
}{https://orcid.org/0000-0001-8038-905X}{Austrian Science Fund (FWF, projects 10.55776/Y1329 and 
10.55776/J4880) and US National Science Foundation (NSF Grant CCF-2505099)}

 \authorrunning{V.\ Chekan, R.\ Ganian, V.\ Korchemna} 
 \Copyright{Vera Chekan, Robert Ganian, Viktoriia Korchemna}



\category{Track 1, Long Paper}

\relatedversion{Accepted at Graph Drawing and Network Visualization (GD 2026).}

\begin{document}
\maketitle

\begin{abstract}
An upward planar drawing of a directed acyclic graph is a planar drawing where every edge is pointed upward from its tail to head. 
Upward planar drawings are among the most natural drawing styles of directed graphs and have been researched in a variety of different settings, recently including that of drawing extension. In the drawing extension setting, one asks: given a graph $G$ and a (typically connected) subgraph $H$ of $G$ with a drawing $\Gamma(H)$, can we complete $\Gamma(H)$ to a drawing of $G$? 

Drawing extension problems have been studied for numerous drawing styles; the vast majority of these are \NP-hard and a typical approach aimed at circumventing their intractability is to design parameterized algorithms where the parameter measures ``how much'' of $G$ is still missing from the pre-drawn graph $H$. Most algorithms obtained within this framework require only a small number of edges to be missing from $H$ in order to remain efficient. In this article, we present a fixed-parameter algorithm for extending upward planar drawings which overcomes this drawback by using the \emph{vertex+edge deletion distance} as the parameter, thus achieving tractability even for instances with many missing edges.
A key ingredient towards our result is a novel characterization of ``canonical'' sets of missing edges which cross a horizontal line segment in the drawing.

\subparagraph{Generative AI Declaration}
Generative AI was not used for the preparation of this manuscript.
\end{abstract}

\section{Introduction}
\label{sec:intro}
The notion of graph planarity is one of the most fundamental concepts in graph theory. However, in some circumstances it is crucial to find a drawing of the input graph which is not only crossing-free, but also satisfies further constraints. In this article we study the notion of \emph{upward planarity}, which is among the most natural and established constrained planar drawing styles~\cite{DBLP:books/ws/NishizekiR04,DBLP:reference/crc/2013gd}.
An upward planar drawing of a directed acyclic graph $G$ is a crossing-free drawing where all edges are drawn as polylines (or, equivalently, curves) that monotonically increase in a common direction (see \cref{fig:upward-planar-drawing}).
Upward planar drawings are used, among others, in the context of visualization of hierarchical network structures and have become the focus of extensive theoretical research to date~\cite{HuttonL91,BertolazziBMT98,Mchedlidze14,LozzoBF20,ChaplickGFGRS22,JungeblutMU23}.

\begin{figure}[ht]
\centering
\includegraphics[scale = 0.8]{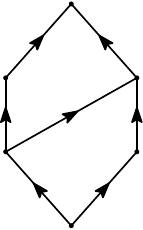}
\caption{An upward planar drawing of a directed graph.}
\label{fig:upward-planar-drawing}
\end{figure}

The problem of computing an upward planar drawing of a given graph $G$ is known to be \NP-hard~\cite{GargT01}, and in fact remains intractable even on tree-like graphs~\cite{BMPJansenGD}. Given these lower bounds, the problem has also been investigated from the perspective of parameterized complexity, and is by now known to be fixed-parameter tractable when parameterized by the number of sources~\cite{ChaplickGFGRS22}. However, in many settings one is not tasked with computing a drawing for $G$ from ``scratch''---instead, it may well happen that most of $G$ is already drawn, and we merely need to complete the missing parts. This gives rise to the by now well-established line of research into \emph{drawing extension problems}. In such problems, we are given a graph $G$ and a subgraph $H$ of $G$ with a drawing $\Gamma(H)$, which is called a \emph{partial drawing} of $G$; the task is then to complete $\Gamma(H)$ to a drawing of~$G$~\cite{Patrignani06,KlavikKKW12,adfjkp-tppeg-15,ChaplickFK19,EibenGHKN20, mfcs/EibenGHKN20,GanianHKPV21,HammH22,BhoreGKMN23,KlemzS24,DepianFFGN25,DFGN25}.

A fundamental result in this line of research is the work of Angelini et al.~\cite{adfjkp-tppeg-15}, who showed that for planar graphs with a given partial planar drawing, the extension problem can be solved in linear time, thus matching the time complexity of unconstrained planarity testing.
In fact, there is also a corresponding combinatorial characterization of planar graphs with extensible partial planar drawings via forbidden substructures~\cite{jkr-ktppeg-13}.
In contrast to the above results, drawing extension is \NP-hard for many other drawing styles, including 1-planar drawings~\cite{EibenGHKN20, mfcs/EibenGHKN20}, straight-line planar drawings~\cite{Patrignani06}, level-planar drawings~\cite{KlemzS24}, crossing-optimal simple drawings~\cite{GanianHKPV21, HammH22}, stack- and queue-layout drawings~\cite{Bookthicknessextension,DFGN25} and bend-minimal orthogonal drawings~\cite{BhoreGKMN23}. The study of upward planar drawing extension was initiated by Da Lozzo, Di Battista and Frati~\cite{LozzoBF20}, who showed that the problem of computing such an extension is \NP-hard but can be solved efficiently on digraphs with a unique source and sink.

Given the inherent intractability of most drawing extension problems, a natural question tackled in this line of research is whether one can obtain a \emph{fixed-parameter algorithm}~\cite{DowneyF13,CyganFKLMPPS15} when parameterized by a measure of how much of the input graph is still missing from the drawing. 
This line of attack is often pursued under the assumption that the provided partial drawing is connected~\cite{EibenGHKN20,GanianHKPV21,BhoreGKMN23}, as connectivity is both well-motivated in dynamic drawing settings, and it prevents the occurrence of many pairwise nested and intertwined connected components.
While several such algorithms were obtained via a parameterization bounding the number of edges missing from the provided drawing, no suitable fixed-parameter algorithms were known for extending any of the drawing styles mentioned in the previous paragraph when the task is to add many missing edges to $\Gamma(H)$~\cite{EibenGHKN20, mfcs/EibenGHKN20,GanianHKPV21, HammH22,Bookthicknessextension,BhoreGKMN23,DFGN25}. In this article, we present such an algorithm for upward planar drawing extension.

\subparagraph{Problem Statement and Related Work.}
A planar drawing $\Gamma(G)$ of a digraph $G$ is called \emph{upward} if each edge $ab$ is represented as a polyline whose $y$-coordinate monotonically increases from its tail $a$ to its head $b$.
We formalize our problem of interest below.

\pbDef{\textsc{Upward Planar Drawing Extension} (\textsc{UPE})}
{A connected acyclic digraph $G$, a connected subdigraph $H$ of $G$, and an upward planar drawing $\Gamma(H)$ of $H$.}
{Does there exist an extension $\Gamma(G)$ of $\Gamma(H)$, i.e., an upward planar drawing of $G$ such that its restriction to $H$ coincides with $\Gamma(H)$?}

The \NP-hardness of \upe\ easily follows from the fact that the case where $H$ is empty coincides with the \NP-hard problem of testing upward planarity~\cite{GargT01}.
There are two natural parameterizations that have been employed in the study of drawing extension problems. The first is the \emph{edge deletion distance}, which simply measures the cardinality of $E(G)\setminus E(H)$. The extension of 1-planar~\cite{EibenGHKN20, mfcs/EibenGHKN20}, crossing-optimal simple~\cite{GanianHKPV21} and bend-minimal orthogonal~\cite{BhoreGKMN23} drawings is known to be fixed-parameter tractable when parameterized by the edge deletion distance; however, the drawback of all these algorithmic upper bounds is that they only target instances with a small number of missing edges.\footnote{We call the edges (vertices) in $E(G)\setminus E(H)$ ($V(G)\setminus V(H)$) \emph{missing}, as they are missing in $\Gamma(H)$.} A more desirable parameterization is the \emph{vertex}+\emph{edge deletion distance}~\cite{EibenGHKN20, mfcs/EibenGHKN20,GanianMNZ21,Bookthicknessextension,DepianFFGN25,DFGN25}, i.e., the minimum number of vertex and edge deletion operations required to obtain $H$ from~$G$, whereas a single vertex deletion operation naturally also removes all of the incident edges.
Formally, this parameter is given by 
$|V(G)\setminus V(H)|+|E(G[V(H)])\setminus E(H)|$ where $G[V(H)]$ is the subgraph of $G$ induced on $V(H)$. 
Since removing one vertex may delete $\Theta(n)$ edges, this parameterization is strictly more general than the edge deletion distance.

The two parameterizations coincide on bounded-degree graphs, but in the general setting it seems much more desirable---and challenging---to obtain algorithms utilizing the vertex+edge deletion distance as the parameter. In fact, the fixed-parameter tractability of extending 1-planar drawings (which is where these parameterizations were first formalized) under the vertex+edge deletion distance arguably remains the most prominent open question in the parameterized study of extension problems~\cite{EibenGHKN20, mfcs/EibenGHKN20,GanianMNZ21}. The problem is known to become fixed-parameter tractable if one places further restrictions on the sought-after drawing~\cite{EibenGHKN20}, but for 1-planarity itself even \XP-tractability is non-trivial and was only established in a follow-up manuscript dedicated to that problem~\cite{mfcs/EibenGHKN20}. On the other hand, extending both stack and queue layouts is \NP-hard even if the vertex+edge deletion distance is $2$~\cite{Bookthicknessextension,DFGN25}.

\subparagraph{Contribution and Overview.} \label{subsec:contribution-and-overview}
In this work, we establish the fixed-parameter tractability of \textsc{Upward Planar Drawing Extension} when parameterized by the vertex+edge deletion distance. The obtained algorithm is deterministic and constructive: it can also output an extension of $\Gamma(H)$ as a witness for YES-instances. Before proceeding towards a description of our approach, we highlight two high-level challenges that had to be overcome via novel techniques and insights into the combinatorial properties of upward planar drawings.

\begin{enumerate}
\item The instances we need to work with are neither purely combinatorial, nor purely geometric, they contain elements from both worlds: while the set of missing vertices and edges is a combinatorial object, the drawing $\Gamma(H)$ imposes geometric constraints on their placement. 
\item The parameterization requires us to deal with a possibly large set of missing edges without using branching or other exponential techniques per edge: recall that even~$n$-vertex instances with a parameter value of one may contain $\Theta(n)$ missing edges.
\end{enumerate}

\subparagraph*{The Initial Setup (Section~\ref{sec:setup}).}
We begin by applying a set of individually straightforward preprocessing steps on the given instance of \textsc{UPE}. 
First, to simplify our future considerations, we subdivide each edge $e$ in $E(G)\setminus E(H)$ with both endpoints in $E(H)$. This results in an equivalent instance of our problem (as, in a hypothetical solution for the original instance, we may place the newly created vertex anywhere on the upward curve of $e$), but allows us to assume w.l.o.g.\ that each missing edge has at least one endpoint in $V(G)\setminus V(H)$ (i.e., that $H$ is induced on $V(H)$) and that the parameter is simply the number $|V(G)\setminus V(H)|$ of missing vertices. 

Second, we apply an exhaustive branching subroutine to determine a partition of the missing vertices into groups which will be placed in the same face of $\Gamma(H)$. This yields a reduction of the original problem to the special case where we only need to deal with the vertices placed in one specific face of the provided drawing. 
Third, we replace each bend point on each polyline in the drawing $\Gamma(H)$ with a degree-$2$ vertex, allowing us to assume that $\Gamma(H)$ is straight-line. And finally, we apply a simple perturbation argument to assume that no two vertices of $\Gamma(H)$ have the same $y$-coordinate (i.e., that the drawing is in general position as far as the vertical coordinate is concerned).
Modulo a single caveat---specifically, the case where we need to place some group of missing vertices into the outer face of the provided drawing (which we will return to later)---the above considerations allow us to concentrate our efforts on solving the subproblem stated below.

\vspace{0.3cm}

	\noindent\fbox{
		\begin{minipage}{0.96\linewidth}
			\begin{tabular*}{\linewidth}{@{\extracolsep{\fill}}lr} \textsc{Upward Planar Drawing Extension in a face $f$} (\textsc{UPEF}) & \\ \end{tabular*}
			{\bf{Input:}} A connected acyclic digraph $G$, a connected induced subdigraph $H$ of $G$, an upward straight-line planar drawing $\Gamma(H)$ of $H$ whose vertices each have a unique $y$-coordinate, 
            and an inner face $f$ of this drawing. \\
   {\bf{Parameter:}} $k=|V(G)\setminus V(H)|$\\ 
			{\bf{Question:}} Does there exist an extension $\Gamma(G)$ of $\Gamma(H)$ such that all vertices of $V(G)\setminus V(H)$ are placed in $f$?
		\end{minipage}
	}

\vspace{0.3cm}

\subparagraph*{Strips and Trees (Section~\ref{sec:strips}).} 
While \textsc{UPEF} may seem simpler than \textsc{UPE}, the true source of the problem's difficulty is finding a placement for $k$ missing vertices, each of which may be incident with a large number of edges, in a possibly highly non-convex face of $\Gamma(H)$. 
As our first step towards dealing with this task, we rely on a combinatorial representation of the faces in an upward planar drawing that is conceptually related to the notion of \emph{vertical decompositions}~\cite{BergCKO08}.
In particular, we partition the face $f$ into trapezoidal regions called \emph{strips} and show that these form a (rooted) tree-like structure---the \emph{strip tree} $T_f$ (see \cref{fig:strip-graph}).

\begin{figure}[t]
\centering
\includegraphics[scale = 1.0]{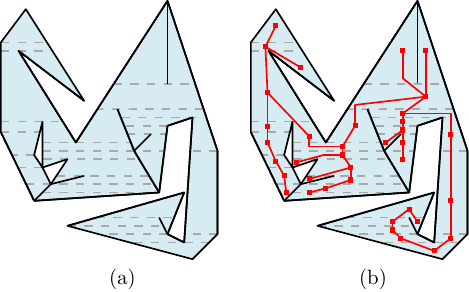}
\caption{An inner face $f$ of $\Gamma(H)$; all bends in $\Gamma(H)$ are vertices. (a) A partitioning of $f$ into strips; the gray horizontal boundaries between strips are called \emph{gates}. (b) The strip tree of $f$ is red: for each strip there is a vertex and for every gate, there is an edge between the two strips sharing it.}
\label{fig:strip-graph}
\end{figure}

\subparagraph*{Dynamic Programming Setup (Section~\ref{sec:dynprogsetup}).}
We now proceed to the core of our result, which is determining whether $k$ missing vertices may be placed inside an inner face $f$. 
For this we will perform a leaf-to-root dynamic programming procedure along $T_f$. Towards obtaining our algorithm, we show that the edges that cross the boundary of a strip (e.g., when passing towards the root of the strip tree) may be grouped into $\bigoh(k^3)$ many \emph{types} in the so-called \emph{gate word}. And each such type can essentially be treated as a single object. 
At first glance, one would expect that this already suffices to solve the problem via leaf-to-root dynamic programming along the strip tree. However, there is still a crucial obstacle that needs to be overcome: when processing a strip, it is not at all obvious which of the missing edges actually leave the region processed so far. In particular, while one may think that whether an edge $Av$ incident with precisely one missing vertex $A$ passes through a strip's horizontal boundary (also called a \emph{gate}) depends solely on whether $A$ had already been placed in the processed part of the strip tree, this is not the case if $v$ is a cut-vertex of $H$---see Figure~\ref{fig:routing-depends}.

\begin{figure}[t]
  \centering \includegraphics[width=0.4 \textwidth]{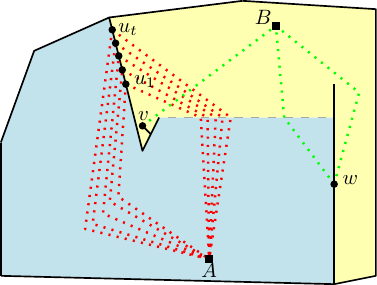} 
\caption{A ``processed'' subtree of $T_f$ (in blue). 
The missing vertex $A$ is placed in the blue region. We do not know whether any of the red dashed edges will be crossing the gate between the blue and the yellow (depicting the ``unprocessed part'') regions. If at least one such edge does so, it will be impossible to draw an edge between the missing vertex $B$ and the ``unprocessed'' vertex $v$ later on. Hence, when processing the strip tree in a bottom-up fashion one needs to keep track of which of the (possibly numerous) missing edges cross the gate separating the processed part from the future. 
}
\label{fig:routing-depends} 
\end{figure}

The key towards overcoming this obstacle lies in a non-trivial structural result which establishes a bound on the number of ``canonical'' subsets of edges that could cross the gate (Theorem~\ref{thm:relevant-records}). Essentially, the proof of the result partitions the pre-drawn vertices into \emph{bundles} 
such that for each bundle one can employ a careful iterative redrawing argument to obtain a solution which either ``draws'' all edges incident with this bundle so that they cross the gate, or ``avoids'' this gate for the whole bundle. 
We also show that the number of bundles is bounded by a function of $k$ so that the number of relevant partial solutions also is. 
One specific complication that we overcome in the proof of the correctness of the algorithm is that the bundles change over time, i.e., when the algorithm moves higher up in the strip tree.

\subparagraph*{Dynamic Programming (\cref{sec:dynprog-new}) and Handling the Outer Face (Section~\ref{sec:outer}).}
With the above results in hand, we can finally proceed to the fixed-parameter algorithm that solves \textsc{UPEF} by leaf-to-root dynamic programming. 
For this, we describe 
how one computes the records for a strip $s$ based on the records for its at most two children, which---while the details and the proof of correctness are technical and non-trivial---essentially boils down to a careful brute-force enumeration of all choices of records for each of the children and a bounded number of choices of how to place vertices and route the edges inside $s$ itself. 

The final piece of the puzzle is the case where, when constructing our reduction from \textsc{UPE} to \textsc{UPEF}, some missing vertices are placed in the outer face $f'$ of $\Gamma(H)$. 
To deal with the outer face, 
we prove a sequence of lemmas which yield a non-trivial Turing-reduction from this case to \textsc{UPEF}. 
This is achieved by establishing the existence of a bounded number of potential ``clean cuts'' in $f'$: each such cut consists of a bounding box around $\Gamma(H)$ and a subdivided straight-line edge $e$ connecting the bounding box to $\Gamma(H)$. While this construction transforms $f'$ into an inner face, the issue is that a hypothetical solution may need to insert a large number of missing edges, each crossing $e$. To account for this, in our construction we combine branching with the \emph{vertex splitting} operation~\cite{EppsteinKKLLMMV18,NollenburgSTVWW22} to carefully ``split'' a selected subset of missing vertices into two copies each; these copies are then pre-assigned to the correct sides of~$e$. An illustration of the general idea is provided in Figure~\ref{fig:outer-face}.

\section{Preliminaries}
For a non-negative integer $k$, we define the set $[k] = \{ i \mid 1 \leq i \leq k\}$. 
Similarly, we let $[k]_0 = [k] \cup \{0\}$.
We assume basic familiarity with the parameterized complexity paradigm, notably with the notion of \emph{fixed-parameter algorithms}~\cite{DowneyF13, CyganFKLMPPS15}. 
When expressing the running times, we use the~$\bigoh^*$ notation to suppress polynomial factors of the input size. 
In this work, we primarily deal with simple directed acyclic graphs, i.e., graphs with no loops and no directed cycles. 
For brevity, we refer to these as \emph{graphs} most of the time, and understand connectivity in such graphs to mean weak connectivity. 
In the few cases where we deal with undirected graphs (notably in the context of the strip tree), we explicitly refer to these as either undirected graphs or as trees. Let $G$ be a graph.
For vertices $u$ and $v$ of $G$, $uv$ denotes the edge from $u$ to $v$, i.e., having its \emph{tail} in $u$ and \emph{head} in $v$.
We use $\{u,v\}$ to refer to an edge whose two endpoints are $u$ and $v$ (without specifying its orientation). 

A (\emph{polyline}) \emph{drawing} $\Gamma(G)$ of a graph $G$ maps each vertex to a point in the plane and each edge to a polyline between the endpoints of the edge. $\Gamma(G)$ is \emph{planar} if no two edges intersect, except at common endpoints. A planar drawing partitions the plane into regions, called \emph{faces}. The bounded faces are called \emph{internal}, while the unbounded face is the \emph{outer face}. $\Gamma(G)$ is \emph{upward} if every edge of $G$ is represented as a polyline which monotonically increases in the $y$-direction from its tail to its head. 
A drawing $\Gamma(G)$ of $G$ is called an \emph{extension} of a drawing $\Gamma(H)$ of a subgraph $H$ of $G$ if the restriction of $\Gamma(G)$ to the vertices and edges of~$H$ coincides with~$\Gamma(H)$. 

We formalize our problem of interest below:
\pbDefp{\textsc{Upward Planar Drawing Extension} (\textsc{UPE})}
{A connected acyclic digraph $G$, a connected subdigraph $H$ of $G$, and an upward planar drawing $\Gamma(H)$ of $H$.}
{The vertex+edge deletion distance $k$ from $G$ to $H$, i.e., $|V(G)\setminus V(H)|+|\{E(G[V(H)])\setminus E(H)\}|$.}
{Does there exist an upward planar extension $\Gamma(G)$ of $\Gamma(H)$?}
We call a drawing $\Gamma(G)$ satisfying the above property a \emph{solution} (or \emph{extension}) for the given instance of \textsc{UPE}. 
The elements of $V_M = V(G) \setminus V(H)$ and $E_M = E(G) \setminus E(H)$ are called \emph{missing} vertices resp.\ edges.
The vertices and edges of $H$ are called \emph{pre-drawn}. Throughout the paper, we often capitalize missing vertices to distinguish them from pre-drawn.
We define the input size~$n$ of \textsc{UPE} as~$|G|$ plus the number of bends and vertices in $\Gamma(H)$. 
In line with past work on drawing extension problems~\cite{adfjkp-tppeg-15,EibenGHKN20, mfcs/EibenGHKN20,GanianHKPV21,HammH22,BhoreGKMN23}, we 
do not include the encoding length of the coordinates in the input size.

\subparagraph{Boundary Curves and Walks.}
In some of our later arguments, we will make use of curves and walks which ``follow'' the boundary of a face $f$.

Let $v\in V(H)$ be an arbitrary vertex on the boundary of a face $f$. 
We say that a \emph{boundary curve} $C(f)$ of $f$ is a closed curve which starts at a point $p$ which lies at a distance of~$\varepsilon$ from $v$ inside $f$ (for some choice of $\varepsilon$ that is sufficiently smaller than the lowest distance between any two points inside $\drawing$), and then follows the boundary of $f$ while maintaining a distance of~$\varepsilon$ from the boundary of $f$. We say that the boundary curve $C(f)$ \emph{visits} a vertex $w$ and/or edge $e$ incident with $f$ when it is at a distance of $\varepsilon$ from $w$ or $e$, respectively. 

Observe that $C(f)$ visits each edge on the boundary of $f$ either once or twice (precisely once for each side of that edge which lies in $f$), but may visit each vertex on the boundary of $f$ multiple times. More precisely:
\begin{itemize}
\item $C(f)$ visits each edge incident with one strip of $f$ once, and each edge incident with two strips of $f$ twice (once from each strip), and
\item For each vertex $v$ on the boundary of $f$ and each strip $s$ which touches $v$, $C(f)$ visits $v$ while being inside $s$ precisely once.
\end{itemize}

A second combinatorial structure we will need is the \emph{boundary walk} $W(f)$, which is the circular non-repeating sequence of vertices and edges that are visited when following $C(f)$ in a counterclockwise manner. We explicitly remark that 
$W(f)$ may visit each vertex and edge multiple times. An illustration of these two concepts is provided in \cref{fig:boundary-walk}.

\begin{figure}[ht]
\centering
\includegraphics[scale = 1.3]{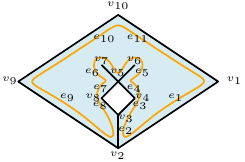}
\caption{The face $f$ is blue. The curve $C(f)$ is orange and the boundary walk $W(f)$ is given by $v_1$, $e_1$, $v_2$, $e_2$, $v_3$, $e_3$, $v_4$, $e_4$, $v_5$, $e_5$, $v_6$, $e_5$, $v_5$, $e_6$, $v_7$, $e_6$, $v_5$, $e_7$, $v_8$, $e_8$, $v_3$, $e_2$, $v_2$, $e_9$, $v_9$, $e_{10}$, $v_{10}$, $e_{11}$.}
\label{fig:boundary-walk}
\end{figure}

\section{Initial Setup}
\label{sec:setup}
The aim of this section is to ``peel off'' some of the simpler layers of \upe\ to facilitate our attack at the core difficulties of the problem. 
\subparagraph*{Step 1.}
Given an instance $\III'=(G',H',\Gamma(H'))$ of \upe, we begin by constructing the instance $\III_1=(G_1,H_1,\Gamma(H_1))$ as follows. $\Gamma(H_1)$ is  obtained by replacing each bend in $\Gamma(H')$ with a new vertex of degree $2$, and $G'$ and $H'$ are then adapted to $G_1$ and $H_1$ by accordingly applying the appropriate subdivisions to edges of~$E(H)$. 
Clearly, the above construction can be carried out in polynomial time and preserves the existence of a solution to the instance as well as the value of the parameter while ensuring that each edge of $H_1$ forms a straight-line segment in $\Gamma(H_1)$.

\subparagraph*{Step 2.} Next, we construct the instance $\III_2=(G_2,H_2,\Gamma(H_2))$ from $\III_1$ as follows. We subdivide each edge $ab\in E(G_1)\setminus E(H_1)$ such that $a,b\in V(H_1)$---i.e., each missing edge with two pre-drawn endpoints---to form two new edges $ac$ and $cb$, and add both of these edges as well as the vertex $c$ to $G_2$ but not to $H_2$. 
Clearly, this construction can be carried out in polynomial time and preserves the existence of a solution: every extension $\Gamma(G_1)$ can be transformed into an extension $\Gamma(G_2)$ by placing each new vertex $c$ anywhere on the upward polyline connecting $a$ to $b$, and vice-versa removing the new vertices from $\Gamma(G_2)$ yields an extension $\Gamma(G_1)$ for $\III_1$. Crucially, the parameter $k$ in $\III_1$ is precisely the same as in $\III_2$, as each choice of $ab$ in $\III_1$ contributes $1$ to $k$ and in $\III_2$ the whole set $\{ac,bc,c\}$ contributes to~$1$ to the parameter as well. We observe that at this point, the graph $H_2$ forms an induced subgraph of $G_2$.

\subparagraph*{Step 3.} We construct the instance $\III_3=(G_3,H_3,\Gamma(H_3))$ from $\III_2$ as follows. For each two distinct vertices  $a,b\in V(H_2)$ which have the same $y$-coordinate in $\Gamma(H_2)$, we test whether~$G_2$ contains a directed $a$-$b$ path and immediately reject if this is the case for any such pair. 
Otherwise, for each such $a,b\in V(H_2)$, we perturb the $y$-coordinate of $a$ by a sufficiently small value while preserving the upward planarity of the drawing. This results in a new drawing~$\Gamma(H_3)$ such that each pair of distinct vertices $a,b\in V(H)$ receive distinct $y$-coordinates, whereas the instance $\III_3$ is defined by setting $G_3=G_2$ and $H_3=H_2$.

 \begin{observation}
 $\III_3$ is a positive instance if and only if so is $\III_2$.
 \end{observation}

 \begin{proof}
Given a solution $\Gamma(G_2)$ for the latter instance, we can construct a solution for the former by merely slightly shifting the $y$-coordinates of the missing vertices in $V(G_2)\setminus V(H_2)$ and of the drawing of the missing edges to avoid being unnecessarily close to the vertices of $H_2$. Given a solution $\Gamma(G_3)$ for the former instance, we first observe that since we have already excluded the presence of directed paths between vertices that had the same $y$-coordinates in~$\Gamma(G_2)$, no vertex of $V(G_3)\setminus V(H_3)$ needs to be ``sandwiched'' between the $y$-coordinates of any two vertices which received the same $y$-coordinate in $\Gamma(G_2)$. Hence, we may once again  slightly shift the $y$-coordinates of the missing vertices in $V(G_3)\setminus V(H_3)$ and of the drawing of the missing edges to avoid being unnecessarily close to the vertices of $H_3$---thus allowing us to also construct a solution for the latter instance.
 \end{proof}

 \subparagraph*{Step 4.}
For the final---and only technically involved---step of this section, we apply exhaustive branching to reduce $\III_3$ to the task of solving the special case where we merely need to place the missing vertices in a single face of $\Gamma(H_3)$. Before proceeding, let us formally define the ``target'' problem we will focus on in the subsequent Sections~\ref{sec:strips}-\ref{sec:dynprog-new}: 

\vspace{0.3cm}
	\noindent\fbox{
		\begin{minipage}{0.96\linewidth}
			\begin{tabular*}{\linewidth}{@{\extracolsep{\fill}}lr} \textsc{Upward Planar Drawing Extension in a face $f$} (\textsc{UPEF}) & \\ \end{tabular*}
			{\bf{Input:}} A connected acyclic digraph $G$, a connected induced subdigraph $H$ of $G$, an upward straight-line planar drawing $\Gamma(H)$ of $H$ whose vertices each have a unique $y$-coordinate, 
            and an inner face $f$ of this drawing. \\
   {\bf{Parameter:}} $k=|V(G)\setminus V(H)|$\\ 
			{\bf{Question:}} Does there exist an extension $\Gamma(G)$ of $\Gamma(H)$ such that all vertices of $V(G)\setminus V(H)$ are placed in $f$?
		\end{minipage}
	} 
\vspace{0.3cm}

We call an extension $\Gamma(G)$ satisfying the conditions stipulated above a \emph{baseline solution} of the instance. For technical reasons, we will also allow baseline solutions to route missing edges multiple times, i.e., allow $\Gamma(G)$ to include multiple distinct (and pairwise non-intersecting) curves connecting the endpoints of each missing edge.
Let \textsc{UPEF-Outer} be defined analogously to \textsc{UPEF} (cf.\ Section~\ref{sec:intro}) but where all vertices of $V(G)\setminus V(H)$ must be placed in the \emph{outer face} of $\Gamma(H)$. 
An algorithm for \textsc{UPEF} will be obtained in Section~\ref{sec:dynprog-new} and \textsc{UPEF-Outer} will be handled in Section~\ref{sec:outer}. 
Here, we complete the transition from \upe\ to these two problems by proving:

\begin{lemma}\label{lem:time1}
Let $\mathbb{A}$, $\mathbb{B}$ be fixed-parameter algorithms solving \textsc{\textup{UPEF}} and \textsc{\textup{UPEF-Outer}} in time at most $\bigoh^*(f_1(k))$, $\bigoh^*(f_2(k))$, respectively. Then there is an algorithm which solves every instance $\III_3=(G_3,H_3,\Gamma(H_3))$ in time at most $\bigoh^*(k^k\cdot \max(f_1(k),f_2(k)))$.
\end{lemma}

 \begin{proof}
Recall that the missing part of the drawing consists of precisely $k$ vertices and their incident edges. We begin by exhaustively branching over all of the at most $k^k$ many possible partitions $\mathcal{P}$ of $V(G_3)\setminus V(H_3)$ such that $G_3$ does not contain edges between vertices from distinct $P_1, P_2 \in \mathcal{P}$.  For each such partition $\mathcal{P}$, each $P\in \mathcal{P}$ and each inner face $f$ of $\Gamma(H_3)$, we invoke $\mathbb{A}$ to determine whether $P$ can be placed in $f$; if yes, we say that the pair $\{P,f\}$ is \emph{compatible}. Moreover, for each such $P\in \mathcal{P}$ we invoke $\mathbb{B}$ to determine whether $P$ can be placed in the outer face $f'$ of $\Gamma(H_3)$, and similarly say that the pair $\{P,f'\}$ is compatible if the answer is positive. 

Once these steps are complete for a choice of $\mathcal{P}$, we construct an auxiliary instance of \textsc{Bipartite Matching} where one side of the bipartition contains all the faces of $\Gamma(H_3)$, the other side contains all the sets in the partition $\mathcal{P}$, and the edges are the compatible pairs. We then solve each such constructed instance of \textsc{Bipartite Matching} in polynomial time using one of the standard algorithms, and output ``Yes'' if and only if at least one of these instances contains a matching of size $|\mathcal{P}|$.

For correctness, assume that $\III_3$ admits some baseline solution $\Gamma(G_3)$. Then there exists a partition $\mathcal{P}$ of $V(G_3)\setminus V(H_3)$ into sets of vertices which are placed by $\Gamma(G_3)$ in the same face of $\Gamma(H_3)$. Moreover, there must also exist a matching between the faces $\Gamma(G_3)$ used to place these missing vertices and the sets in $\mathcal{P}$, and our algorithm would thus be guaranteed to output ``Yes''. On the other hand, if our algorithm outputs ``Yes'', then $\III_3$ must admit a baseline solution as we may partition the vertices of $V(G_3)\setminus V(H_3)$ according to the choice of $\mathcal{P}$ for which a matching was found, and refer to the correctness of $\mathbb{A}$ and $\mathbb{B}$ for a guarantee that there exists a baseline solution to $\III_3$ that can be obtained by placing each set in $\mathcal{P}$ in a distinct face of $\Gamma(H_3)$.
\end{proof}

Putting the four steps together, we conclude:

\begin{lemma}
\label{lem:setup}
Let $\mathbb{A}$, $\mathbb{B}$ be fixed-parameter algorithms solving \textsc{\textup{UPEF}} and \textsc{\textup{UPEF-Outer}} in time at most $\bigoh^*(f_1(k))$ and $\bigoh^*(f_2(k))$, respectively. Then \textsc{\textup{Upward Planar Drawing Extension}} can be solved in time at most $\bigoh^*(k^k\cdot \max(f_1(k),f_2(k)))$.
\end{lemma}

Given Lemma~\ref{lem:setup}, in the next three sections we focus on solving a given instance $\III=(G,H,\Gamma(H),f)$ of \upefshort.
The sets $V_M$ and $E_M$ denote the sets $V(G) \setminus V(H)$ and $E(G) \setminus E(H)$ of \emph{missing} vertices and edges, respectively.
We partition $E_M$ into the set $E_M^1$ of missing edges incident with precisely one vertex of $V_M$ and the set $E_M^2\subseteq (V_M)^2$ with both endpoints in missing vertices.
Conversely, we call the vertices and edges of $H$ \emph{pre-drawn}.

To improve the clarity of our exposition we will conventionally use uppercase letters (e.g., $A$, $B$ etc.) for missing vertices and lowercase letters (e.g., $u$, $v$ etc.) for pre-drawn vertices.

\section{Gates, Strips and Trees}
\label{sec:strips}
In this section, we formalize the combinatorial representation of a face we will employ for the pre-drawn part of $\III$ and use it to introduce a further round of simplifications for \upefshort.

\subparagraph{Gates and Strips.} 
Let $f$ be an inner face of the drawing $\drawing$ and let $v$ be a vertex on its boundary such that a horizontal ray leaving the vertex $v$ to its left enters the face $f$ first.
Let~$p$ be the first intersection of the ray with the boundary of $f$ other than $v$ (it exists since $f$ is an inner face). 
Since the drawing $\Gamma(H)$ is in general position with respect to $y$-coordinates (i.e., no two vertices are placed on the same $y$-coordinate), $p$ is a point on some straight-line segment and, in particular, no vertex is drawn at $p$.

Note that since an upward planar drawing does not contain horizontal segments, the point $p$ is well-defined.
We call the straight-line segment between $v$ and $p$ the \emph{left gate} at $v$.
The \emph{right gates} are defined analogously.
A \emph{gate} of an inner face $f$ is a left or right gate at some vertex $v$ such that this gate is contained in the face $f$ (see \cref{fig:strip-graph}~(a)).
The general position of vertices in $\Gamma(H)$ ensures that for every gate $g$, there is exactly one vertex of $H$ that is an endpoint of $g$.
Note that the internal points of a gate are contained inside the face $f$ and, in particular, no edge of $H$ crosses the gate in its non-endpoints and no vertex of~$H$ lies on the gate, apart from exactly one of its endpoints.
The gates of an inner face $f$ partition it into connected regions called \emph{strips}~(see \cref{fig:strip-graph}~(a)); in particular, the boundary of each strip consists of a set of its gates and line segments and points of $\Gamma(H)$. 
Each gate $g$ occurs in precisely two strips which \emph{share} that gate: one lies above the gate ($g$ is one of its \emph{bottom-gates}) and the other below it ($g$ is one of its \emph{top-gates}).
\subparagraph{Strip Trees.} 

Let $T_f$ be the graph whose vertices are the strips of $f$ and where there is an edge between two strips if and only if they share a gate (see \cref{fig:strip-graph}~(b)).
Then for two strips adjacent in $T_f$ one of them lies above their common gate and the other lies below it.
Thus for any pair of adjacent vertices in $T_f$, one of them lies below the other one.
In view of the following 
theorem, 
we will refer to $T_f$ as the \emph{strip tree} (of $f$). 
\begin{proposition}
\label{pro:striptree}
    The graph $T_f$ of an inner face $f$ of $\Gamma(H)$ is a tree 
    that has $\mathcal{O}(n)$ nodes and in which every node has at most two neighbors lying below it and at most two neighbors above it.
    Thus, $T_f$ has maximum degree at most $4$.
\end{proposition}
\begin{proof}
\emph{For convenience, we provide a direct and self-contained proof of all claims in the theorem, as it serves as a starting point for many of our future considerations (but does not constitute a major technical contribution on its own). It is worth noting that some of the claimed properties can also be ascertained by leveraging known facts about vertical decompositions and the sweep-line algorithm~\cite{BergCKO08}.}

    First of all, note that the general position with respect to the $y$-coordinates implies that exactly two vertices lie on the boundary of any strip $s$ of $f$ and each of these vertices yields at most two gates for $s$.
    Therefore, $s$ has at most two adjacent strips above it and at most two adjacent strips below it as claimed.
    Furthermore, every vertex lies on the gate of at most two strips so the number of strips is bounded by $\mathcal{O}(n)$.

    Next, we show that $T_f$ is connected. Let $A \neq B$ be two distinct strips in $f$ and let~$a \in A$ and $b \in B$ be two points in these strips.
    Since $f$ is a connected region, there exists a polyline~$\phi$ in $f$ connecting $a$ to $b$.
    By slightly perturbing the bends of $\phi$, we may assume that no segment is horizontal.
    This polyline naturally defines a sequence $\psi$ of strips visited by $\phi$ starting with $A$ and ending with $B$.
    Observe that by definition, any point in the intersection of a pair of strips (if such a point exists) belongs to their common gate. 
    Therefore, if two strips $s_1$ and $s_2$ are consecutive in $\psi$ (i.e., $\phi$ visits them consecutively), the strips $s_1$ and $s_2$ share a gate and they are adjacent in $T_f$.
    Hence $\psi$ is a walk from $A$ and $B$ in $T_f$, establishing connectivity.

    For the final step, suppose for a contradiction that $T_f$ contains a cycle $C$. Intuitively, we will now construct a plane drawing (of some graph) that visualizes the face $f$ along with its gates.    
    
    Let $\overline{H_f}$ denote the underlying undirected subgraph of $H$ which consists of all vertices and edges on the boundary of $f$, and $\Gamma(\overline{H_f})$ the restriction of $\Gamma(H)$ to the vertices and edges of $\overline{H_f}$. 
    For the purposes of our proof, we further alter $\overline{H_f}$ and its corresponding restriction of the drawing $\Gamma(H)$ as follows: for each gate $g$ which connects a vertex $a\in V(\overline{H_f})$ to a point $p$ on the straight-line segment representing the edge $e$, we (1) subdivide $e$ by creating a new vertex $v_p$ and placing it precisely at the point $p$, and (2) connecting $p$ to $a$ with an edge. Let $H'$ denote the resulting graph, and let $\Gamma(H')$ be the drawing of $H'$ obtained by replacing each gate with a straight-line segment representing the edge arising from that gate. Observe that $\Gamma(H')$ is a plane drawing of $H'$ where gates are replaced by \emph{gate-edges} which form ``solid'' line segments.

Next, let us consider the dual graph $D$ of~$\Gamma(H')$, i.e., $D$ contains a vertex for each face of $\Gamma(H')$ and where two faces are adjacent if and only if they share an edge on their boundary. Observe that every strip of $T_f$ is a vertex in $D$, and in particular $C$ is also a cycle in $D$ whose vertices are strips and edges connect pairs of strips which share a gate-edge~$D$. By the cut-cycle duality, there exists a curve $\gamma$ crossing only the gate-edges of~$D$ and a partition~$(X, Y)$ of $V(H')$ such that $X$ and $Y$ are non-empty sets separated by $\gamma$, i.e., every path from~$X$ to $Y$ in $H'$ uses a gate-edge crossed by $\gamma$. 

 To complete the proof, notice that there must exist vertices $x\in V(H)\cap X$ and $y\in V(H)\cap Y$; indeed, by construction of $H'$, every vertex in $X$ (or $Y$) is either simply a vertex of $H$ or a point on a straight-line segment $\zeta$ of $\Gamma(H)$ which is not crossed by $\gamma$, and hence $\zeta$ connects two vertices of $H$ which must both lie in $X$ ($Y$). Moreover, inside $\Gamma(H)$ the curve~$\gamma$ exists completely inside the face $f$. Hence, the existence of a cycle $C$ in $T_f$ contradicts the connectivity of $H$. This implies that $T_f$ is acyclic and connected, and hence forms a tree. 
 \end{proof}

\subparagraph{Simplifications and Streamlining.}
Let us call a strip \emph{non-polar} if it has precisely one top-gate, precisely one bottom-gate and there is no vertex adjacent to $V_M$ on its boundary; otherwise, we call it \emph{polar}. 
For our following considerations, it will be useful to ensure that 
 \begin{enumerate}
 \item if a baseline  solution exists, then there also is a baseline solution which (a) places missing vertices only in non-polar strips and (b) places at most one vertex in each strip;
 \item $T_f$ has maximum degree $3$,
 \item for any polar strip $s$ of degree 3 in $T_f$, there exists at most on vertex adjacent to $V_M$ on the boundary of $s$. And if such a vertex exists, then it lies on two gates of $s$.
 \end{enumerate}
For this, we will perform a few straightforward modifications of the instance (and hence of the strip tree $T_f$). 
Begin by marking each node of $T_f$ as \emph{unprocessed}; we will first exhaustively process leaf nodes and then proceed to the non-leaf nodes. 

Let $s$ be an unprocessed leaf in $T_f$ with unique neighbor $p$; for the following we assume that $p$ lies above $s$, as the converse case is entirely symmetric, and refer to \cref{fig:leaf-strip} for an illustration of the forthcoming step.
We choose a non-horizontal straight-line segment $\phi$ of the boundary of $s$ (i.e., a non-gate straight-line segment which therefore belongs to $\Gamma(H)$ and represents some edge $e$). We then subdivide $e$ by adding a new pre-drawn vertex, say~$v_s$, and add $v_s$ into $\Gamma(H)$ while ensuring that its $y$-coordinate is unique (i.e., preserving our general position property). 
 Finally, we further adapt the instance by adding a new vertex $w_s$ adjacent only to $v_s$ and alter the drawing $\Gamma$ in such a way that the new edge $(v_s, w_s)$ is drawn as a ``tiny'' (say, of length $\varepsilon_s$ for sufficiently small $\varepsilon_s > 0$) 
 non-horizontal straight-line line segment which starts at a point inside $s$ and below $v_s$ and which is fully contained in~$s$---the placement of $w_s$ is once again chosen so as to preserve $y$-wise general position.

\begin{figure}[ht]
\centering
\includegraphics{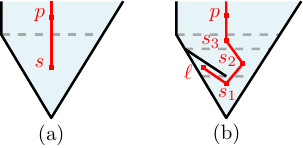}
\caption{Ensuring that every leaf strip has no vertex incident with a missing edge on its boundary. (a) A leaf $s$ in a strip tree (illustrated in red). (b) The leaf $s$ is partitioned into four new strips so that no missing vertex needs to be placed in the new leaf $\ell$.}
\label{fig:leaf-strip}
\end{figure}

 This way, the strip $s$ is split into the three strips $\ell$, $s_1$, $s_2$, and $s_3$ and the corresponding edge $s p$ is replaced by the four edges $s_1 \ell$, $s_1 s_2$, $s_2 s_3$, and $s_3 p$ (see \cref{fig:leaf-strip}~(b)).
 Among these four strips, in the arising strip tree, only the strip $\ell$ is a leaf and by construction, its boundary does not contain a vertex incident with a missing edge. We observe that this change does not influence the existence of a baseline solution for the instance, and crucially, we may assume that the sought-after baseline solution does not place any missing vertex or any missing edge inside $\ell$.
 Indeed, on the one hand, if the new instance admits such a baseline solution, then the original clearly does as well.
 On the other hand, if the original instance admits a baseline solution, then it also admits one where all missing  vertices and edges are placed at distance at least $3 \cdot \varepsilon_s$ from $v_s$ and in particular, no vertex or edge is drawn inside $\ell$ or crosses $e_s$; this yields a baseline solution of the new instance where no missing edge or vertex is placed inside $\ell$.
 
  We mark the leaf $\ell$ as processed and $s_1$, $s_2$ and $s_3$ as unprocessed. Observe that exhaustively applying the procedure outlined above increases the size of the instance $\mathcal{I}$ at most by a constant factor of $4$, the parameter does not increase, and that the procedure results in a new instance whose strip tree $T_f$ has every leaf marked as processed. 

Afterwards, let us consider an unprocessed non-leaf node~$p$ of $T_f$ such that the top-gate of~$p$ has the lowest $y$-coordinate among all unprocessed non-leaf nodes (resolving ties arbitrarily). We once again choose a non-horizontal straight-line segment $\phi$ of the boundary of $p$ where $\phi$ corresponds to the edge $e$ in $H$, subdivide $e$ $k+1$ many times, and adapt $\Gamma(H)$ by placing the new vertices on the original drawing of $\phi$ while maintaining $y$-wise general position (i.e., avoiding the use of $y$-coordinates which are already used by other vertices). We mark all the strips in the new drawing that are contained in $p$ as processed. Moreover, if $\phi$ occurs on the boundary of another unprocessed node~$p'$, we note that such $p'$ must be unique (by the choice of $p$) and mark all the strips in the new drawing that are contained in~$p'$ as processed as well.

Observe that this subdivision partitions the strip $p$ (and $p'$ if it exists) into precisely $k+2$ strips and each of these strip has the degree of at most $3$ in the arising strip tree.
Further, at most two of these strips (namely, the topmost and the bottommost) can be polar, the remaining $k$ strips are non-polar.
Finally, if some of two possibly non-polar strips contained in $p$ has the degree $3$ and it has a vertex adjacent to $V_M$ on its boundary, then this vertex is precisely the vertex that lies on two gates of this strip (i.e., on the side of $s$ with two neighbors)---this is because the new subdivision vertices are not adjacent to $V_M$.

Also observe that the exhaustive application of the above procedure will terminate after at most linearly many iterations: each iteration replaces either $1$ or $2$ unprocessed nodes in~$T_f$ with $k+2$ processed nodes each, without altering any other part of $T_f$. The carried out operation (i.e., subdividing an existing edge of $H$) obviously cannot influence the existence of a baseline solution for the instance. Crucially, given any hypothetical baseline solution~$\Gamma(G)$ for the original instance such that $\Gamma(G)$ places multiple missing vertices $x_1,\dots, x_i$ in $p$ (resp.~$p'$), we may construct a new baseline solution $\Gamma(G)'$ which coincides with $\Gamma(G)$ outside of $p$ (resp.~$p'$) and where the missing vertices $x_1,\dots, x_i$ are each placed in a distinct processed non-polar strip inside $p$ ($p'$). Indeed, assume $x_1,\dots,x_i$ are ordered in an ascending fashion based on their $y$-coordinates; since the strips are convex trapezoidal regions, one may apply a direct redrawing argument to place each $x_j$, $j\in [i]$ into the $j$-th lowest new non-polar strip inside $p$ (resp.\ $p'$). An illustration is provided in Figure~\ref{fig:nonleaf-strip}.

We summarize the outcome of our streamlining and simplification steps below:

\begin{figure}[ht]
\centering
\includegraphics{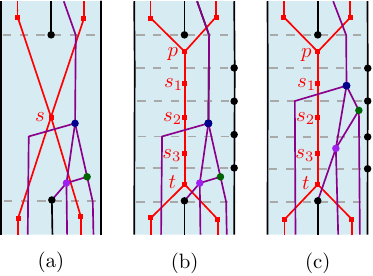}
\caption{The transformation for non-leaf strips. (a) A strip $s$ along with its neighbors in the strip tree (in red) and a hypothetical solution (in purple, blue, and green). (b) The strips resulting from~$s$ after applying our transformation. (c) Adapting the original solution to a new one where missing vertices are placed in distinct non-polar strips.}
\label{fig:nonleaf-strip}
\end{figure}

   \begin{lemma}
   \label{lem:streamlining}
There is a polynomial-time parameter-preserving reduction which takes as input an instance $\III'$ of \upefshort\ and outputs an instance $\III$ of \upefshort\ such that $|\III|\leq k\cdot |\III'|$ and if $\III'$ is a YES-instance, then $\III$ admits a baseline solution which only places missing vertices in pairwise distinct non-polar strips. Also, the strip tree of $\III$ contains $\bigoh(kn)$ vertices, none of them is of degree greater than $3$, every polar strip $s$ has at most one vertex $v$
adjacent to $V_M$ on its boundary, and if $s$ has three gates, then $v$ is on the side of $s$ with the two gates. Furthermore, a solution for $\III'$ can be constructed from a solution for $\III$ in linear time.
  \end{lemma}

Given Lemma~\ref{lem:streamlining}, we will assume that the strip tree $T_f$ is subcubic, and will call baseline solutions which place missing vertices in distinct non-polar strips \emph{solutions}.
Finally, we choose an arbitrary leaf-node $r'$ as the root of $T_f$, and let $r$ denote its unique neighbor.
We root $T_f$ at some leaf-node $r'$.
\section{Dynamic Programming Setup}\label{sec:properties}
\label{sec:dynprogsetup}

Our aim is to perform a leaf-to-root dynamic programming procedure along the strip tree.
As we already sketched in the Introduction, there are two main obstacles that prevent us from directly providing such a procedure at this point.
Both of these stem from the fact that the number of missing edges is not bounded by a function of the number of missing vertices (i.e., our parameter) while we need to be able to store only an amount of information about a partial solution that is bounded by a function of the number of missing vertices.
For a strip $s$, we use $f^s$ to denote the union of strips in the subtree of $s$.

First, to be able to later check if some \emph{partial solutions} of the children of a node $s$ in the strip tree can be combined to form a larger partial solution at $s$, we need to keep certain information about the ordering in which the missing edges cross the gate of a strip. 
For this we show that we can group the edges into so-called \emph{threads} of the same \emph{type} so that the edges of the same thread cross the gate consecutively.
This allows us to define a so-called \emph{gate word} of a partial solution and show that it has bounded length.

\begin{figure}[t]
\centering
\includegraphics{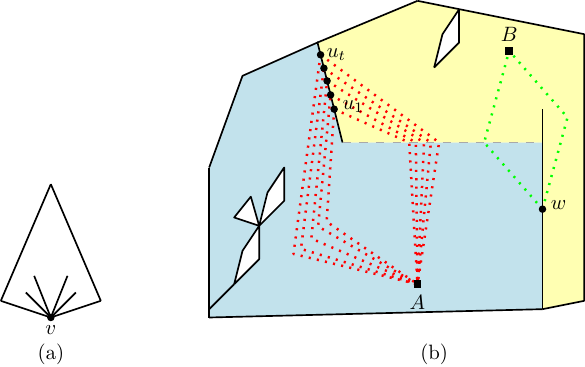}
\caption{(a) A vertex $v$ having a large number of incident strips. (b) In blue we have the region $f^s$, in yellow the region $f \setminus f^s$. The union of the two is the face $f$, note that it can contain white gaps corresponding to further faces of $\Gamma(H)$. The gate between the two regions is gray. Each of the edges $A u_1, \dots, A u_t$ and $w B$ can be routed inside $f^s$ or inside $f \setminus f^s$ (dotted lines) leading to $2^{\Omega(t)}$ partial solutions. The choice if the edge is routed or not is possible for missing edges incident with vertices already placed in $f^s$ (e.g., $A$) as well as for vertices not placed in $f^s$ (e.g., $B$).}
\label{fig:cut-vertices}
\end{figure}

The second obstacle is formed by cut-vertices of $G$.
Such vertices can have up to $\mathcal{O}(n)$ incident strips, and therefore in a solution, a missing edge incident with such a vertex can start in one of a large number of strips (see \cref{fig:cut-vertices}~(a)).
For this reason, we additionally need to keep track of the set of missing edges that have been drawn by the partial solution.
This approach would yield an algorithm with running time $2^{\Omega(n)}$ (see \cref{fig:cut-vertices}~(b) for an example with many partial solutions).
To overcome this issue, we carry out a careful analysis of the vertices on the boundary of the subtree processed so far.
We say that an edge from $E_M^1$ (i.e., a missing edge with precisely one pre-drawn endpoint) with pre-drawn endpoint $v\in V(H)$ is routed through $f^s$ by a solution $S$ if the curve representing this edge leaves $v$ inside $f^s$.
We then show that for some vertices on the boundary of $f^s$, in every solution either all edges incident with such vertices are ``routed" through $f^s$ or none of them is (see, e.g., \cref{lem:enclosed-vertices-use-same-side}).
We also show that it suffices to restrict ourselves to so-called \emph{clean} solutions which allow some other vertices on the boundary of $f^s$ to be ``bundled'' together so that, again, either all of the edges incident with the same bundle are routed through $f^s$ or none of them is (see e.g., \cref{lem:clean-solutions}).
The next subsections carry out this analysis and result in \cref{thm:relevant-records}, which provides a bound on the number of partial solutions we keep track of at any strip $s$ by a function of the parameter $k$.

\subsection{One-Sided Regions and Bundles}
\label{sub:onesided}
The aim of this subsection is to formalize the notion of bundles and show that edges incident with bundled vertices can be considered as a ``single block'' in our dynamic program (as formalized in the technical \cref{lem:clean-solutions}). But first, we will need to introduce a bit of specialized terminology.

Let $T'$ be a connected subtree of $T_f$ that can be obtained from $T_f$ by removing some edge $e'$ and picking one of the two arising connected components.
We call the union $f'$ of strips in $T'$ a \emph{one-sided region}.
Observe that $f'$ is a one-sided region if and only if there exists a strip $s$ such that $f' = f^s$ or $f' = f \setminus f^s$ (in \cref{fig:cut-vertices}, the blue and the yellow regions exemplify one-sided regions).
Let $g'$ denote the gate shared by the end-strips of $e'$; we call this gate the \emph{gate} of $f'$.
Observe that the gate $g'$ splits the face $f$ into two connected regions, one of which is $f'$.
Since $g'$ is a horizontal segment, one of these two regions lies below and the other above $g'$; For this reason, in the following we will speak about the cases where $f'$ lies below or above this gate.

\begin{figure}[!h]
\centering
\includegraphics[scale=0.9]{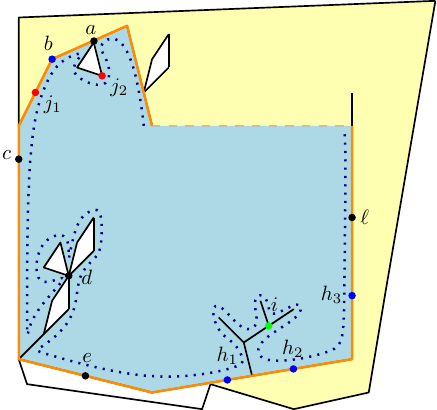}
\caption{The region $f'$ is depicted in blue and the region $f \setminus f'$ in yellow. The union of the two is the face $f$; further faces are white. The gate between the two regions is gray. 
The dotted navy curve is $C'(f')$. 
The outer boundary of $f'$ is orange.
The vertices $a,b,j_1,c,e,h_1,h_2,h_3,\ell$ are external to~$f'$, the vertices $j_2, c, d,e,i$ are internal to $f'$.
Red resp.\ blue resp.\ green vertices are adjacent to a missing vertex $A$ resp.\ $B$ resp.\ $C$.
The vertex $a$ has an $A$-attachment in $f'$ (due to~$j_2$).
The vertex $b$ is $A$-enclosed in $f'$.
The vertex $j_1$ is not interesting while $h_1$ is interesting with respect to $f'$.
The vertices $h_2$ and $h_3$ belong to the same $B$-bundle with respect to $f'$, while the vertices $h_1$ and $h_2$ do not (due to $i$). 
}
\label{fig:one-sided-regions}
\end{figure}

Let us emphasize that the following definitions of vertex categories depend on the one-sided region $f'$: for example, for two different one-sided regions~$f'$ and $f''$ (each being a subset of the face $f$), a pre-drawn vertex may be ``internal'' (will be defined in a moment) to~$f'$ and not internal $f''$; the same applies to enclosedness, bundles etc.
Also, the following definitions are independent of (partial) solutions, they only depend on the pre-drawn graph~$H$ and its drawing $\Gamma(H)$.

The \emph{outer boundary} of $f'$ is a closed curve consisting of the unique path connecting the endpoints of the gate of $f'$---in particular, the path starting from the vertex endpoint of that gate and ending at the edge containing the endpoint of that same gate---together with the gate of $f'$.
We emphasize that in general the interior of the outer boundary of $f'$ contains holes corresponding to some further faces of $\Gamma(H)$ (see \cref{fig:one-sided-regions}).

Let $C'=C(f')$ be the subcurve of the boundary curve $C(f)$ following the boundary of the region $f'$ (see \cref{fig:one-sided-regions}) and let $W'=W'(f')$ be the subwalk of $W(f)$ restricted to the subcurve $C'(f')$.
We say that the \emph{boundary} of $f'$ consists of all vertices and edges that occur in $W'$.
Let $v^*(f') \in V(H)$ be the unique vertex that lies on the gate of $f'$.
We say that a pre-drawn vertex $v \in V(H)$ with $v \neq v^*(f')$ on the boundary of $f'$ is \emph{internal} to $f'$ if $f'$ contains all strips incident with $v$. 
We say that a vertex $v \in V(H)$ on the boundary of $f'$ is \emph{external} to~$f'$ if it occurs on the outer boundary of $f'$---we use $\outerv(f')$ to denote the set of all vertices external to $f'$.
Also note that a vertex can be both internal and external to $f'$ at the same time (for example the vertices $c$ and $e$ in \cref{fig:one-sided-regions}).
Since all vertices external to $f'$ lie on the outer boundary of $f'$, there is a well-defined order in which the external vertices occur along $C'$.

Let $v_1(f'), \dots, v_q(f')$ be the vertices  of $\outerv(f')$ in the order they occur along the outer boundary of $f'$.
Note that this is the same ordering in which these vertices occur in $W'$: There we first see all occurrences of $v_1$ in $W'$, after that all occurrences of $v_2$, then all occurrences of $v_3$ etc.\ possibly interleaved by occurrences of further vertices on the boundary of $f'$.
For~$B \in V_M$, we call an external vertex $v_i \in \outerv(f')$ \emph{$B$-enclosed} if there exist vertices~$u$ and $w$ on the boundary of $f'$ with the following properties.
First, both $u$ and $w$ are adjacent to $B$.
Second, in $W'$ we first see all occurrences of $u$ in $W'$, then all occurrences of $v_i$ in~$W'$, and then all occurrences of $w$ in $W'$.
Observe that in particular, $v_i$ is $B$-enclosed if there exist~$i_1$ and $i_2$ such that $1 \leq i_1 < i < i_2 \leq q$ and $v_{i_1}$ and $v_{i_2}$ are adjacent to $B$.
For~$A \neq B \in V_M$, we use~$\enclosed_A^B(f')$ to denote the set of all external vertices adjacent to $A$ that are~$B$-enclosed.

Consider a solution $S$ of an instance $\mathcal{I}$ and let $X_S$ denote the set of vertices placed inside~$f'$ by $S$.
Let $v$ be a vertex on the boundary of $f'$ and let $A \in V_M$ be a missing vertex adjacent to $v$.
We say that the edge $\{v, A\}$ is \emph{routed inside} $f'$ by $S$ if there exists a curve representing~$\{v, A\}$ that leaves the vertex $v$ in $f'$.
Observe that if $A \in X_S$, then the edge~$\{v, A\}$ is routed inside $f'$ if and only if it does not cross the gate.
Similarly if $A \notin X_S$, then the edge $\{v, A\}$ is routed inside $f'$ if and only if it crosses the gate.
This is because no upward curve crosses the gate $g'$ more than once and this gate separates the one-sided regions $f'$ and $f \setminus f'$ from each other.

With the above terminology, we can formulate and prove our first statement that restricts how missing edges are routed inside $f'$ (see also Figure~\ref{fig:enclosed-vertices-same-side} for the intuition behind the proof).

\begin{lemma}\label{lem:enclosed-vertices-use-same-side}
    Let $i < j \in [q]$ be such that $v_i, v_j \in \enclosed^B_A(f')$ are two vertices on the outer boundary of $f'$ that are both adjacent to some $A \in V_M$ and both $B$-enclosed for some~$B \in V_M$ with $B \neq A$.
    Then in every solution $S$ of $\mathcal{I}$ either both $\{v_i, A\}$ and $\{v_j, A\}$ are routed~inside $f'$ (and not routed inside $f \setminus f'$) or both are routed inside~$f\setminus f'$ (and not routed inside $f'$).
\end{lemma}

\begin{proof}
    Let $u_i$ and $w_i$ be two vertices adjacent to $B$ such that in $W'$ we first see all occurrences of $u_i$ in $W'$, then all occurrences of $v_i$ in $W'$, and then all occurrences of $w_i$ in $W'$.
    Similarly, let $u_j$ and $w_j$ be two vertices adjacent to $B$ such that along $W'$ we first see all occurrences of~$u_j$ in $W'$, then all occurrences of $v_j$ in $W'$, and then all occurrences of $w_j$ in $W'$.
    Let~$u = u_i$ and let $w = w_j$. 
    The vertices $v_i$ and $v_j$ belong to the outer boundary of $f'$ and we have $i < j$.
    So along $W'$ we first see all occurrences of $u$ in~$W'$, then all occurrences of $v_i$ in $W'$, then all occurrences of $v_j$ in $W'$, and then all occurrences of $w$ in $W'$.

    Let $\alpha$ denote the curve that together with the gate $g'$ forms the outer boundary of~$f'$.
    Let~$P'$ denote the path in $H$ that consists of all vertices and edges lying on the outer boundary of $f'$.
    Note that $\alpha$ consists of the drawing of $P'$ in $\Gamma(H)$ without the straight-line segment of the unique edge in $H$ that corresponds to the gate $g'$.
    Recall that all vertices external to $f'$ belong to $P'$.
    
    We consider the subgraph $H_{f'}$ of $H$ consisting of all vertices and edges on the boundary of $f'$, i.e., all vertices and edges that occur in $W'$.
    Let $v'$ be a vertex on the boundary of $f'$ that does not belong to $P'$.
    Since the graph $H_{f'}$ is connected, there is a path $R_{v'}$ from $v'$ to a vertex on $P'$.
    We may assume that all internal vertices of $R_{v'}$ do not belong to $P'$ by taking subpath of $R_{v'}$ that ends in the first vertex on $P'$ otherwise.
    And we use $P'(v')$ to denote the endpoint of $R_{v'}$ other than $v'$.
    Then all occurrences of $v'$ in $W'$ take place between two occurrences of the vertex $P'(v')$.
    For a vertex $v' \in V(P')$, we use $R_{v'}$ to simply denote the zero-length path in $H_{f'}$ that consists of $v'$ only and by $P'(v')$ the vertex $v'$ itself.

    \begin{figure}[t]
        \centering
        \includegraphics{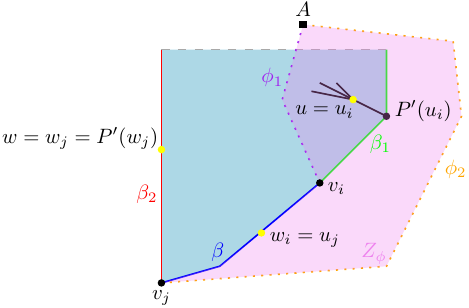}
        \caption{Proof of \cref{lem:enclosed-vertices-use-same-side}. The vertices $v_i$ and $v_j$ are $B$-enclosed, the neighbors of $B$ are depicted yellow. The edge $\{v_i, A\}$ is routed inside $f^s$ (in blue) while $\{v_j, A\}$ is routed outside. The neighbors $u_i$ and $w_j$ of $B$ become separated, i.e., the curve $u_i, B, w_j$ would necessarily cross one of the edges.}
        \label{fig:enclosed-vertices-same-side}
    \end{figure}
    
    Now suppose that in $S$, without loss of generality, the edge $\{v_i, A\}$ is routed inside $f'$ while the edge $\{v_j, A\}$ is routed inside $f \setminus f'$.
    And let $\phi_1$ resp.\ $\phi_2$ be the curves representing the edges $e_1$ and $e_2$ in $S$.
    We recall that $v_i$ and $v_j$ are vertices external to $f'$ so none of them lies on the gate $g'$, and therefore the endpoints of $\alpha$ are not equal to any of $v_i$ and $v_j$.
    By definition of the outer boundary of $f'$, the region $f'$ lies inside the curve defined by $\alpha$ together with the gate of $f'$.
    Let us recall that the region inside this curve, in general, contains holes that do not belong to $f'$ but these holes are disconnected from $f'$, and therefore do not belong to $f$.
    Therefore, since $S$ is the solution of \upef{} with the face $f$, no edge incident with $V_M$ is drawn inside any of these holes.
    
    Recall that $\alpha$ belongs to the outer boundary of $f'$, in the solution $S$, the edge $\{v_i, A\}$ is routed inside $f'$ while the edge $\{v_j, A\}$ is routed inside $f \setminus f'$.
    Therefore, the curves~$\phi_1$ and~$\phi_2$ leave the vertices $v_i$ and $v_j$, respectively, on different sides of the curve $\alpha$.
    Let~$\phi'$ be the closed curve that starts in $v_i$, follows $\phi_1$ until it reaches $A$, and then follows $\phi_2$ until it reaches~$v_j$.
    This curve does not have any common points with $\Gamma(H)$ other than its endpoints~$v_i$ and~$v_j$.
    Let $\beta$ be the subcurve of $\alpha$ between the vertices $v_i$ and~$v_j$.
    And let $\beta_1$ resp.\ $\beta_2$ be the subcurves of $\alpha$ having $v_i$ resp.\ $v_j$ as its endpoint such that $\alpha$ is the concatenation of~$\beta_1$, $\beta$, and $\beta_2$.
    And let $\phi$ be the self-nonintersecting closed curve given by the concatenation of $\phi'$ and $\beta$.
    The curve $\phi$ defines a closed region $Z_\phi$.
    The closed region $Z_\phi$ then either contains the subcurve $\beta_1 \setminus \{v_i\}$ inside and the curve $\beta_2 \setminus \{v_j\}$ lies outside $Z_{\phi}$, or it is the other way around.

    Without loss of generality assume that the former occurs; the other case is symmetric.
    Now recall that we have two vertices $u$ and $w$ lying on the boundary of $f'$ both adjacent to~$B$ such that along~$W'$ we first see all occurrences of $u$ in $W'$, then all occurrences of $v_i$ in~$W'$, then all occurrences of $v_j$ in $W'$, and then all occurrences of $w$ in $W'$.
    Then along $P'$ the vertices $P'(u), v_i$, $v_j$, and then $P'(w)$ occur in this or in the reversed ordering.
    Therefore, the curve $\beta_1 \setminus \{v_i\}$ contains the vertex $P'(u)$ and since the path $R_u$ is disjoint from $P'$ apart from the vertex $P'(u)$, the whole path $R_u$ lies inside $Z_\phi$.
    In particular, the vertex $u$ lies inside $Z_\phi$.
    Similarly, the curve $\beta_2 \setminus \{v_j\}$ contains the vertex $P'(w)$ and since the path $R_w$ is disjoint from $P'$ apart from the vertex $P'(w)$, the whole path $R_w$ lies outside $Z_\phi$.
    
    Finally, consider the curves $\xi_u$ and $\xi_w$ corresponding to the edges $\{u, B\}$ and $\{w, B\}$, respectively, in $S$.
    And let $\xi$ be their concatenation.
    The endpoint $u$ of this curve lies inside the closed region $Z_\phi$ while the other endpoint $w$ lies outside this region.
    Therefore, the curve~$\xi$ intersects the boundary of $Z_\phi$.
    But the boundary of $Z_\phi$ consists of the curves from~$\Gamma(H)$ and $S$ contradicting the planarity of $S$.
\end{proof}
Let us remark that some vertex, say $v$, in $\enclosed^B_A(f')$ may be internal to $f'$, i.e., all of its incident strips belong to $f'$: In this case, the edge $\{v, A\}$ is routed inside $f'$ in any solution, and therefore all edges $\{u, A\}$ are routed inside $f'$ (and not inside $f \setminus f'$) for any $u \in \enclosed^B_A(f')$ in any solution.

Next, we make an observation about vertices in $H$ incident with at least two missing edges.

\begin{observation}\label{obs:large-degree-vertices}
    Let $A \neq B$ and
    let $i_1 < \dots < i_{q_Z}$ be such that $\{i \mid \{A, B\} \subseteq N(v_i) \cap V_M\} = \{i_1, \dots, i_{q_Z}\}$. 
    Then every vertex in $v_{i_2}, \dots, v_{i_{q_Z - 1}}$ is $A$-enclosed and $B$-enclosed.
\end{observation}   

\begin{proof}
    For every $j \in \{2, \dots, q_Z - 1\}$, all occurrences of $v_{i_j}$ in $W'$ take place after all occurrences of $v_{i_1}$ of $W'$ and before all occurrences of $v_{i_{q_Z}}$ in $W'$, and the vertices $v_{i_1}$ and~$v_{i_{q_Z}}$ are adjacent to both $A$ and $B$.
\end{proof}

The next definition and lemma will later allow us to argue that a walk along the boundary of a face of $\Gamma(H)$ cannot visit neighbors of missing vertices in an interleaving fashion.

\begin{definition}
    Let $i \in [q]$ and let $B \in V_M$. 
    We say that the vertex $v_i$ has a $B$-\emph{attachment} if there exists a vertex $w$ adjacent to $B$ such that $w$ occurs between two occurrences of $v_i$ in~$W'$; in this case, we also say that $w$ is attached to $v_i$ in $f'$.
    For $A \neq B \in V_M$, we use $\attachment^B_A(f')$ to denote the set of all vertices external to $f'$ that are adjacent to $A$, have a $B$-attachment, and are not $B$-enclosed.
\end{definition}

\begin{lemma}\label{obs:vertices-with-attachments}
    Let $A \neq B \in V_M$.
    Then $|\attachment^B_A(f')| \leq 2$.
\end{lemma}

\begin{proof}
    Suppose there exist $i_1 < i_2 < i_3 \in [q]$ such that $v_{i_1}, v_{i_2}, v_{i_3} \in \attachment^B_A(f')$.
    And let $u_1$ resp.\ $u_3$ be a vertex attached to $v_{i_1}$ resp.\ $v_{i_3}$ and adjacent to $B$.
    But then the vertex $v_{i_2}$ is~$B$-enclosed contradicting the assumption on $i_2$ for the following reason.
    Along $W'$ we first see all occurrences of $v_{i_1}$ in $W'$, then all occurrences of $v_{i_2}$ in $W'$, and then all occurrences of $v_{i_3}$.
    By definition of an attachment, there are two occurrences of $v_{i_1}$ resp.\ $v_{i_3}$ such that all occurrences of $u_1$ resp.\ $u_3$ happen between them.
    Therefore, along $W'$ we first see all occurrences of $u_1$ in $W'$, then all occurrences of $v_{i_2}$ in $W'$, and then all occurrences of $u_3$.
\end{proof}

If the region $f'$ is below the gate $g'$, we call a vertex $v$ with $N(v) \cap V_M \neq \emptyset$ \emph{interesting} if it is possible to draw an upward curve from $v$ to $g'$ that does not intersect $\Gamma(H)$; note that such an upward curve reaches $g'$ from below and is therefore fully contained in $f'$.
Symmetrically, if $f'$ is above $g'$, we call an external vertex $v$ with $N(v) \cap V_M \neq \emptyset$ \emph{interesting} if there is an upward curve from $g'$ to $v$ that does not intersect $\Gamma(H)$; note that such an upward curve starts above $g'$ and is therefore fully contained in $f'$.
The set $\interesting(f')$ denotes the set of all interesting vertices.

Now fix $A \in V_M$. 

\begin{definition}[Bundled $A$-vertices, part of \textbf{Definition 4} in the main body]\label{def:bundles-vertices}
Let $V^A(f') \subseteq \{v_1, \dots, v_q\}$ be the set consisting of all interesting external vertices $v$ with the following properties:
\begin{enumerate}
    \item the missing neighborhood of $v$ is $\{A\}$,
    \item $v A \in E_M$ if $f'$ lies below $g'$,
    \item $A v \in E_M$ if $f'$ lies above $g'$,
    \item $v$ is not $B$-enclosed for any $B \neq A \in V_M$,
    \item $v$ does not have a $B$-attachment for every $B \neq A \in V_M$.
\end{enumerate}
We call $V^A(f')$ the set of \emph{bundled $A$-vertices} with respect to $f'$.
Further, let $v^A_1, \dots, v^A_{q_A}$ be the order in which the vertices of $V^A(f')$ occur in $W'$.
For $i < j \in [q_A]$ we say that $v^A_i$ and~$v^A_j$ belong to the same \emph{$A$-bundle} if between any occurrence of $v^A_i$ and any occurrence of~$v^A_j$ in $W'$, no neighbor of $V_M \setminus \{A\}$ occurs.
Note that the $A$-bundles partition the set~$V^A(f')$.
We use~$\Delta_A(f')$ to denote the set of $A$-bundles with respect to $f'$.
\end{definition}
Below, we provide a bound on the size of this set.

\begin{lemma}[Part of \textbf{Lemma 5} in the main body]\label{lem:small-number-of-bundles}
    Let $A \in V_M$. Then $|\Delta_A(f')| \leq k$.
\end{lemma}

\begin{proof}
    Suppose there exist $i_1 < i_2 < \dots < i_{k+1} \in [q_A]$ such that the vertices $v^A_{i_1}, v^A_{i_2}, \dots, v^A_{i_k} \in V^A(f')$ belong to pairwise different $A$-bundles.
    Recall that the vertices in $V^A(f')$ do not have $B$-attachments for any $B \neq A \in V_M$.
    Thus for every $j \in [k]$ there exist a pre-drawn vertex $w_j$ on the boundary of  $f'$ and a missing vertex $B_j \in V_M \setminus \{A\}$ with the following properties.
    First, $w_j$ is adjacent to $B_j$.
    Second, in the walk $W'$, the vertex $w_j$ occurs (possibly multiple times) between the occurrences of $v^A_{i_j}$ in $W'$ and $v^A_{i_{j+1}}$ in $W'$, i.e., the occurrences of $w_j$ all happen after all occurrences of $v^A_{i_j}$ and before all occurrences of $v^A_{i_{j+1}}$.
    Since $|V_M \setminus \{A\}| = k - 1$, there exist $j_1 < j_2 \in [k]$ such that $B_{j_1} = B_{j_2}$. 
    But then all occurrences of $w_{j_1}$ in $W'$ happen before all occurrences of $v^A_{i_{j_1 + 1}}$ in $W'$; And all occurrences of $w_{j_2}$ in $W'$ happen after all occurrences $v^A_{i_{j_1 + 1}}$ in $W'$.
    Since both $w_{j_1}$ and $w_{j_2}$ are adjacent to $B_{j_1}$, the vertex $v^A_{i_{j_1 + 1}}$ is $B_{j_1}$-enclosed contradicting the fact that vertices in $V^A(f')$ are not $B$-enclosed for any $B$.
\end{proof}

The next lemma represents the culmination of this subsection. Essentially, it allows us to restrict our attention to solutions where ``bundled'' vertices are routed in the same way. Since the number of bundles is bounded by a function of the parameter (as per \cref{lem:small-number-of-bundles}), this means we will be able to effectively keep track of the routing of bundled vertices in our dynamic program.

\begin{lemma}[Part of \textbf{Lemma 8} in the main body]\label{lem:clean-solutions}
Let $S$ be a solution to an instance $\mathcal{I}$ of \textsc{UPEF} and let $f'$ be a one-sided region in $f$. 
Let $X \subseteq V_M$ be the set of vertices placed in $f'$ by $S$ and let $A \notin X$. 
Consider two bundled $A$-vertices (with respect to $f'$) $v_1$ and $v_2$ which belong to the same $A$-bundle (with respect to $f'$) and such that in $S$ the edge $e_1\coloneq \{v_1, A\}$ is routed inside $f'$ while the edge $e_2 \coloneq \{v_2, A\}$ is not routed inside $f'$.
Then: 
\begin{itemize} 
\item There exists a solution $S'$ to $\mathcal{I}$ which can be obtained by performing a topological transformation of $(S-e_1) \cap f'$ and then adding an upward curve for $e_2$ that so that this edge is routed inside $f'$. 
\item Also, the set of edges routed inside $f'$ in $S'$ is the set of edges routed inside $f'$ in $S$ together with the edge $e_2$ while the set of edges routed in $f \setminus f'$ remains the same.
\item Furthermore, for every edge $e = \{v, B\} \in E_M^1$ other than $e_2$ (with $v \in V(H)$ and $B \in V_M$) the set of strips in which the drawings of $e$ leave $v$ remains the same while for the edge $e$ it remains the same up to one additional strip where the new drawing of $e_2$ leaves $v_2$.
\item The orderings in which the edges other than $e_2$ cross the gate of $f'$ coincide in $S$ and $S'$. 
\item The edge $e_2$ crosses the gate of $f'$ right before or right after some edge $e_1'=\{A,v_1'\}$, where $v_1'$ lies on the boundary of $f'$.
\end{itemize}
\end{lemma}

\begin{proof}
    Without loss of generality, we assume that the gate between $f'$ and $f \setminus f'$ is the top-gate of $f'$. 
    In the opposite case, the proof is symmetric.

    First of all, note that since $A$ is not placed in $f'$, the edge $e_1$ that leaves $v$ in $f'$ crosses the gate of $f'$ at least once.
    On the other hand, no upward curve crosses the gate more than once.
    The edge $e_1$, when followed from tail to head, crosses this gate from bottom to top (and therefore from $f'$ to $f \setminus f'$) so the vertex $A$ is the head of this edge and it lies (in terms of $y$-coordinates) above this gate.
    On the other hand, we know that the vertex $v_2$ is interesting with respect to $f'$ (by definition of bundled $A$-vertices) so there is an upward curve from $v_2$ to the gate of $f'$ so $v_2$ lies below the gate of $f'$, and therefore it lies below $A$ (in terms of $y$-coordinates).
    Since $S$ is a solution, the edge $e_2$ is drawn as an upward curve in~$S$ so $A$ is also the head of the edge $e_2$.
    So our goal is to adapt the solution $S$ in such a way that we can insert an upward curve from $v_2$ to $A$ to it.
    
    Let $W^*$ define the minimal subwalk of $W(f)$ that contains all occurences of $v_1$ and $v_2$ inside $f'$. 
    Note that then the whole subwalk $W^*$ belongs to $f'$. 
    Let $C^*$ be a subcurve of~$C(f)$ corresponding to $W^*$.
    We know that the edges $e_1$ and $e_2$ are in the same bundle, therefore~$W^*$ does not contain any neighbor of $V_M \setminus \{A\}$.
    Since $C(f)$ closely follows the boundary of~$f$, $C^*$ is not crossed by any missing edge other than of the form $\{v, A\}$ where $v$ belongs to~$W^*$.
    At the same time, $C^*$ may be crossed by multiple edges of the form $\{v, A\}$ where $v$ belongs to $W^*$.

    Since $v_2$ is interesting (recall \cref{def:bundles-vertices} of bundled $A$-vertices), there is an upward curve $\phi_2$ that starts in $v_2$ and ends on the gate of $f'$.
    Let $\phi_1$ be the drawing of $e_1$ in $S$. 
    We refer to \cref{fig:clean-solutions-lemma} for an illustration.    
    Note that both $\phi_1$ and $\phi_2$ cross $C^*$; however, $\phi_2$ may cross other curves in $S$, including curves for edges which are not incident with $A$ at all.
    We choose the edge $e_1' = \{v'_1, A\}$ whose drawing $\phi_1'$ in $S$ crosses $C^*$ such that no further edge crosses $C^*$ between the crossing with $\phi_1'$ and $\phi_2$. 
    First, we may assume that $\phi_2$ does not cross any edge of form $\{v, A\}$ with $v \in V(H)$: indeed, otherwise one may construct a new choice of $\phi_2$ by picking the first intersection of $\phi_2$ with such an edge and following the drawing of the crossed edge until we reach the gate of~$f'$. The newly constructed choice of $\phi_2$ remains an upward curve and connects $v_2$ to the gate of~$f'$.
    Moreover, we may assume that the intersections of $\phi_1'$ and $\phi_2$ with the gate coincide: this can be achieved by using an almost horizontal straight-line segment from the intersection of $\phi_2$ with the gate to the intersection of $e_1'$ with gate.
    Note that this does not create any intersections with edges of form $\{v, A\}$ since by the choice of $e_1'$, the intersections of~$e_1'$ and~$e_2$ are consecutive on the gate of $f'$.
    Let us now define the curve $\phi'_2$ that is identical to $\phi_2$ until reaching the gate of $f'$ and then follows $\phi_1'$ until reaching $A$.
    If $\phi_2$ does not cross any other edge, the lemma holds (even without requiring any further change to $S$).
    
    Now recall that the edge $e_1'$ is drawn in the solution $S$ so it does not cross any other edge.
    Since the curve $\phi_2'$ follows $e_1'$ from the gate of $f'$, it also does not cross any edge above the gate of $f'$. 
    However $\phi'_2$ can still cross some edges inside $f'$, i.e., in the part where it matches~$\phi_2$.
    To complete the proof, we provide a redrawing argument to eliminate such crossings. 

    Let us define the \emph{forbidden} region $F$ as the area enclosed by the following closed curve: we first follow $\phi_1'$ from gate till $v_1'$, then we follow $C^*$, and finally we follow $\phi_2$ (recalling that we ensured $\phi_2$ essentially touches $\phi_1'$ near the gate of $f'$), see \cref{fig:clean-solutions-lemma}(a) for the illustration.
    We will now adapt the drawing $S$ in such a way that no missing edge and no missing vertex is placed inside $F$.  

    Let $T^*$ denote the subforest of $T_f$ induced by all strips $s'$ such that the intersection of $s'$ with $F$ contains a drawing of a missing vertex or a missing edge other than $e_1'$. 
    Let $P$ denote the path in $T^*$ that consists of all strips $s' \in T^*$ such that an internal point of $s'$ is visited by $\phi_2$ (i.e., all strips $s' \in T^*$ that are crossed by $\phi_2$ from bottom to top). 
    The next claim shows that one can locally transform the solution $S$ to always avoid a ``dangling leaf'' in $T^*$.
    
    In the following, we will carry out certain redrawing procedures to obtain the desired solution in the end.
    Along the way, we will, however, violate some properties of a solution: namely, we will possibly place some missing vertices into polar strips and possibly place multiple missing vertices inside the same strip---in the very end, we will ensure that this properties are not violated.
    So we say that a drawing is an \emph{almost-solution} if is satisfies all properties of a solution apart from, possibly, the aforementioned two.
        
    \begin{claim}
    \label{cl:redrawfar}
        Assume there exists a strip $\ell \notin P$ such that $\ell$ is a leaf of $T^*$. Then $\ell$ is adjacent to another strip $t$ of $T^*$ and there is an almost-solution $S'$ with the following properties:
        \begin{itemize}
\item $S'$ precisely coincides with $S$ everywhere in $f$ except in $F \cap (t \cup \ell)$;
\item $S'$ does not intersect $F\cap \ell$.
        \end{itemize}
\end{claim}

        \begin{figure}[h]
            \centering
            \includegraphics[scale=0.8]{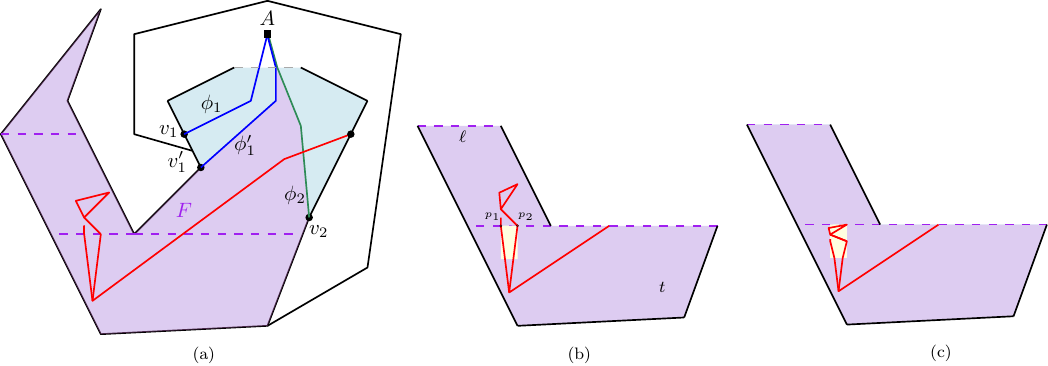}
            \caption{The proof of \cref{cl:redrawfar}. (a) The forbidden region $F$ defined by the curves $\phi_1'$, and $\phi_2$ is sketched purple. (b) We ensure that the $x$-coordinates of $\beta$ lie between $x(p_1)-\varepsilon'$ and $x(p_q)+ \varepsilon'$. (c)~We scale down vertically and shift the drawing out of $\ell$. The region $\zeta$ is depicted light yellow.}
            \label{fig:clean-solutions-lemma}
        \end{figure} 
    
    \begin{claimproof}
        Recall that by the choice of $e_1'$, no missing edge (apart from the curves~$\phi_1'$ and~$\phi_2$) crosses the subcurve of $C^*$ between the intersections with $\phi_1'$ and $\phi_2)$.
        In other words, along this subcurve $C$ no drawing of a missing edge other than $e_1'$ and $e_2$ starts.
        Therefore, the drawing of any path from a missing vertex placed in $\ell \cap F$ to the pre-drawn graph $H$ must cross the curve $\phi_2$, i.e., visit a strip in $P$.
        Recall that the graph $G$ is connected so if some missing vertex is placed in $\ell \cap F$, then there is a path in $T^*$ from $\ell$ to a vertex of $P$: this path (in $T^*$) is obtained by following the path (in $G$) from a missing vertex in $\ell$ until we cross $\phi_2$ for the first time.
        So if $\ell$ contains at least one missing vertex, then~$\ell$ is not isolated and in particular must be adjacent to a strip $t$ in $T^*$.
        
        Now suppose that $\ell \cap F$ only contains segments of missing edges and no missing vertex.
        We show that $\ell$ cannot be a leaf of $T^*$ that does not belong to $P$.
        Let $e$ be an arbitrary edge such that the drawing of $e$ in $S$ intersects $\ell \cap F$.
        Moreover, let $\alpha$ denote the part of the drawing of $e$ in $S$ restricted to the strip $\ell$.
        The strip $\ell$ does not belong to $P$, i.e., it is not crossed by $\phi_2$. 
        Moreover, $S$ is a planar drawing, so $\alpha$ crosses neither $\Gamma(H)$ nor~$\phi_1'$.
        Therefore,~$\alpha$ is fully contained in $F$.
        Also, $e_1'$ and $e_2$ belong to the same bundle and the choice of $e_1'$ ensures that the edge $e$ does not have an end-vertex on $W^*$ between the vertices~$v_1'$ and~$v_2$.
        Since no missing vertex is placed inside $\ell \cap F$, the end-vertices of $e$ lie outside~$\ell$. 
        This implies that there exist two strips~$s_1$ and $s_2$ such that each of them is a neighbor of~$\ell$ in the strip tree and the drawing of the edge $e$ in $S$ intersects $s_1 \cap F$ and $s_2 \cap F$.
        But then~$s_1$ and~$s_2$ are two neighbors of $\ell$ in the forest $T^*$ contradicting the fact that $\ell$ is a leaf of this forest.
        So indeed~$\ell$ is not isolated in $T^*$ and we use $t$ to denote its unique neighbor in~$T^*$.
        Moreover, this shows that for every edge that intersects $\ell \cap F$, at least one of its end-vertices is a missing vertex in $\ell \cap F$.
        And either the other end-vertex belongs to $\ell \cap F$ as well and then the drawing of $e$ is fully contained in $\ell \cap F$; or the drawing of the edge $e$ continues in~$t \cap F$.
        
        Now we provide a redrawing to obtain the desired almost-solution $S'$, we refer to \cref{fig:clean-solutions-lemma} for an illustration of this procedure. 
        Without loss of generality, we assume that $\ell$ lies above~$t$.
        Let $\beta$ be the part of $S$ consisting of all partial curves and points that intersect $F\cap \ell$. 
        Let~$p_1, \dots, p_q$ be all the points on the gate between $\ell$ and $t$ which are touched by $\beta$ (i.e., where the curves intersecting $\beta$ cross that gate), in the order of increasing $x$-coordinates; note that it may occur that $p_1=p_q$.
        Let us now consider a redrawing of $\beta$, say denoted~$\beta'$, which is still placed inside $\ell$ but whose $x$-coordinates lie entirely within the interval between~$x(p_1) - \varepsilon'$ and $x(p_q) + \varepsilon'$ for a sufficiently small value of $\varepsilon'$ so that no other curve of $S$ crosses the gate of $f'$ within this interval.
        In particular, we achieve this via a simple horizontal rescaling of all parts of $\beta$ which extend beyond the horizontal interval spanned by~$x(p_1) - \varepsilon'$ and $x(p_q) + \varepsilon'$, thus ensuring that no two curves in $\beta'$ intersect each other (see \cref{fig:clean-solutions-lemma}~(b)).         
        Further, we slightly perturb the drawing of $\beta'$ so that no vertex and no bend occurs in $F$ on the gate between $\ell$ and $t$. 

        Let $\gamma$ be a horizontal line segment on the gate of $f'$ between $x(p_1) - \varepsilon'$ and $x(p_q) + \varepsilon'$.
        Let $\varepsilon$ be a second sufficiently small distance, in particular satisfying the following property: the rectangle $\zeta$ obtained by projecting $\gamma$ into $t$ with height $\varepsilon$ contains no vertices or bends of~$S$ and only intersects curves of $S$ on its bottom and top boundary (see \cref{fig:clean-solutions-lemma}~(c)). In particular, this ensures that the rectangle $\zeta$ only contains straight-line segments representing the edges that cross the gate at the points $p_1,\dots,p_q$.

We now:
        \begin{itemize}
        \item vertically scale the part of $S$ in $\zeta$ down by a factor of $2$, 
        \item vertically scale all of $\beta'$ to span a vertical height less than $\varepsilon/2$, and
        \item place the rescaled part of $\zeta$ and the rescaled part of $\beta'$ just below~$\gamma$.
        \end{itemize}
        See \cref{fig:clean-solutions-lemma}~(c) for an illustration.     

        Let $S'$ be the object obtained by replacing $\beta\cup \zeta$ with the construction defined in the previous paragraph. Observe that since $\beta'$ was constructed in a way that no two curves were pairwise intersecting, no missing vertex or bend were placed in $\zeta$ by $S$ (now $\zeta$ contains all the elements previously represented by $\beta$), and since $S'$ coincides with $S$ outside of $\ell\cup \zeta$, we conclude that $S'$ is an almost-solution to $\mathcal{I}$, as desired.
        Also note that this redrawing did not change, for any drawing of any edge, say $\{v, B\} \in E_M^1$ with $v \in V(H)$, the strip in which this drawing leaves $v$.
    \end{claimproof}

    The exhaustive application of Claim~\ref{cl:redrawfar} allows us to reduce to the case where $T^*\subseteq P$ as for the following reason. 
    Suppose there is a strip $x$ in $T^* \setminus P$, but all leaves of $T^*$ belong to $P$, in particular, $x$ is an internal node of $T^*$.
    Since $T^*$ is a forest and $x$ is an internal node, there are two distinct leaves $a$ and $b$ of $T^*$ such that the path, say $Q$, from~$a$ to~$b$ in~$T^*$ contains~$x$.
    Let $Q'$ be the subpath of this path such that first, it contains $x$, second, its end-vertices belong to $P$, and finally, no internal vertex belongs to $P$: the path $Q'$ can be obtained by going along $Q$ from $x$ in both directions and ending at the first vertices that belong to $P$.  
    But then $Q'$ together with the subpath of $P$ between the end-vertices of $Q'$ forms a cycle in the strip tree $T_f$ contradicting the fact that $T_f$ is acyclic (see \cref{pro:striptree}).
    Therefore, as long as $T^* \setminus P \neq \emptyset$, there is a leaf of $T^*$ that does not belong to $P$ and we can apply \cref{cl:redrawfar} to eliminate this leaf.
    So from now on we may assume that $T^* \subseteq P$ holds.
    
    For a distance $\varepsilon$, let the $\varepsilon$-\emph{band} of $\phi_2$ be the connected region containing all points of~$f' \setminus F$ which lie at a horizontal distance of at most $\varepsilon$ from $\phi_2$ and simultaneously lie at a vertical distance of least $\varepsilon$ from the endpoints of $\phi_2$. Intuitively, one can imagine that the $\varepsilon$-band is obtained by projecting $\phi_2$ horizontally out of $F$ to a distance of $\varepsilon$, and then cutting off a small ($\varepsilon$-high) piece from the bottom and top. The \emph{spine} of the $\varepsilon$-band is the curve that connects (i.e., contains) all points $p$ in the $\varepsilon$-band which are equidistant to the two boundary points of the $\varepsilon$-band with the same $y$-coordinate as $p$; intuitively, the spine can be seen as a ``central curve'' of the $\varepsilon$-band. Observe that by construction, if $\varepsilon$ is sufficiently small then the spine, the left boundary, and the right boundary of the $\varepsilon$-band of $\phi_2$ all have the same shape, see \cref{fig:redrawing-spine}(a) for the illustration.

    \begin{figure}[h]
\centering
\includegraphics[scale = 0.9]{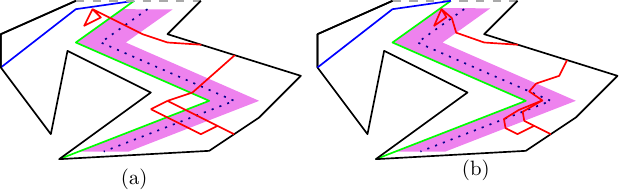}
\caption{In pink $\varepsilon$-band of $\phi_2$---the region to which all the missing edges and vertices, that are placed by $S$ inside $F$, will be shifted. (a) A drawing with $T^* \subseteq P$. The $\varepsilon$-band of $\phi_2$ is pink, its spine is in dashed dark blue. The curves $\phi_1'$ and $\phi_2$ are solid lines, blue and green respectively. (b) The almost-solution $S'$ obtained in proof of \cref{claim:spine}.}
\label{fig:redrawing-spine}
\end{figure}

    \begin{claim}\label{claim:spine}
     Assume that $T^*\subseteq P$. Then there exists a sufficiently small distance $\varepsilon$ and an almost-solution $S'$ with the following properties:
         \begin{itemize}
 \item $S'$ does not intersect $F$; and
 \item $S'$ precisely coincides with $S$ outside of $F$ and the $\varepsilon$-band of $\phi_2$.
         \end{itemize}
\end{claim}

\begin{claimproof}
Let $y_\downarrow$ and $y_\uparrow$ be the smallest and largest $y$-coordinates of a strip in~$P$. Observe that, by definition, the curve $\phi_2$ spans
the whole vertical distance of all strips in~$P$, i.e., from $y_\downarrow$ to $y_\uparrow$. Moreover, since $T^*\subseteq P$, 
the missing vertices and the straight-line segments of edges which $S$ places inside $F$ span a union of subintervals of the $y$-coordinates from $y_\downarrow$ to $y_\uparrow$. By slightly perturbing the $y$-coordinates of bends and missing vertices of $S$, we can choose an $\varepsilon$ such that:
\begin{itemize}
\item every missing edge and vertex of $S$ either intersects $F$ or does not intersect the $\varepsilon$-band of~$\phi_2$, 
\item no missing edge other than $e_2$ intersects the horizontal boundaries of the $\varepsilon$-band,
\item $\drawing$ does not intersect the $\varepsilon$-band of $\phi_2$, and 
\item every missing vertex and bend of $S$ is placed at a vertical distance of at least $\varepsilon$ from $y_\downarrow$ and $y_\uparrow$.
\end{itemize}

Given the above, we may obtain $S'$ by:
\begin{itemize}
\item[(1)] horizontally shifting all curves and points placed by $S$ inside $F$ by $\varepsilon/2$ to touch the spine of the $\varepsilon$-band of $\phi_2$ rather than $\phi_2$ itself, 
\item[(2)] horizontally downscaling the shifted drawing to fit entirely inside the side of the $\varepsilon$-band of $\phi_2$ which touches $\phi_2$ while preserving the contact points with the spine, so that no parts of the drawing are placed inside the $\frac{\varepsilon}{100}$-band of $\phi_2$;  
\item[(3)] horizontally downscaling all curves and points placed by $S$ inside the $\varepsilon$-band by a factor of $2$ while preserving the contact points with the side of the $\varepsilon$-band which does not touch~$\phi_2$.
Also note that this redrawing did not change, for any drawing of any edge, say~$\{v, B\} \in E_M^1$ with $v \in V(H)$, the strip in which this drawing leaves $v$.
\end{itemize}
 
An illustration of the construction is provided in \cref{fig:redrawing-spine}(b). Correctness follows by the fact that the intersection between $S'$ and the $\varepsilon$-band of $\phi_2$ was obtained by rescaling $S$ while preserving the contact points between the parts of the drawing where different scaling factors were applied.
\end{claimproof}

Applying \cref{claim:spine} makes the interior of the forbidden region $F$, enclosed by $\phi_1'$, $\phi_2$ and the fragment of $C^*$ between their endpoints $v_1'$ and $v_2$, free from any missing vertices and edges. Since the $\frac{\varepsilon}{100}$-band of~$\phi_2$ is free as well, the curve~$\phi_2$ does not cross a drawing of any other missing edge or vertex in the modified almost-solution $S'$. Moreover, all the edges that were routed inside $f'$ by $S$ are routed inside $f'$ by $S'$ as well and cross the gate of $f'$ in the same order. The only new edge routed inside $f'$ is the edge $e_2$ corresponding to the curve~$\phi_2$, and~$\phi_2$ crosses the gate right before or right after $\phi_1'$.
Note that the constructed object~$S'$ is an almost-solution.
We can now proceed analogously to the proof of \cref{lem:streamlining} to alter the drawing by slightly perturbing the missing vertices to obtain the desired solution $S'$.
\end{proof}

\begin{definition}
    We say that a solution is \emph{everywhere clean} if the following holds for every gate $g$.
    Let $f_0$ and $f_1$ denote the two one-sided regions into which $f$ is split by $g$. 
    Then for every $j \in \{0, 1\}$, every vertex $A$ placed inside $f_j$ by $S$, every $A$-bundle $\gamma$ with respect to $f_{1-j}$, and every pair $u \neq q \in \gamma$, if $\{u, A\}$ is routed inside $f_{1-j}$, then the edge $\{v, A\}$ is also routed inside $f_{1-j}$.
\end{definition}
Let us remark that the word ``everywhere'' in this notion refers to ``every gate $g$''.

\begin{remark*}
    First, since the vertex $A$ is placed in $f_j$, the above condition can be reformulated as follows: if some drawing of $\{u, A\}$ crosses~$g$, then also some drawing of $\{v, A\}$ crosses~$g$.
    We also emphasize that we do not forbid the drawings of $\{u, A\}$ or $\{v, A\}$ to be routed inside~$f_j$ as well. 
\end{remark*}
\setcounter{theorem}{7}
\begin{lemma}\label{lem:everywhere-clean-solutions}
    Every YES-instance of \upefshort\ admits an everywhere clean solution.
\end{lemma}

\begin{proof}
    The proof follows by exhaustive application of \cref{lem:clean-solutions}:
    We start with some solution of the instance.
    As long as there exists a gate $g$ (let $f_0$ and $f_1$ denote the two one-sided regions into which $f$ is split by $g$) $j \in \{0, 1\}$, vertex $A$ placed inside $f_j$ by $S$, an $A$-bundle~$\gamma$ with respect to $f_{1-j}$, and a pair $u \neq q \in \gamma$ such that~$\{u, A\}$ is routed inside $f_{1-j}$ and~$\{v, A\}$ is not routed inside $f_{1-j}$, we apply \cref{lem:clean-solutions} to extend the solution $S$ with a new curve for the edge $\{v, A\}$ such that $\{v, A\}$ is now also routed inside $f_{1-j}$.
    The statement of \cref{lem:clean-solutions} implies that after this redrawing, the number of pairs $(e, s)$ where $e \in E_M^1$, $s$ is a strip, and the solutions contains a drawing of $e$ that leaves its pre-drawn end-vertex in $s$ increased by one---therefore, the process terminates after a finite number of changes to the initial solution.
\end{proof}

\subsection{Partial Solutions and Records}
\label{sub:partialrec}
Using the results obtained in \cref{sub:onesided}, we will now formalize the records and the notion of partial solutions that will be used in our algorithm.
For the purposes of this section, let us fix an arbitrary non-root strip $s$ in $T_f$ and let $X \subseteq V_M$.
Recall that $T_f^s$ is the subtree of the strip tree $T_f$ rooted at strip $s$ and $f^s$ is the restriction of the face $f$ to the strips contained in $T_f^s$. 
Without loss of generality, assume for the rest of this section that the parent of $s$ lies above $s$, the definitions and results for the other case are entirely symmetric.

We say that the vertex $v$ is internal/external/interesting with respect to $s$ if it is internal/external/interesting with respect to the one-sided region $f^s$.
Recall that for every vertex $A$, the bundled $A$-vertices with respect to $f^s$ are denoted by $V^A(f^s)$, and the bundled $A$-vertices with respect to $f \setminus f^s$ are denoted by $V^A(f \setminus f^s)$.
We define the \emph{inner~$A$-bundles} of $s$ as the bundles of $V^A(f^s)$ with respect to the one-sided region $f^s$.
And we define the \emph{outer $A$-bundles} of $s$ as the bundles of $V^A(f \setminus f^s)$ with respect to the one-sided region $f \setminus f^s$.
Let us emphasize that for an outer $A$-bundle of $s$ some of the vertices of this bundle might, in general, not occur on the boundary of $f^s$.
However, since the definition of outer $A$-bundles of $s$ depend on the instance $\mathcal{I}$ only, and not on partial solutions, our dynamic-programming algorithm in \cref{sec:dynprog-new} will be able to efficiently check whether two pre-drawn vertices belong to the same outer $A$-bundle of $s$.

\begin{lemma}\label{lem:outer-bundles-behaviour}
    Let $s$ be a non-leaf strip such that the parent of $s$ lies above $s$.
    And let $t$ be a child of $s$ lying below $t$. 
    Further, let $v$ be a vertex 
    that is $A$-bundled with respect to $f \setminus f^t$, i.e., we have $v \in V^A(f \setminus f^t)$.
    If $v \notin V^A(f \setminus f^s)$, then $v$ is internal to $s$ or $v$ is not interesting with respect to $f \setminus f^s$.
\end{lemma}

In simpler words, the lemma states that
if a vertex is~$A$-bundled with respect to the child~$t$ but not~$A$-bundled with respect to the parent $s$ ``anymore'', then this could only because either $v$ is internal for~$f^s$, or~$v$ is not interesting with respect to $f \setminus f^s$---in particular, $v$ cannot get an attachment or become enclosed in $f \setminus f^s$ if this was not the case for $f \setminus f^t$.
\begin{proof}

    \begin{figure}[t]
        \centering
        \includegraphics[scale = 0.9]{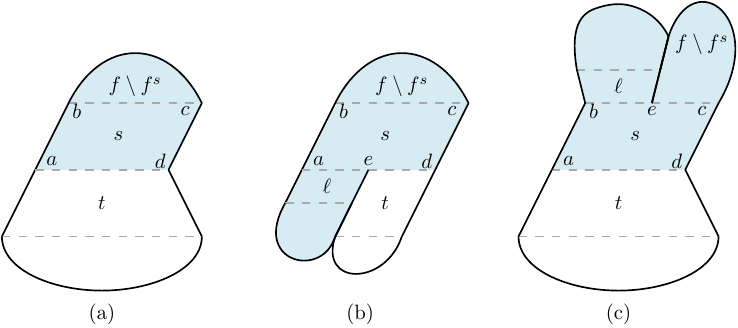}
        \caption{Illustration to the proof of \cref{lem:outer-bundles-behaviour}: the walk $W(f \setminus f^s)$ is a subwalk of $W(f \setminus f^t)$. The one-sided region $f \setminus f^t$ is sketched in blue.}
        \label{fig:lem-outer-bundles}
    \end{figure}
    
    First, let $v$ be a vertex in $V^A(f \setminus f^t) \setminus V^A(f \setminus f^s)$.
    Recall that by definition, we have~$f \setminus f^s \subseteq f \setminus f^t$.
    We make a case distinction on the properties of $v$ on the boundary of~$f \setminus f^s$.
    If $v$ does not appear on the boundary of $f \setminus f^s$, then all strips incident with $v$ belong to~$f^s$, i.e., $v$ is internal to $f^s$ as desired.
    So in the remainder of the proof we may assume that $v$ appears on the boundary of $f \setminus f^s$.
    Suppose $v$ is not external to $f \setminus f^s$, i.e., $v$ lies inside the polygon defined by the outer boundary of $f \setminus f^s$.
    Then $f \setminus f^s \subseteq f \setminus f^t$ implies that $v$ also lies inside the polygon defined by the outer boundary of $f \setminus f^t$, and therefore $v$ is not external to~$f \setminus f^t$ contradicting $v \in V^A(f \setminus f^t)$.
    Thus, $v$ is external to $f \setminus f^s$.
    
    Further, $f \setminus f^s \subseteq f \setminus f^t$ implies that $W(f \setminus f^s)$ is a subwalk of $W(f \setminus f^t)$.
    To see this more formally, we consider the three following cases that can occur by assumption made due to \cref{lem:streamlining}.
    First, suppose that $t$ is the unique child of $s$ and let $a,b,c,d$ be the vertices on the boundary of $s$ (see \cref{fig:lem-outer-bundles}~(a))---then $W(f \setminus f^t) = a, \{a,b\}, W(f \setminus f^s), \{c, d\}, d$.
    Let us remark that, formally speaking, we cannot have, for example, $b$ and $c$ both being the vertices since $\Gamma(G)$ is in general position.
    So technically speaking, the above equality holds only if we drop some vertices on the boundary of $s$ that do not exist (e.g., one of~$b$ and $c$) from the walk on the right-hand side, and also replace the corresponding edges (e.g.,~$ab$) by the edges to which the straight-line segment on the boundary of $s$ really belong to.
    However, this technicality does not invalidate the fact that $W(f \setminus f^s)$ appears as a subwalk on the right-hand side.
    For this reason, we decided to keep this notation in the following two cases as well, despite the slight abuse of notation.
    Second, suppose that $s$ has exactly two children and the child, say $\ell$, other than $t$ also lies below $s$.
    Without loss of generality we may assume that the gate shared by $\ell$ and $s$ lies to the left of the gate shared by $s$ and~$t$.
    Let $a,b,c,d,e$ be the vertices on the boundary of $s$ (see \cref{fig:lem-outer-bundles}~(b))---then we have~$W(f \setminus f^t) = W(f^\ell), \{a,b\}, W(f \setminus f^s), \{c, d\}, d$.
    Finally, we may assume that $s$ has exactly two children and the child, say $\ell$, other than $t$, lies above $s$.
    Without loss of generality we may assume that the gate shared by $\ell$ and $s$ lies to the left of the gate shared by $s$ and its parent.
    Let $a,b,c,d,e$ be the vertices on the boundary of $s$ (see \cref{fig:lem-outer-bundles}~(c))---then~$W(f \setminus f^t) = a, \{a, b\} W(f^\ell) W(f \setminus f^s), \{c, d\}, d$.
    So in any case $W(f \setminus f^s)$ is indeed a subwalk of $W(f \setminus f^t)$.
    
    Therefore, if $v$ would have a $B$-attachment with respect to $f \setminus f^s$ for some $B \neq A \in V_M$, then this would also be the case with respect to $f \setminus f^t$ contradicting $v \in V^A(f \setminus f^t)$.
    Similarly, if $v$ would be $B$-enclosed with respect to $f \setminus f^s$ for some $B \neq A \in V_M$, this would also be the case with respect to $f \setminus f^t$ contradicting $v \in V^A(f \setminus f^t)$.
    Therefore, by definition of~$V^A(f \setminus f^t)$ (\cref{def:bundles-vertices}) the only reason for $v \notin V^A(f \setminus f^s)$ remaining is that $v$ is not interesting with respect to $f \setminus f^s$ as claimed.
\end{proof}

\begin{figure}[t]
\centering
\includegraphics[scale = 0.9]{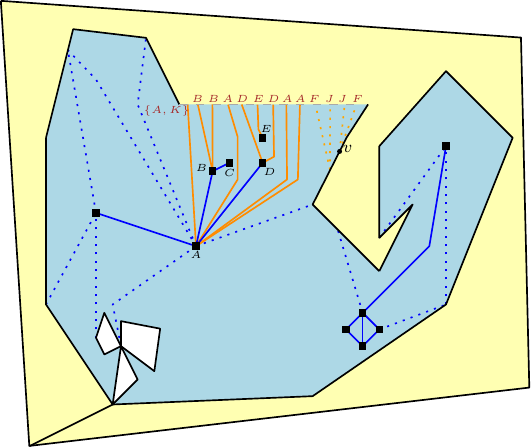}
\caption{A partial solution for a strip $s$, the region $f^s$ is sketched in light blue. The set $X$ of vertices placed is depicted by squares. The finished edges are blue and the curves ending on the gate are drawn orange. 
The drawings of routed edges are dotted.
We do not depict the labels of the labeled curves, instead providing their types in dark red along the gate. The type of the labeled curve is the intersection of its label with the set of missing vertices. In this example, all orange curves apart from the leftmost (labeled $AK$, and therefore having the type $\{A, K\}$) are labeled with edges from $E_M^1$.  
The partial solution thus yields the gate word $\{A, K\}BADEDAFJF$. We recall that consecutive orange curves of the same type (e.g., the two of type $B$) are grouped together in a thread represented by only one letter in the gate word.}
\label{fig:partial-solutions}
\end{figure}

\medskip

\begin{definition}\label{def:partial-solution}
A \emph{partial solution} $S$ for a strip $s$ (see also \cref{fig:partial-solutions}) consists of: 
\begin{itemize}
\item \label{def:ps-first} an upward planar drawing of some subset $X\subseteq V_M$ and $Y\subseteq E_M$; for $X$ and $Y$ we say that they are \emph{finished} by $S$---the drawings of $Y$ are \textbf{unlabeled} curves,,
\item a set of \textbf{labeled} (will be made precise next) upward curves of non-zero length connecting some subset of vertices on the boundary of $f^s$ to the gate of $s$, 
\item and a set of \textbf{labeled} upward curves connecting vertices from $X$ to the gate of $s$, 
\end{itemize}
with the following properties:
\begin{enumerate}
\item \label{def:ps-no-crossing-gamma} All parts of the partial solution are drawn inside $f^s$ (i.e., they do not intersect $\Gamma(H)$) and no missing vertex is placed on the gate of $s$.
\item No two curves intersect. 
\item No drawing is placed inside any of the leaves of $T_f$.
\item No missing vertex is placed inside a polar strip, and no two missing vertices are placed inside the same non-polar strip.
\item No two labeled curves share the same endpoint on the gate.
\item \label{def:ps-label} For every labeled curve, its label is a missing edge in $E_M$ such that the tail of this edge matches the non-gate endpoint of that curve.
\item \label{def:ps-placed-vertices} For every finished missing vertex $A \in X$, the following hold:
\begin{itemize}
    \item Every missing edge, say $e$, with the head $A$ is finished and there is no curve labeled $e$.
    \item For every missing edge $e$ with tail $A$, we have $e \in Y$ or there is a curve starting in $A$ with label $e$. 
\end{itemize}
\item \label{def:ps-E-M-2} For every missing edge $\{A, B\} \in E_M^2$ (i.e., with both end-vertices from $V_M$):
\begin{itemize}
    \item If both $A$ and $B$ belong to $X$, then the edge $\{A, B\}$ is finished by $S$, i.e., $\{A, B\} \in Y$, and no curve is labeled $\{A, B\}$ in $S$.
    \item If exactly one of $A$ and $B$ belongs to $X$, then there is exactly one curve labeled $\{A, B\}$ in $S$.
\end{itemize}
\end{enumerate}
\end{definition}

In the following, unless specified otherwise, by \emph{curve} we refer to a labeled curve and we sometimes refer to unlabeled curves as \emph{finished edges} as those are drawings of finished edges.
For every edge $e = \{v, A\} \in E_M^1$ (i.e., with exactly one end-vertex from $V_M$) with $v \in V(H)$ on the boundary of $f^s$ and $A \in V_M$, we say that $e$ is \emph{routed} by $S$ if it is finished by $S$ or there is a curve in $S$ starting in $v$, i.e., in the \textbf{pre-drawn} endpoint of this edge, and labeled with $e$.
Informally speaking, the edge $e \in E_M^1$ with a pre-drawn end-vertex $v \in V(H)$ is routed by $S$ if there is a ``possibly partial'' drawing of $e$ in $S$ starting in $v$ and either ending in its other end-vertex or on the gate of $s$; note that all finished edges and all curves in $S$ lie inside $f^s$ so in particular, this partial drawing leaves $v$ inside $f^s$.

\begin{remark*}
    Let us explicitly point out that, in general, a partial solution is allowed to ``draw'' an edge from $E_M^1$ multiple times: in particular, it is allowed to have multiple curves with the same label, or to have multiple drawings of the same edge in $Y$.
\end{remark*}
Note that every solution to an instance $\mathcal{I}$ of \upef\ induces a partial solution for $s$ that is obtained by merely removing everything outside of $f^s$ and labeling every curve with the edge it is representing.
 
The \emph{type} of a missing edge $e=\{a,b\} \in E_M$ is defined as $\{a,b\}\cap V_M$, i.e., the subset of its endpoints that are missing.
Recall that every edge in $E_M$ has at least one end-vertex in~$V_M$.
For simplicity, we identify the set $V_M$ with the set of all singletons-sets consisting of some element of $V_M$, i.e., $\{\{A\} \mid A \in V_M\}$.
This way we may assume that the type of every missing edge is an element of the set $V_M \cup {V_M \choose 2}$.
The \emph{type} of a curve in a partial solution is the type of its label.

Since no two curves have the same endpoint on the gate, there is a well-defined left-to-right ordering in which these curves reach the gate.
Let $t_1, \dots, t_q$ be the types of curves reaching this gate in the order they reach this gate from left to right.
Now we are ready to define the \emph{record} of $S$: this is a tuple $(\omega, X, E^1, \Phi = (\Phi^A)_{A \in X})$ defined as follows:
\begin{enumerate}
    \item $\omega$, called the \emph{gate word}, is the word over the alphabet $V_M \cup {V_M \choose 2}$ obtained from the word $t_1 \dots t_q$ as follows.
    Let $r_1, \dots, r_p \in V_M \cup {V_M \choose 2}$ and $\alpha_1, \dots, \alpha_p \in \mathbb{N}^+$ be such that~$t_1 \dots t_q = r_1^{\alpha_1} r_2^{\alpha_2} \dots r_p^{\alpha_p}$ and $r_i \neq r_{i+1}$ for all $i \in [p-1]$.
    Note that the choice of~$r_1, \dots, r_p$ and $\alpha_1, \dots, \alpha_p$ is unique.
    Then we define $\omega$ as $\omega = r_1 r_2 \dots r_p$, i.e., $\omega$ is obtained from $t_1 \dots t_q$ by replacing every maximal non-empty subword that consists of multiple occurrences of the same letter with a single occurrence of this letter.
    \item $X$ is the above-mentioned set of missing vertices finished by $S$.
    \item \label{def:record-phi} For every $A \in X$, the set $\Phi^A$ is a subset of $\Delta_A(f \setminus f^s)$ with the following property: for every $\gamma \in \Delta_A(f \setminus f^s)$, we have $\gamma \in \Phi^A$ if and only if $S$ contains a curve ending on the gate of $s$ and labeled $Av$ for some $v \in \gamma$.
    \item \label{def:record-e1} Let $E' \subseteq E_M^1$ denote the set of edges routed by $S$, that is, the edges in $E_M^1$ finished by $S$ as well as edges $e \in E_M^1$ for which there is a curve labeled with $e$ in $S$ that starts in its \textbf{pre-drawn} endpoint on the boundary of $f^s$. Then it holds that 
    \[
        E^1 = E' \setminus \{Bv \mid B \in X, v \in \Phi^B\}.
    \]
\end{enumerate}

To elaborate on the use of the gate word, suppose we have a set of curves that are all labeled with edges of the same type (say $A$) and their endpoints on the gate are consecutive.
Intuitively, it may be useful to think of this set of curves as a ``thread'',  imagine that the endpoints on the gate are very close to each other, and fix $A \notin X$ (as the other case is symmetric).
From this gate on, all of these curves can be essentially treated as a single curve since they all head towards $A$: if one of these curves can be extended to an upward planar curve reaching $A$, then all of them can follow this curve very closely to reach $A$ in an upward planar way. 
Therefore, to obtain $\omega$, we eliminate all repetitions of $t_1 \dots t_q$.
In our proofs, we will sometimes say that these curves \emph{correspond} to a certain occurrence of $A$ in the gate word.

We now formalize the notion of being ``well-behaved'' for partial solutions. 
Crucially, we then use the insights from \cref{sub:onesided} (in particular, \cref{lem:everywhere-clean-solutions}) to prove that one can safely focus solely on such ``well-behaved'' partial solutions. 

\begin{definition}\label{def:clean-ps}
A partial solution $S$ of $f^s$ is called \emph{clean} if the following holds. 
Let~$X \subseteq V_M$ denote the set of missing vertices placed by $S$. 
First, for every $A \in V_M \setminus X$, every inner~$A$-bundle $\gamma$ of $s$ (i.e., $\gamma \in \Delta_A(f^s)$), and every pair $u, v \in \gamma$, if $\{u, A\}$ is routed by~$S$, then~$\{v, A\}$ is also routed by $S$.
Second, for every $A \in X$, every outer $A$-bundle~$\gamma$ of $s$ (i.e.,~$\gamma \in \Delta_A(f \setminus f^s)$), and every pair $u, v \in \gamma$, if there is a curve labeled $\{u, A\}$ in $S$ ending on the gate of $s$, then there is also a curve labeled $\{v, A\}$ in $S$ ending on the gate of $s$.
\end{definition}

Observe that the restriction of a solution clean at $s$ to $f^s$ is a clean partial solution of $f^s$.

\begin{definition}
We say that a partial solution $S$ at $s$ is \emph{realizable} if there exists a solution $S'$ of $\mathcal{I}$ such that 
\begin{itemize}
    \item the restriction of $S'$ to $f^s$ coincides with the drawing of $S$, and
    \item for every edge $e$ whose drawing in $S'$ crosses the gate, the curve in $S$ that coincides with the drawing of $e$ in $f^s$ is labeled with $e$.
\end{itemize}
In this case, we simply say that the restriction of $S'$ to $f^s$ \emph{coincides} with $S$.
\end{definition}

In the remainder of this section, we establish a series of properties that are present in all realizable partial solutions (Lemmas~\ref{lem:respecting-internal-vertices}-\ref{lem:respecting-partial-solutions}) and use these individual insights to prove that the number of records required for our dynamic program---formalized in \cref{def:acceptable}---is small enough to support an efficient dynamic programming procedure (\cref{thm:relevant-records}).

\begin{definition}[Respecting internal vertices]
    We say that a partial solution $S$ of $s$ \emph{respects internal vertices}, if for every vertex $v$ on the boundary of $f^s$ that is internal to $f^s$, every missing edge incident with $v$ is routed by $S$.
\end{definition}

\begin{lemma}[Part of \textbf{Lemma 7} in the main body]
\label{lem:respecting-internal-vertices}
    Every realizable partial solution respects internal vertices.
\end{lemma}

\begin{proof}
    Let $v$ be an internal vertex of $f^s$, i.e., every strip incident with $v$ belongs to $f^s$.
    Then in any solution $S'$, any curve representing an edge incident with $v$ leaves $v$ in a strip that belongs to $f^s$.
    Therefore, for the partial solution $S$ which coincides with $S'$ on $f^s$, all edges incident with $v$ are routed by $S'$.
\end{proof}

\begin{definition}[Respecting non-interesting vertices]
    Let $S$ be a partial solution of $s$ and let~$X$ be the set of vertices placed by $S$. 
    We say that $S$ \emph{respects non-interesting vertices} if 
    the following holds for every $A \in X$. 
    If $v$ is a vertex of $f^s$ adjacent to $A$ and not interesting with respect to $f \setminus f^s$, then $\{v, A\}$ is routed by $S$.
\end{definition}

\begin{lemma}[Part of \textbf{Lemma 7} in the main body]
\label{lem:respects-non-interesting}
    Every partial solution respects non-interesting vertices.
\end{lemma}

\begin{proof}
    Let $S$ be a realizable partial solution of $s$, let $X$ be the set of vertices placed by $S$, and let $S'$ be a solution whose restriction to $f^s$ coincides with $S$.
    If $A$ is the head of $\{v, A\}$, the claim follows from the fact that the definition of the partial solution (\cref{def:ps-placed-vertices}). 
    
    Hence, we may assume that $A$ is the tail of $\{v, A\}$, i.e., we are speaking about the edge~$Av$.
    The vertex $v$ is not interesting with respect to $f \setminus f^s$ so there is no upward curve from the gate of $s$ to $v$ in $f \setminus f^s$.
    The vertex $A$ is placed inside $f^s$.
    If the curve representing $Av$ would cross the gate of $s$, there would then be an upward curve from the gate to $v$.
    Therefore, this curve lies inside $f^s$ and it is, in particular, routed by $S$.
\end{proof}

\begin{definition}
    We say that a word $\omega$ over $V_M \cup {V_M \choose 2}$ is \emph{short} if
    every letter from ${V_M \choose 2}$ occurs in $\omega$ at most once, every letter from $V_M$ occurs in $\omega$ at most $k+{k \choose 2}$ times, and no two consecutive letters in $\omega$ coincide. 
    A partial solution of $s$ is \emph{short}, if for its record $(\omega, X, E_1, \Phi)$, the gate word $\omega$ is short.
\end{definition}

\begin{observation}\label{obs:short-words-cubic}
    Every short word over $V_M \cup {V_M \choose 2}$ has length at most ${k \choose 2} + k (k+{k \choose 2}) \in \mathcal{O}(k^3)$.
    The number of short words over $V_M \cup {V_M \choose 2}$ is therefore bounded by $\mathcal{O}(k^2)^{\mathcal{O}(k^3)} = k^{\mathcal{O}(k^3)}$.
\end{observation}

\begin{lemma}[Part of \textbf{Lemma 6} in the main body]
\label{lem:short-partial-solutions}
    Every realizable partial solution is short.  
\end{lemma}

\begin{proof}
    Let $S$ be a realizable partial solution of $s$, let $(\omega, X, E_1, \Phi)$ be the record of $S$, and let $S'$ be a solution whose restriction to $f^s$ coincides with $S$.
    By the definition of $\omega$, no two consecutive letters of $\omega$ are equal.
    Since for every $t \in {V_M \choose 2}$, the graph $G$ contains at most one edge of type $t$ and this edge crosses the gate of $s$ at most once, the letter $t$ appears at most once in $\omega$.
    So it remains to show that every letter from $V_M$ occurs at most once in $\omega$. 
    
    If $k = 1$, then the claim holds since 
    the alphabet consists of only one letter and it cannot repeat by definition of the gate word of a partial solution.
    So from now on let $k \geq 2$.
    Now suppose there is a letter $A \in V_M$ that occurs at least $k+{k \choose 2}+1$ times in $\omega$.
    Recall that no two consecutive letters in $\omega$ are equal.
    So between any two occurrences of $A$, there is an occurrence of a letter $(V_M \setminus \{A\}) \cup {V_M \choose 2}$.
    Again, every letter from ${V_M \choose 2}$ occurs at most once in $\omega$ and there are $k-1$ letters in $V_M \setminus \{A\}$.
    So the fact that $A$ occurs at least $k+{k \choose 2} + 1$ times in $\omega$ implies that there exists a letter $B \in V_M$ such that $A B A B A$ is a subsequence of~$\omega$.
    So there exist curves $\phi_A^1, \phi_B^1, \phi_A^2, \phi_B^2, \phi_A^3$ in $S$ representing the edges of types~$A$,~$B$,~$A$,~$B$,~$A$, respectively, that end on the gate of $s$ and the order in which these curves end on the gate from left to right is given by $\phi_A^1, \phi_B^1, \phi_A^2, \phi_B^2, \phi_A^3$.

    \begin{figure}[t]
        \centering
        \includegraphics[scale=1.5]{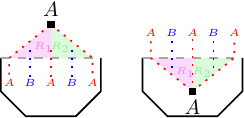}
        \caption{Proof of \cref{lem:short-partial-solutions}. If the gate word contains $ABABA$ as a subsequence for some $A \neq B \in V_M$, then the vertex $B$ must lie inside both $R_1$ and $R_2$ contradicting their disjointness. On the left resp.\ right, the vertex $A$ is placed inside resp.\ outside $f^s$.}
        \label{fig:short-words}
    \end{figure}
    
    We consider two internally disjoint connected regions $R_1$ and $R_2$ which are defined by the solution $S'$ as follows (see \cref{fig:short-words} for an illustration). 
    The boundary of the region $R_1$ follows the curve $\phi_A^1$ from $A$ to the gate of $s$, then it follows the gate till the intersection with~$\phi_A^2$, and then it follows the curve~$\phi_A^2$ till $A$.
    Similarly, the boundary of $R_2$ follows the curve~$\phi_A^2$ from $A$ to the gate of $s$, then it follows the gate till the intersection with~$\phi_A^3$, and then the it follows curve $\phi_A^3$ till $A$.
    Observe that the interiors of $R_1$ and $R_2$ are indeed disjoint. 
    Moreover, no pre-drawn vertex or edge lies inside of these two regions for the following reason. 
    The drawing is planar so no edge crosses either of the curves $\phi_A^1$, $\phi_A^2$, and~$\phi_A^3$ and by definition of a strip, no pre-drawn edge can cross the interior of a gate.
    The graph~$H$ is connected and since no pre-drawn edge crosses the gate between the endpoints of~$\phi_A^1$ and~$\phi_A^3$, no pre-drawn vertex lies in $R_1 \cup R_2$.

    Now recall that the curve $\phi_B^1$ resp.\ $\phi_B^2$ intersects the gate between the curves $\phi_A^1$ and $\phi_A^2$ resp.\ between $\phi_A^2$ and $\phi_A^3$.
    Let $e_B^1$ resp.\ $e_B^2$ denote the edges corresponding to $\phi_B^1$ and $\phi_B^2$ in~$S'$.
    For both these edges one end-vertex is $B$ and the other is pre-drawn.
    The edge $e_B^1$ resp.~ $e_B^2$ crosses the boundary of $R_1$ resp.\ $R_2$ exactly once so exactly one of its end-vertices lies inside $R_1$ resp.\ $R_2$.
    No pre-drawn vertex lies inside either of $R_1$ and $R_2$ so $B$ lies in both~$R_1$ and~$R_2$ contradicting the disjointness of the interiors of $R_1$ and $R_2$.
\end{proof}

\begin{definition}[Respecting enclosing]
    Let $S$ be a partial solution of $s$. 
    We say that a partial solution $S$ of $s$ \emph{respects enclosing} if for every pair $A \neq B \in V_M$ the following holds.
    Let~$f' = f^s$ if $A \notin X$ and let $f' = f \setminus f^s$ otherwise.
    For every pair $u \neq v$ of vertices external to $f'$ that are both adjacent to $A$ and both $B$-enclosed with respect to~$f'$, either both edges~$\{u, A\}$ and $\{v, A\}$ are routed by $S$ or both are not.
\end{definition}
Let us remark that for $A \in X$ (and therefore, $f' = f \setminus f^s$), one of the vertices, say $u$, might, in general, not be on the boundary of $f^s$: In this case, the edge $\{u, A\}$ cannot be routed by~$S$, and therefore if $S$ respects enclosing the edge $\{v, A\}$ is not routed by~$S$.

\cref{lem:enclosed-vertices-use-same-side} directly implies the following:
\begin{corollary}[Part of \textbf{Lemma 7} in the main body]\label{lem:respecting-partial-solutions}
    Every realizable partial solution respects enclosing.
\end{corollary}

\begin{definition}\label{def:acceptable}
    For non-root strip $s$, we call a tuple $R = (\omega, X, E^1, \Phi = (\Phi^A)^{A \in X})$ 
    with~$X \subseteq V_M$, $\omega$ being a word over $V_M \cup {V_M \choose 2}$, $E_1 \subseteq E_M^1$, and $\Phi^A \subseteq \Delta_A(f \setminus f^s)$ for every $A \in X$ an \emph{acceptable record} for $s$ if the following properties hold:
    \begin{enumerate}
        \item \label{def:acceptable-short} $\omega$ is short,
        \item \label{def:acceptable-internal} $E^1$ \emph{respects internal vertices}, i.e., for every vertex $v$ internal to $f^s$ and every edge~$\{v, A\} \in E_M^1$, we have $\{v, A\} \in E^1$.
        \item \label{def:respects-internal-outside} For every vertex $v$ internal to $f \setminus f^s$ and every $\{v, A\} \in E_M^1$, we have $\{v, A\} \notin E^1$.
        \item \label{def:acceptable-non-interesting} $E^1$ \emph{respects non-interesting vertices}, i.e., for every vertex $v$
        not interesting with respect to $f \setminus f^s$ and every $A \in X$ with $\{A, v\} \in E^1_M$, we have $\{v, A\} \in E^1$,
        \item \label{def:acceptable-non-interesting-inside} For every vertex $v$ 
        not interesting with respect to $f^s$ and every $A \in V_M \setminus X$ with~$\{v, A\} \in E^1_M$, we have $\{v, A\} \notin E^1$,
        \item \label{def:acceptable-enclosing} $E^1$ \emph{respects enclosing}, i.e., the following condition is satisfied for every pair $A \neq B \in V_M$.
        Let $f' = f^s$ if $A \notin X$ and let $f' = f \setminus f^s$ otherwise.
        For every pair $u, v$ of vertices external to~$f'$ that are both adjacent to $A$ and both $B$-enclosed with respect to $f'$, i.e.,~$u, v \in \enclosed^B_A(f')$, if we have $\{u, A\} \in E^1$, then we also have $\{v, A\} \in E^1$.
        \item \label{def:acceptable-inner-bundles} For every $A \in V_M \setminus X$, every inner $A$-bundle $\gamma$ of $s$, and every $u, v \in \gamma$, if $uA \in E^1$, then we also have $vA \in E^1$.
        \item \label{def:acceptable-E-M-2} For every pair $A \neq B \in V_M$ of vertices, it holds that $\{A, B\}$ occurs in $\omega$ if and only if we have $\{A, B\} \in E_M^2$, we have $|X \cap \{A, B\}| = 1$, and the unique element of this set is the tail of the edge $\{A, B\}$.
        \item \label{def:acceptable-not-placed} For every vertex $A \in V_M \setminus X$ and every edge $Av \in E_M^1$ we have $Av \notin E^1$.
        \item \label{def:acceptable-incoming edges} For every vertex $A \in X$ and every $vA \in E_M^1$ we have $vA \in E^1$.
        \item \label{def:acceptable-outgoing-edges} For every $A \in X$ and every $v \notin V^A(f \setminus f^s)$ with $Av \in E^1_M$, we have $Av \in E^1$ or $A$ occurs in $\omega$.
        \item \label{def:acceptable-outer-bundles-routed} For every $A \in X$, every $\gamma \in \Delta_A(f \setminus f^s) \setminus \Phi^A$ we have $\{Av \mid v \in \gamma\} \subseteq E^1$.
        \item \label{def:acceptable-outer-bundles-letter} For every $A \in X$, if $\Phi^A \neq \emptyset$, then $\omega$ contains the letter $A$.
        \item \label{def:acceptable-A-letter} For every $A \in X$, if $A$ occurs in $\omega$, then there exists a vertex $v$ with $Av \in E_M^1$ and $v \notin \gamma$ for every $\gamma \in \Delta_A(f \setminus f^s) \setminus \Phi^A$.
    \end{enumerate}
\end{definition}

\begin{remark*}
    We emphasize that the above definition of an acceptable record $R$ \textbf{does not} require that there is a partial solution of $s$ with record $R$. 
    The usage of acceptable records is as follows: in \cref{thm:relevant-records} we will show that on the one hand, every clean realizable partial solution has an acceptable record. 
    And on the other hand, we will show that the number of acceptable records of $s$ is bounded by a function of $k$ and, furthermore, we can generate all acceptable records of $s$ in fixed-parameter time.
    Then the dynamic-programming algorithm in \cref{sec:dynprog-new} will consider only acceptable records of $s$ as candidates for which realizability is checked implying the desired running time.
\end{remark*}

We will later show that we can check if a partial solution is acceptable by considering its record only, this property will be used by the dynamic programming algorithm in the next section.

\begin{theorem}\label{thm:relevant-records}
    For every clean realizable partial solution $P$, the record of $P$ is acceptable.
    Further, for a strip $s$, the number of acceptable records of $s$ is bounded by $k^{\bigoh(k^3)}$.
    Moreover, the set of all acceptable records of $s$ can be computed in time $k^{\bigoh(k^3)} n^{\bigoh(1)}$.
\end{theorem}

\begin{proof}
    Recall that without loss of generality we assume that the gate shared by $s$ and its parent is a top-gate of $s$.
    We start with proving the first part of the statement.
    Let $P$ be a realizable partial solution and let $R = (\omega, X, E^1, (\Phi^A)_{A \in X})$ be the record of $s$.
    The fact that~$R$ is acceptable almost follows from the definition of a clean partial solution (see \cref{def:partial-solution} and \cref{def:clean-ps}) and  together with \cref{lem:short-partial-solutions}, \cref{lem:respecting-internal-vertices}, \cref{lem:respects-non-interesting}, and \cref{lem:respecting-partial-solutions}.
    However, one needs to be careful since for the set $E'$ of edges routed by $P$, in general, we have $E^1 \subsetneq E'$ by definition of the record (\cref{def:record-e1}), namely we have
    \[
        E^1 = E' \setminus \{Bv \mid B \in X, v \in \Phi^B\}.
    \]
    For this reason, we go through the properties of an acceptable record (see \cref{def:acceptable}) one by one and argue that $R$ satisfies them.
    
    The property \cref{def:acceptable-short} follows from \cref{lem:short-partial-solutions}.
    By \cref{lem:respecting-internal-vertices} $P$ respects internal vertices, i.e., 
    for every vertex $A \in V_M$ and every vertex $v \in V(H)$ internal to $f^s$ with $\{v, A\} \in E_M^1$ we have $\{v, A\} \in E^1$.
    Since the set difference $E' \setminus E^1$ only contains edges incident with pre-drawn~$A$-bundled vertices of $f \setminus f^s$ for $A \in V_M$ (in particular, these pre-drawn vertices lie on the boundary of $f \setminus f^s$), 
    we then also have $\{v, A\} \in E^1$---so \cref{def:acceptable-internal} is satisfied.
    Further, for a vertex $v$ internal to $f \setminus f^s$, all strips incident with $v$ belong to $f \setminus f^s$, i.e., $v$ does not occur on the boundary of $f^s$---so no edge incident with $v$ can be routed by $P$ and \cref{def:respects-internal-outside} is satisfied.
    By \cref{lem:respects-non-interesting}, the partial solution $P$ respects non-interesting vertices, i.e., for every~$A \in X$ and every an external vertex $v$ of $f^s$ adjacent to $A$ and not interesting with respect to~$f \setminus f^s$, we have $\{v, A\} \in E'$.
    Since $v$ is not interesting with respect to $f \setminus f^s$, we have $v \notin \gamma$ for any~$\gamma \in \Phi^A$.
    Thus, it holds that $\{v, A\} \in E^1$---and therefore, \cref{def:acceptable-non-interesting} holds.
    Let $v$ be a vertex not interesting with respect to $f^s$.
    For every vertex $A \in V_M \setminus X$ with~$\{v, A\} \in E_M^1$, the edge~$\{v, A\}$ cannot be finished by $P$, no curve can start in $A$, and there cannot be an upward curve starting in the non-interesting vertex $v$ and ending on the gate of $s$.
    Thus we have~$\{v, A\} \notin E^1$---so \cref{def:acceptable-non-interesting-inside} holds.

    Further, by \cref{lem:respecting-partial-solutions} the partial solution $P$ respects enclosing, that is the following holds for every $A \neq B \in V_M$.
    Let $f' = f^s$ if $A \notin X$ and let $f' = f \setminus f^s$ otherwise.
    For every pair $u \neq v$ of vertices external to $f'$ that are both adjacent to $A$ and both $B$-enclosed with respect to $f'$, either we $\{u, A\} \in E'$ and $\{v, A\} \in E'$ or we have $\{u, A\} \notin E'$ and $\{v, A\} \notin E'$.
    For $A \notin X$, then this also applies to $E^1$ since $E'$ and $E^1$ coincide on the edges incident with vertices in $V_M \setminus X$.
    For $A \in X$, we recall that $B$-enclosed vertices with respect to~$f \setminus f^s$ are not $A$-bundled with respect to $f \setminus f^s$ (see \cref{def:bundles-vertices}).
    In particular, this implies~$u, v \notin \gamma$ for any $\gamma \in \Phi^A$, and therefore we also have either $\{u, A\} \in E^1$ and $\{v, A\} \in E^1$, or we have~$\{u, A\} \notin E^1$ and $\{v, A\} \notin E^1$---so \cref{def:acceptable-enclosing} is satisfied by $R$ as well.

    The partial solution $P$ is clean so for every $A \in V_M \notin X$, every inner $A$-bundle $\gamma$ of~$s$, and every $u, v \in \gamma$, if we have $uA \in E^1$, then we also have $vA \in E^1$.
    Since $E^1$ and $E'$ coincide on edges incident with $V_M \setminus X$, this property then also applies to $E'$---thus \cref{def:acceptable-inner-bundles}.
    By definition of the partial solution, $P$ satisfies \cref{def:ps-E-M-2} of \cref{def:partial-solution}.
    Furthermore by \cref{def:ps-label} every curve in $P$ labeled $AB$ for some $A, B \in V$ ending on the gate of $s$ starts in the tail $A$ of this edge so we have $B \in X$.
    Therefore, $R$ satisfies \cref{def:acceptable-E-M-2}.
    For every $A \in V_M \setminus X$, no edge incident with $A$ can be finished by $P$ and no curve can start in $A$ so we have $Av \notin E^1$ for any $v \in V(H)$, and therefore \cref{def:acceptable-not-placed} is satisfied.

    Further, consider $A \in X$ a vertex $v \in V(H) \setminus V^A(f \setminus f^s)$ with $Av \in E_M^1$ and $Av \notin E^1$.
    Due to $v \notin V^A(f \setminus f^s)$, we have $v \notin \gamma$ for any $\gamma \in \Phi^A$ so we also have $Av \notin E'$.
    So in particular, $Av$ is not finished by $P$.
    Since $P$ is a partial solution, by \cref{def:ps-placed-vertices} there exists a curve in $P$ starting in $A$, ending on the gate of $s$, and labeled $Av$.
    By definition of the gate word $\omega$ of $P$, the letter $A$ then occurs in $\omega$---so \cref{def:acceptable-outgoing-edges} is also satisfied by $R$.

    Since $P$ is realizable, there exists a solution $S$ whose restriction to $f^s$ coincides with~$P$.
    Consider a vertex $A \in X$, an outer $A$-bundle $\gamma$ of $s$ not belonging to $\Phi^A$ (i.e., $\gamma \in \Delta_A(f \setminus f^s) \setminus \Phi^A$), and a vertex $v \in \gamma$.
    By definition of the record $R$ of $P$ (see \cref{def:record-phi}), there is no curve labeled $Av$ in $P$ ending on the gate of $s$.
    Therefore, in the solution $S$, no drawing of the edge $Av$ crosses the gate of $s$.
    Due to $A \in X$, the vertex $A$ is placed inside $f^s$ by $S$.
    Thus, every drawing of $Av$ in $S$ is fully contained in $f^s$.
    Since $S$ and $P$ coincide on $f^s$, the edge $Av$ is routed by $P$, i.e., we have $Av \in E'$.
    Then $v \in \gamma \in \Delta_A(f \setminus f^s) \setminus \Phi^A$ implies that we also have $Av \in E^1$---so \cref{def:acceptable-outer-bundles-routed} also holds for $R$.
    
    Finally, consider a vertex $A \in X$ such that $\Phi^A \neq \emptyset$.
    By definition of the record $R$ (see \cref{def:record-phi}), there exists a curve labeled $Av$ in $P$ ending on the gate of $s$ with $v \in \gamma$ and $\gamma \in \Phi^A$.
    So the gate word $\omega$ of $P$ contains the letter $A$ as desired and \cref{def:acceptable-outer-bundles-letter} is satisfied by $R$.
    Finally, if for some $A \in X$, the letter $A$ occurs in $\omega$, there exists a curve labeled $Av$ for some~$Av \in E_M^1$ in $P$ ending on the gate of $s$.
    By definition of the record (see \cref{def:record-phi}), we then have~$v \notin \gamma$ for any $\gamma \in \Delta_A(f \setminus f^s) \setminus \Phi^A$---so \cref{def:acceptable-A-letter} is satisfied.
    Altogether, we obtain that~$R$ is an acceptable record for~$s$ proving the first part of the claim.

    We will prove the remaining two statements of the theorem together.
    For this we will show that we can define ``encodings'' whose number is upper-bounded by a function of $k$ and show that there is a way to generate every acceptable record from some encoding.
    An \emph{encoding} for $s$ is a tuple $\Lambda = (\omega, X, a, b, c, \Phi, \Psi)$ with the following properties:
    \begin{enumerate}
        \item $X \subseteq V_M$. 
        \item $\omega$ is a short word over $V_M \cup {V_M \choose 2}$.
        \item $a$ maps every pair $A \neq B \in V_M$ to a value $a_{A,B} \in \{0, 1\}$. 
        \item $b$ maps every triple $A \neq B \in V_M$ and $j \in [2]$ to a value $b_{A, B, j} \in \{0,1\}$.
        \item $c$ maps every triple $A \neq B \in V_M$ and $j \in [2]$, to a value $c_{A,B,j} \in \{0,1\}$.
        \item $\Phi$ maps every $A \in X$ to a set $\Phi^A \subseteq \Delta_A(f \setminus f^s)$, i.e., $\Phi^A$ is a subset of outer $A$-bundles of~$s$.
        \item $\Psi$ maps every $A \in V_M \setminus X$ to a set $\Psi^A \subseteq \Delta_A(f^s)$, i.e., $\Phi^A$ is a subset of inner $A$-bundles of~$s$.
    \end{enumerate}

    \begin{claim}
        The number of encodings for $s$ is bounded by $k^{\bigoh(k^3)}$.
    \end{claim}

    \begin{proof}
        For $X$ there are $2^k$ options.
        The number of options for $\omega$ is bounded by $k^{\bigoh(k^3)}$ by \cref{obs:short-words-cubic}.
        The number of options for each of $a$, $b$, and $c$ is bounded by $2^{\bigoh(k^2)}$.
        Further, by \cref{lem:small-number-of-bundles}, we have $\Delta_A(f \setminus f^s) \leq k$ and $\Delta_A(f^s) \leq k$ for every $A \in V_M$.
        Therefore, the number of options for each of $\Phi$ and $\Psi$ is bounded by $(2^k)^k = 2^{k^2}$.
        Altogether, the number of options for $\Lambda$ is bounded by 
        \[
            k^{\bigoh(k^3)} \cdot 2^{\bigoh(k^2)} \cdot 2^{\bigoh(k^2)} \cdot 2^{\bigoh(k^2)} \cdot 2^{k^2} \cdot 2^{k^2} \in k^{\bigoh(k^3)}. \qedhere
        \]
    \end{proof}
    
    Now we define the following algorithm which, as we will argue later, enumerates all acceptable records of $s$.
    First of all, in time polynomial in the size of the instance, we can construct the walks $W(f^s)$ and $W(f \setminus f^s)$.
    From these for every $f' \in \{f^s, f \setminus f^s\}$ we can now identify the internal vertices of $f'$, the external vertices of $f'$, interesting vertices of $f'$ (for this we check, for a vertex $v$ in question, if there is strip $t$ such that $v$ lies on the boundary of $t$ and there is a directed path from $t$ to $s$ in the strip tree $T_f$), as well as vertices that have $B$-attachments in $f'$ for $B \in V_M$, that are $B$-enclosed in $f'$ for $B \in V_M$, that have at least two neighbors in $V_M$ (together with the missing neighborhood of these vertices), and finally the $A$-bundles of $f'$ for every $A \in V_M$.
    Now we initialize the set $D_s = \emptyset$---in the end it will contain precisely the acceptable records of $s$.
    
    Now we iterate through all encodings $\Lambda = (\omega, X, a, b, c, \Phi, \Psi)$---this can be done in~$k^{\bigoh(k^3)} n^{\bigoh(1)}$ by exhaustively trying all options for $\omega, X, a, b, c, \Phi, \Psi$---and proceed as follows.
    We create a binary vector $\alpha_{\Lambda}$ indexed by $E_M^1$ and initialized with ``dummy values'' $\bot$ for all $e \in E_M^1$.
    In the following, we will describe how to fill the entries of $\alpha_{\Lambda}$ with binary values depending on~$\Lambda$.
    Let us remark that some of the values may be rewritten in this process.
    
    First of all, for every $A \in X$ and every $vA \in E_M^1$, set $\alpha_{\Lambda}(vA) = 1$.
    Similarly, for every~$A \in V_M \setminus X$ and every $Av \in E_M^1$, we set $\alpha_{\Lambda}(Av) = 0$.

    Now for every vertex $v$ inner to $f^s$ and every edge $e$ incident with $v$ we set $\alpha_{\Lambda}(e) = 1$.
    And for every vertex $v$ inner to $f \setminus f^s$ and every edge $e$ incident with $v$ we set $\alpha_{\Lambda}(e) = 0$. 
    For every $A \in V_M$, we define $f^A = f^s$ if $A \notin X$ and $f^A = f \setminus f^s$ if $A \in X$.
    Further, for every vertex $v$ adjacent to $A$ and not interesting with respect to $f^A$ we set $\alpha_{\Lambda}(\{v, A\}) = 0$ if~$A \notin X$ and $\alpha_{\Lambda}(\{v, A\}) = 1$ otherwise.
    
    Further, for every $A \neq B \in V_M$ we proceed as follows.
    Let $S = \enclosed^B_A(f^A)$.
    Then for every $v \in S$, we set $\alpha_{\Lambda}(\{v, A\}) = a_{A,B}$.
    Further, by \cref{obs:large-degree-vertices} there exist at most two (the precise number is denoted by $r$ where $r \leq 2$) external vertices, say $v_1$ and $v_2$, on the boundary of $f^A$ adjacent to both $A$ and $B$ and not $B$-enclosed.
    Then for $i \in [r]$ we set~$\alpha_{\Lambda}(\{v_i, A\}) = b_{A, B, i}$.
    Further, by \cref{obs:vertices-with-attachments}, there exist at most two ( the precise number is denoted by $z$ where $z \leq 2$) external vertices, say $u_1$ and $u_2$, on the boundary of~$f^A$ adjacent to~$A$, having a $B$-attachment, and not $B$-enclosed.
    Then for $i \in [z]$ we set~$\alpha_{\Lambda}(\{u_i, A\}) = c_{A, B, i}$.
    
    Further, for every $A \in X$, every $\gamma \in \Delta_A(f^A)$, and every $v \in \gamma$, we set $\alpha_{\Lambda}(\{v, A\}) = 0$ if~$\gamma \in \Phi^A$ and $\alpha_{\Lambda}(\{v, A\}) = 1$ otherwise.
    And finally, for every $A \in V_M \setminus X$, every $\gamma \in \Delta_A(f^A)$, and every $v \in \gamma$, we set $\alpha_{\Lambda}(\{v, A\}) = 1$ if $\gamma \in \Psi^A$ and $\alpha_{\Lambda}(\{v, A\}) = 0$ otherwise.

    \begin{claim}
        In the end $\alpha_{\Lambda}$ does not contain any $\bot$ entry.
    \end{claim}

    \begin{proof}
        Consider an edge $\{v, A\} \in E_M^1$.
        As the very first step we set the values for all edges incoming into $X$ and all edges outgoing from $V_M \setminus X$ so in the remainder we focus on edges outgoing from $X$ or incoming into $V_M \setminus X$.
        If $v$ is internal with respect to $f^s$ or $f \setminus f^s$, the algorithm sets the value $\alpha_{\Lambda}(\{v, A\})$ as well.
        Thus we may assume that $v$ is external with respect to both $f^s$ and $f \setminus f^s$.
        If $v$ is not interesting with respect to $f^A$, the value~$\alpha_{\Lambda}(\{v, A\})$ was also set in the second step.
        In the remaining cases, $v$ can either be~$B$-enclosed or have a~$B$-attachment for some, or be one of at most $2$ vertices adjacent to $B$ for some~$B \neq A \in V_M$, or be one of at most $2$ vertices with a~$B$-attachment for some $B \neq A \in V_M$, or~$v$ is $A$-bundled with respect to $f^A$.
        In any of these cases, the value $\alpha_{\Lambda}(\{v, A\})$ was set.
        So the claim holds.    
    \end{proof}

    We define the set $E^1(\Lambda) = \{e \in E_M^1 \mid \alpha(e) = 1\}$.
    After that we define $R = (\omega, X, \Phi, E^1(\Lambda))$ and in time polynomial in $n$ we check whether $R$ is an acceptable record for~$s$.
    If so, we add it to $D_s$, otherwise, we discard it and continue with the next $\Lambda$.

    \begin{claim}
        In the end, $D_s$ contains precisely the acceptable records $s$.
    \end{claim}

    \begin{proof}
        First, the check at the very end ensures that only acceptable records of $s$ are contained in $D_s$.
        On the other hand, let $R = (\omega, X, E^1, \Phi)$ be an acceptable record of $s$.
        We now construct an encoding $\Lambda$ with $E^1(\Lambda) = E^1$.
        We let $\Lambda = (\omega, X, a, b, c, \Phi, \Psi)$ defined as follows.
        For every $A \neq B \in V_M$ we let $a_{A,B} = 1$ if for every $v \in \enclosed^B_A(f^A)$ we have $\{v, A\} \in E^1$ and we let $a_{A,B} = 0$ if for every $v \in \enclosed^B_A(f^A)$ we have $\{v, A\} \notin E^1$---\cref{def:acceptable-enclosing} ensures that precisely one of these cases occurs (unless $\enclosed^B_A(f^A) = \emptyset$ in which case we let $a_{A,B}$ be arbitrary as it is never used by the algorithm). 

        Further, for $A \neq B \in V_M$ and $j \in [r]$ we set $b_{A, B, j} = 1$ if and only if $\{v_j, A\} \in E^1$. 
        We also set $b_{A, B, j} = 1$ for $r < j \leq 2$ (this value is never used then).
        Similarly, for $A \neq B \in V_M$ and $j \in [z]$ we set $c_{A, B, j} = 1$ if and only if $\{u_j, A\} \in E^1$. 
        We also set $c_{A, B, j} = 1$ for $r < j \leq 2$ (this value is never used then).

        Finally, for every $A \in V_M \setminus X$ we start with $\Psi^A = \emptyset$.
        Then for every $\gamma \in \Delta_A(f \setminus f^s)$ we add $\gamma$ to $\Psi^A$ if we have~$\{vA \mid v \in \gamma\} \subseteq E^1$ and we do not add it to $\Psi^A$ if we have~$\{vA \mid v \in \gamma\} \cap E^1 = \emptyset$---note that one of these cases occurs since the acceptable record~$R$ satisfies \cref{def:acceptable-inner-bundles}.

        Note that now by construction of $\alpha_{\Lambda}$, it only sets some value $\alpha(e)$ to 1 if we have $e \in E^1$, and it only sets some value $\alpha(e)$ to 0 if we have $e \notin E^1$.
        In particular, it never rewrites a value $\alpha_\Lambda(e) = 0$ by $\alpha_\Lambda(e) = 1$ or vice versa. 
        Together with the fact that by the above remark in the end, $\alpha_\Lambda$ does not contain any $\bot$ entry, we obtain $E^1(\Lambda) = E^1$.
        Since $R$ is acceptable, it is not discarded due to the check in the end of the algorithm. 
    \end{proof} 
    Thus, the size of $D_s$ is upper-bounded by the number of encodings of $s$, i.e., by $k^{\bigoh(k^3)}$.
    Note that for a fixed $\Lambda$, the procedure runs in time polynomial in $n$.
    Therefore, the set $D_s$ of acceptable records of $s$ is computed in time $k^{\bigoh(k^3)} n^{\bigoh(1)}$ concluding the proof.
\end{proof}

\section{The Dynamic Programming Algorithm}
\label{sec:dynprog-new}

With the results of the previous section in hand, we are finally ready to provide a fixed-parameter algorithm that solves \upefshort\ by dynamic programming along the strip tree $T_f$.
Let $S$ be an everywhere clean solution and $s$ be a strip. 
A record $R$ is the \emph{record of} $S$ \emph{at} $s$ if $R$ is the record of the partial solution obtained from $S$ by restricting it to $f^s$.
The high-level idea of the algorithm is to proceed in a leaf-to-root fashion and to compute for every node $s \in V(T_f)$, a \emph{representative} set $\mathcal{R}_s$ of records, that is (1) every record in $\mathcal{R}_s$ is acceptable for $s$, (2) for every record $R \in \mathcal{R}_s$, there exists a clean partial solution $P$ for $s$ with record $R$, and (3) for every everywhere clean solution $S$, the record of $S$ at $s$ belongs to $\mathcal{R}_s$.

\begin{corollary}\label{cor:size-of-representative}
    If $\mathcal{R}_s$ is a representative of $s$, it holds that $|\mathcal{R}_s| \leq k^{\bigoh(k^3)}$.
\end{corollary}
The bound from the above corollary is immediately obtained by \cref{thm:relevant-records}.
Our algorithm will compute a representative for every strip $s$.
 Our preprocessing (\cref{lem:streamlining}) ensures that every leaf node has exactly one acceptable record, namely $(\omega_0, \emptyset, \emptyset, \emptyset)$ where~$\omega_0$ denotes the empty word, so we can use the set $\mathcal{R}_{\ell} = \{(\omega_0, \emptyset, \emptyset, \emptyset)\}$ as the base case of our algorithm for every leaf $\ell$ of $T_f$.
 The primary task in this section is thus to provide a fixed-parameter procedure for computing representatives in a leaf-to-root fashion: 
 
\medskip

\begin{theorem}\label{thm:dp-algorithm}
There exists an algorithm $\mathcal{A}$ 
that takes as input a non-leaf strip $s$ of $T_f$ and a representative $\mathcal{R}_t$ of $t$ for each child $t$ of $s$, runs in time $k^{\mathcal{O}(k^3)} n^{\mathcal{O}(1)}$, and outputs a representative $\mathcal{R}_s$ of $s$.
\end{theorem}

We split the proof of this theorem into three subsections: one dedicated to the case where~$s$ is non-polar, one where~$s$ is polar and its two children lie on the same side, and one where $s$ is polar and its two children lie on opposite sides (i.e., bottom and top).
We remark that a polar strip might have exactly one child; in this case we treat the non-existing second child as a ``dummy'' child $t$ with the representative $\mathcal{R}_{t} = \{(\omega_0, \emptyset, \emptyset, \emptyset)\}$ where $\omega_0$ denotes the empty word, and then run the procedure for the case with two children.

\subsection{Non-Polar Strips}
In this section we consider a non-polar strip $s$ in $T_f$ and let $t$ denote its unique child.
As previously, without loss of generality we assume that the gate shared by $s$ and its parent is a top-gate of $s$.
Then by definition of non-polar strips, the gate shared by $t$ and $s$ is a bottom-gate of $s$.
We start with the following simple lemma.
\begin{lemma}\label{lem:outer-bundles-non-polar}
    For every $A \in V_M$, the sets of outer $A$-bundles of $s$ and $t$ coincide, i.e., it holds that $\Delta_A(f \setminus f^s) = \Delta_A(f \setminus f^t)$.
\end{lemma}

\begin{proof}
    It holds that $f^s = f^t \cup s$ since $t$ is the unique child of $s$.
    Recall that by definition of non-polar strips, $s$ does not have any vertices adjacent to $V_M$ on its boundary.
    Thus, if we restrict the walks $W(f \setminus f^s)$ and $W(f \setminus f^t)$ to the vertices adjacent to $V_M$, these two walks coincide.
    Furthermore, for a vertex $v$ adjacent to $V_M$, there is a downward curve from $v$ to the gate of $t$ inside $f^t$ if and only if there is a downward curve from $v$ to the gate of $s$ inside~$f^s$:
    One direction is trivial, while the other holds because the parent of $s$ lies above~$s$ which, in turn, lies above $t$.
    Thus, a vertex adjacent to $V_M$ is interesting with respect to~$f \setminus f^s$ if and only if it is interesting with respect to $f \setminus f^t$.
    So the sets of outer~$A$-bundles~$\Delta_A(f \setminus f^s)$ and $\Delta_A(f \setminus f^t)$ are equal.
\end{proof}

Now we are ready to describe the dynamic programming procedure for a non-polar strip~$s$, after which we establish correctness.
We recall that by construction every vertex on the boundary of $s$ is a ``subdivision-vertex'' that appeared in \cref{sec:strips}, and therefore no vertex on the boundary of $s$ is adjacent to missing vertices.

We start with the empty set $\TT_s$.
First of all, we add every record in  $\mathcal{R}_t$ to $\mathcal{R}_s$. 
This corresponds to partial solutions in which no missing vertex is placed inside $s$.

Next, we compute the set $D_s$ of all acceptable records of $s$ (see \cref{thm:relevant-records}).
Now we iterate through all records $R_t = (\omega_t, X_t, E^1_t, \Phi_t) \in \mathcal{R}_t$ and all acceptable records $R= (\omega, X, E^1, \Phi) \in D_s$ for $s$
to check if $R$ can be ``obtained from'' $R_t$ by placing one vertex inside~$s$ as follows.
If at least one of the upcoming checks fails (we also call them \emph{algorithmic checks}), we discard this pair and continue with the next one; Otherwise, in the very end, the record will be added to $\mathcal{R}_s$.

First of all, we check that $X_t \subseteq X$ and $|X \setminus X_t| = 1$ holds.
Now let $A$ denote the unique vertex in $X \setminus X_t$.
If we have $\Phi^A \neq \Delta_A(f \setminus f^s)$, i.e., $\Phi^A$ is missing at least one outer $A$-bundle of $s$, we discard this candidate.
For every $B \in X_t$ we also check that $\Phi_t^B = \Phi^B$ holds 
-- this is because 
by \cref{lem:outer-bundles-non-polar} we have $\Delta_B(f \setminus f^t) = \Delta_B(f \setminus f^s)$.
Further, we check that $E_t^1 = E_t$ holds.

Next, we check that all edges incoming into $A$ are routed as follows.
First, we check that~$\{vA \mid vA \in E_M^1\} \subseteq E_t^1$. 
And second, for every edge $BA \in E_M^2$, it holds that $B \in X_t$.

Similarly, we carry out some checks for the  edges outgoing from $A$ as follows. 
First, we check that $A$ has at least one outgoing edge in $E_M^1$ if and only if the letter $A$ occurs in $\omega$ at least once.
And second, for every vertex $B \in V_M$ the word $\omega$ contains the letter $\{A, B\}$ if and only if $AB \in E_M^2$ and $B \notin X$ holds.

Next, we carry out the checks for the gate words.
First, we check that in each of $\omega$ and~$\omega_t$ the letters containing $A$ occur consecutively. 
If $\omega$ contains at least one letter containing $A$, we define the unique words $\omega^1, \omega^A, \omega^2$ such that:
\begin{itemize}
    \item $\omega = \omega^1 \omega^A \omega^2$,
    \item every letter in $\omega^A$ contains $A$,
    \item no letter in any of $\omega^1$ and $\omega^2$ contains $A$.
\end{itemize}
Similarly, if $\omega^t$ contains at least one letter containing $A$, we define the unique words $\omega^1_t, \omega^A_t, \omega^2_t$ such that:
\begin{itemize}
    \item $\omega_t = \omega^1_t \omega^A_t \omega^2_t$,
    \item every letter in $\omega^A_t$ contains $A$,
    \item no letter in any of $\omega^1_t$ and $\omega^2_t$ contains $A$.
\end{itemize}
Note that at least one of $\omega^A$ and $\omega_t^A$ is non-empty since otherwise $A$ would be isolated in the connected graph $G$. 
Now we distinguish between the following three cases.

\textbf{Case 1: both $\omega^A_t$ and $\omega^A$ are non-empty.} In this case, we check that $\omega^1_t = \omega^1$ and~$\omega^2_t = \omega^2$ hold.

\textbf{Case 2: $\omega^A_t$ is non-empty and $\omega_A$ is empty (i.e., $A$ only has incoming edges).}
We check that the following two conditions are satisfied. 
If the last letter of $\omega^1_t$ and the first letter of $\omega^2_t$ do not coincide, then $\omega$ is equal to $\omega^1_t \omega^2_t$.
And if these two letters coincide, then~$\omega$ is equal to the concatenation of $\omega^1_t$ and $\omega^2_t$ from which we remove the first letter of~$\omega^2_t$.

\textbf{Case 3: $\omega^A_t$ is empty and $\omega_A$ is non-empty (i.e., $A$ only has outgoing edges).}
We check that the following two conditions are satisfied.  
If the last letter of $\omega^1$ and the first letter of $\omega^2$ do not coincide, then $\omega_t$ is equal to $\omega^1_t \omega^2_t$. And if these two letters coincide, then~$\omega_t$ is equal to the concatenation of~$\omega^1$ and~$\omega^2$ from which we remove the first letter of~$\omega^2$.

If none of the previous checks failed, the record $R$ is added to $\mathcal{R}_t$ and we continue with the next pair $R_t, R$. 
This concludes the description of the algorithm; we now prove its correctness, i.e., show that the computed table is indeed a representative of $s$.
This is implied by the following three claims.

\begin{claim}[Non-Polar Claim 1]
    Every record in $\TT_s$ is acceptable.
\end{claim}

\begin{proof}
    We first show that every record, say $R = (X, \omega, E^1, \Phi)$, in $\TT_t$ is acceptable for $s$.
    Recall that $\TT_t$ is a representative of $t$ so $R$ is acceptable for $t$.
    It is not difficult to verify that~$R$ is then also an acceptable record for $s$ because $s$ is a non-polar strip; we shortly argue this below.
    Since the boundary of $s$ contains no neighbors of $V_M$, the sets of internal vertices of~$f^s$ and~$f^t$ adjacent to $V_M$ coincide.
    As argued in the proof of \cref{lem:outer-bundles-non-polar}, the sets of vertices interesting with respect to $f \setminus f^s$ and $f \setminus f^t$ adjacent to $V_M$ coincide as well.
    Symmetrically, the sets of vertices interesting with respect to $f^s$ and $f^t$ adjacent to~$V_M$ coincide as well---this is because if we restrict the walks $W(f^s)$ and $W(f^t)$ to vertices adjacent to $V_M$, these two walks coincide.
    This also implies that for any $A \neq B \in V_M$, we have $\enclosed^B_A(f^s) = \enclosed^B_A(f^t)$ and $\enclosed^B_A(f \setminus f^s) = \enclosed^B_A(f \setminus f^t)$.
    Moreover, for every $A \in V_M$, we have $\Delta_A(f^s) = \Delta(f^t)$ and by \cref{lem:outer-bundles-non-polar} we have $\Delta_A(f \setminus f^s) = \Delta(f \setminus f^t)$.
    These properties imply that $R$ is indeed acceptable for $s$ as well.
    This concludes the proof of this claim since for every $R \notin \TT_t$ added to $\TT_s$, the algorithm ensures that $R$ is acceptable for $s$ by only considering acceptable ``candidates'' $R \in D_s$. 
\end{proof}

\begin{claim}[Non-Polar Claim 2]
    For every everywhere clean solution $S$, the table $\mathcal{R}_s$ contains the record of $S$ at~$s$.
\end{claim}

\begin{proof}
    Consider an arbitrary everywhere clean solution $S$.
    For $x \in \{s, t\}$ let $S^s$ be the partial solution of $x$ obtained from $S$ by restricting it to $f^x$ and let $R_x = (\omega_x, X_x, E^1_x, \Phi_x)$ denote the record of $S^x$.
    And let $R = R_s = (\omega, X, E^1, \Phi)$ for simplicity.
    We will distinguish whether $S$ places some vertex inside $s$ or not (by definition of a solution, at most one vertex is placed inside $s$).

    We start with the simpler case where $S$ does not place any vertex inside $S$.
    Recall that the boundary of $s$ does not contain any vertex adjacent to $V_M$.
    Thus, the following properties apply.
    First, the sets of placed vertices as well as the sets of edges routed by $S^s$ and $S^t$ coincide.
    Also, the labels of curves ending on the gate of $s$ in $S^s$ are exactly the same as those ending on the gate of $t$ in $S^t$. 
    Furthermore, these curves end on these gates in exactly the same ordering---otherwise, two of them would cross inside $s$.
    So the gate words of~$S^s$ and~$S^t$ coincide.
    Also, if for an outer $A$-bundle, say $\gamma$, of $s$ and a vertex $v \in \gamma$, there is a curve in~$S^t$ labeled $Av$ ending on the gate of $t$ if and only if there is a curve in $S^s$ labeled $Av$ ending on the gate of $s$.
    Now recall that we have $\Delta_A(f \setminus f^s) = \Delta_A(f \setminus f^t)$ by \cref{lem:outer-bundles-non-polar}.
    Altogether, this implies that $R = R_t$. 
    Since $\mathcal{R}_t$ is representative of $t$, we then have $R = R_t \in \mathcal{R}_t$.
    Finally, recall that at the beginning of the algorithm, we add all records from $\mathcal{R}_t$ to $\mathcal{R}_s$.
    Thus, we have $R = R_t \in \mathcal{R}_s$ as desired.

    It remains to resolve the case where $S$ places exactly one vertex, say $A$, inside $s$.
    We will now show that 
    in the iteration corresponding to the pair $R_t$ and $R$, the record $R$ is added to~$\mathcal{R}_s$. 
    First, we define $X_t = X \setminus \{A\}$.
    Recall that $X$ is precisely the set of vertices placed by~$S^s$ inside $f^s$, only the vertex $A$ is placed inside $s$, and we have $f^s = f^s \cup s$.
    Therefore,~$X_t$ is the set of vertices placed by $S^t$ inside $f^t$.
    Next, we define $\omega_t$ to be the gate word of $S^t$.

    Now let $E'$ be the set of edges routed by $S^s$.
    Then since $R$ is a record of $S^s$, we have 
    \[
        E^1 = E' \setminus \{Bv \mid B \in X, v \in \Phi^B\}.
    \]
    We let $(\Phi_t^A)_{B \in X_t}$ be such that $\Phi_t^B = \Phi^B$ for every $B \in X_t$.
    Note that we have $\Phi^t_B \subseteq \Delta_B(f \setminus f^t)$ because by \cref{lem:outer-bundles-non-polar}, $\Delta_B(f \setminus f^s) = \Delta_B(f \setminus f^t)$ holds for every $B \in V_M$.
    We now define~$R'_t = (X_t, \omega_t, E^1, \Phi_t)$ and we will show that $R'_t = R_t$ holds, i.e., $R'_t$ is the record of $S^t$.
    For this, we show that $E^1$ and $\Phi_t$ satisfy the definition of a record of $S^t$.
    First, since the boundary of $s$ does not contain any neighbor of $V_M$, the sets of edges routed by $S^s$ and $S^t$ coincide, i.e., the set of edges routed by $S^t$ is equal to $E'$ as well.
    We claim that
    \begin{align*}
        E^1 &= E' \setminus \{Bv \mid B \in X, v \in \Phi^B\} \\
            &= E' \setminus \{Bv \mid B \in X_t = X \setminus \{A\}, v \in \Phi^B\} 
    \end{align*}
    holds.
    Observe that $S^t$ cannot route any edge $Av \in E_M^1$ (due to $A \notin X_t)$.
    Therefore, we have $Av \notin E'$ for any $v \in V(H)$, and therefore the above equality indeed holds.
    To show that $R'_t$ is a record of $S^t$, it remains to show that for every $B \in X_t$ and every $\gamma \in \Delta_B(f \setminus f^t)$, there exists a curve in $S^t$ labeled $Bv$ for some $v \in \gamma$ and ending on the gate of $t$ if and only if $\gamma \in \Phi^B$ holds.
    But this holds since for vertices $B \in X_t$ and $v \in V(H)$ there is a curve in~$S^t$ labeled $Bv$ and ending on the gate of $t$ if and only if there is a curve in $S^s$ labeled $Bv$ and ending on the gate of $s$---this is because no such curve can end in $s$.
    So $R'_t$ is indeed a record of $S^t$, i.e., $R_t = R'_t$.
    Since $\TT_t$ is a representative of $t$, we have $R_t \in \TT_t$. 

    So it remains to show that in the iteration of the algorithm in which $R_t$ and $R$ are considered, no algorithmic check fails.
    Some of the checks are satisfied just by definition of~$R_t$ so we only focus on the remainder.
    Consider an arbitrary outer $A$-bundle $\gamma \in \Delta_A(f \setminus f^s)$ of~$s$ and an arbitrary vertex $v \in \gamma$.
    The edge $Av$ cannot be finished by $S^s$ because any curve starting in $A$, ending in some pre-drawn neighbor of $A$, and fully contained in $f^s$ would necessarily need to cross the gate between $s$ and $t$ from top to bottom, i.e., such a curve is not upward.
    Thus there exists a curve in $S^s$ labeled $Av$ and ending on the gate of $s$.
    Since~$R$ is a record of $S^s$, it then holds that $\gamma \in \Phi^A$.
    Since $\gamma$ was chosen arbitrarily, we obtain~$\Phi^A = \Delta_A(f \setminus f^s)$ and this check is satisfied.

    Next, we analyze the check $\{vA \mid vA \in E_M^1\} \subseteq E_t^1$.
    Since $S^s$ is a realizable partial solution, all edges incoming into $A$ are routed by $S^s$, i.e., we have $\{vA \mid vA \in E_M^1\} \subseteq E'$.
    Recall that $E'$ is also the set of edges routed by $S^t$, by definition of the record $R_t$ of $S^t$ we have
    \[
        E^1_t = E' \setminus \{Bv \mid B \in X_t = X \setminus \{A\}, v \in \Phi^B\}.
    \]
    Since in $E' \setminus E^1_t$ we only have edges with the tail being a missing vertex, we then also have~$\{vA \mid vA \in E_M^1\} \subseteq E^1_t$---and this check is satisfied.
    Further, since $S^s$ is a realizable partial solution, all edges $BA \in E_M^2$ are finished by $S^s$ and in particular, we have $B \in X$.
    Since no missing vertex other than $A$ is placed in $s$, we have $B \in X_t$ as desired.

    Further, since $S^s$ is a realizable partial solution, for every edge $Av \in E_M^1$ there is a curve labeled $Av$ in $S^s$ ending on the gate of $s$ (as argued before it cannot be finished), and therefore $\omega$ contains letter $A$ if there is an edge outgoing from $A$ in $E^1_M$.
    Vice versa,~$\omega$ can only contain the letter $A$ is there is a curve of type $A$ ending on the gate of $s$---by definition of a partial solution (\cref{def:ps-placed-vertices}), this curve then starts in the vertex $A$ and therefore has label~$Av$ for some $Av \in E_M^1$.
    Similarly, for every edge $AB \in E_M^2$, this edge cannot be finished by~$S^s$ since a curve starting in $A$, ending in $B$, and fully contained in $f^s$ would necessarily cross the gate shared by $s$ and $t$ from top to bottom---and therefore, it cannot be upward.
    Thus,~$S^s$ contains a curve with label $AB$ ending on the gate of $s$.
    So the word~$\omega$ contains the letter~$\{A, B\}$.
    Vice versa, $\omega$ can only contain the letter~$\{A, B\}$ is there is a curve of type~$\{A,B\}$ ending on the gate of $s$.
    By definition of a partial solution (\cref{def:ps-E-M-2}) this curve then starts in the vertex $A$ and has label $AB$ for some $AB \in E_M^2$ for some $B \notin X$.

    Finally, we check that the words $\omega$ and $\omega_t$ satisfy the last check of the algorithm.
    For this we take a closer look at how the restriction of the solution $S$ to the strip $s$ looks like.
    The only endpoint of the edges in $E_M$ on the boundary of and inside $f^s$ is the vertex $A$. 
    First, suppose that the letters containing $A$ in $\omega$ are not consecutive, i.e., in $S^s$ there exist two curves $\phi_1$ and $\phi_2$ both of some type containing $A$ and a curve $\phi$ of some type not containing~$A$.
    The curve $\phi$ starts therefore in $f^t$.
    By definition of a partial solution, the curves $\phi_1$ and~$\phi_2$ start in $A$, and therefore these curves together with the part of the gate form a connected region, say $\psi$, that does not contain any point of the gate shared by $s$ and $t$. 
    Thus, the upward curve $\phi$ enters its endpoint on the gate of $s$ inside $\psi$. 
    But then it has to cross one of the curves $\phi_1$ and $\phi_2$ to enter this region (see \cref{fig:a-consecutive}~(a) for an illustration)---a contradiction.
    A symmetric argument applied to $S \cap (f \setminus f^t)$ shows that the letters containing $A$ are also consecutive in $\omega_t$.
    
    \begin{figure}[t]
        \centering
        \includegraphics{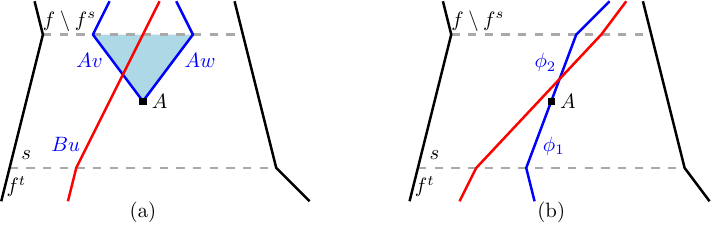}
        \caption{(a) Sketch illustrating that the letters containing $A$ are consecutive in $\omega$. 
        In black, the pre-drawn part $\Gamma(H)$ and in blue the missing vertices and edges.
        Vertex $A$ is the only vertex placed inside $s$, the region $\psi$ (in light blue) is bounded by two curves $\phi_1$ and $\phi_2$ whose type contains $A$. If there would be a curve of a type not containing $A$ whose endpoint one the gate of $s$ lies between the endpoints of $\phi_1$ and $\phi_2$, then there would be a crossing. 
        (b) Arguing that no straight-line segment in $s$ can have the lower endpoint in $\omega^1_t$ and the higher endpoint in $\omega^2$.}
        \label{fig:a-consecutive}
    \end{figure}

    Further, as argued above, the curves ending on the gate of $s$ and not starting in $A$, end on this gate in the same ordering as they crossed the gate shared by $s$ and $t$.
    This already implies that if $A$ has only incoming or only outgoing edges in $G$, the check of the algorithm goes through.
    Otherwise, pick an arbitrary curve $\phi_1$ representing an edge incoming to $A$ in~$S^s$, and pick another arbitrary curve $\phi_2$ representing an edge outgoing from $A$ in~$S^s$; note that, $\phi_2$ must end at a gate of $s$. 
    Then the curve obtained from $\phi_1$ and $\phi_2$ by gluing them at the vertex $A$ partitions the strip $s$ into two connected regions: one to the left and 
    one to the right of this curve.
    Thus, if some curve, say $\phi$, of $S^s$ whose label's type does not contain $A$ crosses the gate shared by $s$ and $t$ before $\phi_1$ and ends on the gate of $s$ after $\phi_2$ (or vice versa), it necessarily crosses $\phi_1$ or $\phi_2$---contradicting the definition of a solution (see \cref{fig:a-consecutive}~(b) for an illustration).
    Hence every curve that crossed the gate shared by $s$ and $t$ before (resp.\ after) all curves whose type contains $A$ also crosses the gate of $s$ before (resp.\ after) all curves whose type contains $A$.
    Therefore, the last check does not fail as well, and~$R$ is added to $\mathcal{R}_t$ as desired.
\end{proof}

\begin{claim}[Non-Polar Claim 3]
    For every record $R$ in $\mathcal{R}_s$ there exists a clean partial solution with this record.
\end{claim}

\begin{proof}
    Consider a record $R = (X, \omega, E^1, \Phi)$ added by the algorithm; we establish that there exists a clean partial solution of $s$ with this $R$ as its record.
    First, consider the simpler case that $R \in \mathcal{R}_t$ holds (the other case will be treated later).
    Since $\TT_t$ is a representative of $t$, there exists a clean partial solution $P$ of $t$ such that $R$ is the record of $t$.
    Let $Q$ be a labeled drawing obtained from $P$ as follows. 
    Let $\phi_1, \dots, \phi_r$ (for some $r \geq 0$) be all curves in $P$ ending on the gate of $t$ in the left-to-right ordering.
    We choose pairwise distinct points $p_1, \dots, p_r$ on the gate of $s$ such that $p_i$ lies to the left of $p_{i+1}$ for every $i \in [r-1]$ (see \cref{fig:polar-no-vertex} for an illustration).
    And we extend $\phi_i$ by connecting its gate endpoint with $p_i$ by a straight-line segment, each curve preserves its label.
    By construction, no crossings are created this way.
    It is straightforward to verify that $Q$ satisfies all properties of a partial solution because $P$ does.
    It remains to show that $R$ is the record of $Q$.
    Clearly, the set of vertices placed by $Q$ is still $X$.
    Further, we did not change the set, labels, or the ordering of the curves ending on the gate, and therefore $\omega$ is still the gate word of $Q$.
    Let $E'$ denote the set of edges routed by $P$.
    First, since $R$ is the record of $P$ we have
    \[
        E^1 = E' \setminus \{Av \mid A \in X, \exists \gamma \in \Phi^A \colon v \in \gamma \}.
    \]
    Second, by construction, the set of edges routed by $Q$ is still $E'$.
    Thus, the set $E^1$ also satisfies the definition of the record of $Q$ (see \cref{def:record-e1}).
    Finally, there is a curve in $P$ ending on the gate of $t$ and having label $e$ if and only if there is a curve in $Q$ ending on the gate of $s$ and having label $e$.
    Thus, the fact that $P$ and $R$ satisfy the property \cref{def:record-phi} in the definition of a record together with $\Delta_A(f \setminus f^s) = \Delta_A(f \setminus f^t)$ for every $A \in V_M$ (see \cref{lem:outer-bundles-non-polar}) implies that $Q$ and $R$ also satisfy this property.
    Hence, $Q$ is indeed a partial solution of $s$ with record~$R$.

    \begin{figure}[t]
        \centering
        \includegraphics{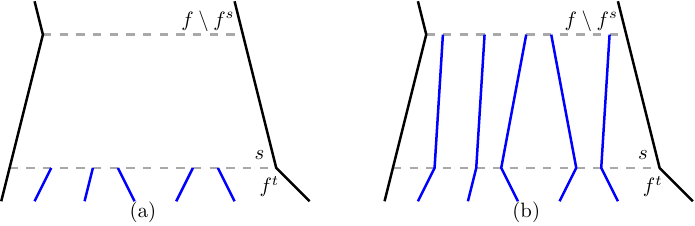}
        \caption{(a) Partial solution $P$ of $t$.
        (b) Extending every curve by a straight-line segment to obtain a partial solution $Q$ that places no vertices inside $s$.}
        \label{fig:polar-no-vertex}
    \end{figure}
    
    It remains to show that $Q$ is clean.
    For edges incident with the vertices $A \notin X$ this is implied by the fact that the record $R$ is acceptable (see \cref{def:acceptable-inner-bundles}) together with the fact that the restrictions of $E'$ and $E^1$ to edges whose missing endpoint is in $V_M \setminus X$ coincide.
    So now consider a vertex $A \in V_M \setminus X$ and recall that we have $\Delta_A(f \setminus f^s) = \Delta_A(f \setminus f^t)$ by \cref{lem:outer-bundles-non-polar}. 
    The partial solution $P$ of $t$ is clean, thus for every outer bundle $\gamma \in \Delta_A(f \setminus f^s) = \Delta_A(f \setminus f^t)$ and every pair $u, v \in \gamma$ either $P$ has a curve labeled $Av$ and a curve labeled $Au$ ending on the gate, or it has no curve labeled $Av$ and no curve labeled $Au$ ending on the gate.
    By construction of $P$ this then also applies to $Q$.
    So $Q$ is indeed a clean partial solution with record $R$.
    
    Now we consider the remaining case that $R \notin \TT_t$, i.e., $R$ was added to $\TT_s$ in an iteration for $R$ and $R' = (X_t, \omega_t, E^1_t, \Phi_t)$ for some $R' \in \TT_t$.
    Again, since $\TT_t$ is a representative of~$t$, there exists a clean partial solution $P$ of $t$ such that $R$ is a record of $t$.
    We will first construct a labeled drawing $Q$, then argue why it is a clean partial solution, and finally show that~$R$ is a record of $Q$.
    First, one check of the algorithm ensures that we have $X_t \subseteq X$ and~$X \setminus X_t = \{A\}$ for some $A \in V_M$.
    
    We now construct the labeled drawing $Q$ inside $f^s$ as follows.
    First of all, it coincides with $P$ inside $f^t$ and we only add a part of the drawing inside the strip $s$ as follows.
    We start by drawing the vertex $A$ in an arbitrary non-boundary point inside $s$.

    Next, we define the integer $r$ as the maximum of the values $|\{Av \mid Av \in E_M^1\}|$ and the number of occurrences of $A$ in $\omega$.
    Now we add $r$ straight-line segments from $A$ to the top-gate of $s$ so that the gate endpoints of these straight-line segments are pairwise distinct (see \cref{fig:polar-a-outgoing} for an illustration).
    We label these $r$ curves with the elements of $\{Av \mid Av \in E_M^1\}$ in such a way that every element of $\{Av \mid Av \in E_M^1\}$ occurs at least once as a label (note that the choice of $r$ ensures that this is possible).
    Let $p_1, \dots, p_r$ be the gate endpoints of these curves in the left-to-right ordering.
    
    Recall that one of the checks of the algorithm ensures that the letters containing $A$ occur consecutively in $\omega$, and the subword containing all these letters is denoted by $\omega^A$.
    Let $p_0$ denote an arbitrary but fixed point on the gate of $s$ to the left of $p_1$ if $r \geq 0$; and otherwise, let~$p_0$
    be just an arbitrary point
    on the gate of $s$.
    Now let $z$ denote the number of occurrences of $A$ in $\omega$ (recall that we have $r \geq z$ by definition of $r$).
    Furthermore, let~$\alpha_0, \dots, \alpha_z \in {V_M \choose 2}^*$ be the words such that $\alpha_0 A \alpha_1 A \alpha_2 \dots A \alpha_z = \omega^A$ (we mean $\alpha_0 = \omega^A$ in case $z = 0$).
    For~$i \in [z - 1]_0$ we proceed as follows (the value $i = z$ will be treated separately) and refer to \cref{fig:polar-a-outgoing} for an illustration. 
    Let $q_i$ be an integer (possibly $q_i = 0$ for~$i = 0$ or $i = z$) and let~$B_{i, 1}, \dots, B_{i, q_i}$ be such that $\alpha_i = \{A, B_{i, 1}\} \dots \{A, B_{i, q_i}\}$ holds.
    Note that those are well-defined since the word $\omega^A$ contains only letters containing $A$ and~$A$ occurs exactly $z$ times in~$\omega^A$.
    Now we pick $q_i$ points $p_{i, 1}, \dots, p_{i, q_i}$ on the gate of $s$ such that all of these lie to the right of $p_i$ and to the left of $p_{i+1}$, and furthermore, for every~$j \in [q_i - 1]$ the point $p_{i, j}$ lies to the left of $p_{i, j+1}$.
    After that we add a straight-line segment~$A p_{i,j}$ labeled with $A B_{i,j}$ for every $j \in [q_i]$.
    Similarly, if $\alpha_z$ is non-empty, we pick $q_z$ points~$p_{z, 1}, \dots, p_{z, q_z}$ on the gate of $s$ such that all of them lie to the right of $p_r$ and furthermore, for every $j \in [q_z - 1]$ the point $p_{z, j}$ lies to the left of $p_{z, j+1}$.
    Afterwards we add a straight-line segment $A p_{z,j}$ labeled with $A B_{z,j}$ for every $j \in [q_z]$.
    Observe that the labeled straight-line segments added in this procedure yield precisely the gate word $\omega^A$.
    Also note that this still did not create any crossings since only straight-line segments with one endpoint $A$ and the other on the boundary of $s$ have been added so far.

\begin{figure}[t]
        \centering
        \includegraphics{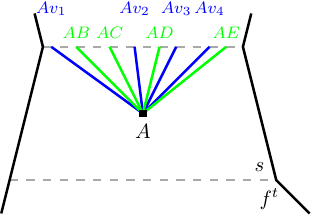}
        \caption{An example for the drawing of edges outgoing from $A$. The edges in $E_1^M$ outgoing from $A$ are $Av_1, Av_2, Av_3, Av_4$ and the word $\omega^A$ is equal to $A \{A, B\} \{A, C\} A \{A, D\} A \{A, E\}$. Here we have $r = 4$, $z = 3$, and the word $\alpha_0$ is empty. The labels are sketched above the corresponding straight-line segments.}
        \label{fig:polar-a-outgoing}
    \end{figure}

    Now for every curve of $P$ that ends on the gate of $s$ and whose label has a type containing~$A$, we add a straight-line segment from the gate endpoint of this curve to $A$.
    Note that this only connects $A$ to some points on the boundary of $s$ via straight-line segments and therefore still no crossings are produced.
    
    If $\omega^A_t$ is non-empty, let $L_b$ and $R_b$ be two points on bottom gate of $s$ such that $L_b$ lies to the left of $R_b$ and all endpoints of straight-line segments added so far on this gate lie between~$L_b$ and $R_b$.
    Similarly if $\omega^A$ is non-empty, let $L_t$ and $R_t$ be two points on top-gate of~$s$ such that $L_t$ lies to the left of $R_t$ and all endpoints of straight-line segments added so far on this gate lie between $L_t$ and $R_t$.
    
    To conclude the construction of $Q$ we describe how to extend the curves ending on the gate of $t$ such that the type of their labels do not contain $A$.
    For this construction we make the same case distinction as in the algorithm.
    Let $\phi_1, \dots, \phi_h$ be all curves in $P$ that end on the gate of $t$ in the left-to-right order of their gate endpoints.

    \subparagraph{Case 1: both $\omega_t^A$ and $\omega^A$ are non-empty} (see \cref{fig:polar-3-cases}~(a) for an illustration)
    We may assume that the vertex $A$ is placed inside the trapezoid, denoted by $\psi$, spanned by~$L_b$,~$R_b$,~$R_t$, and~$L_t$---otherwise, we may shift the vertex $A$ inside this trapezoid and replace each straight-line segment by the one with the same gate endpoint and the other being the new position of~$A$. 
    Thus, all straight-line segments added so far also lie inside $\psi$.
    By definition of this case there exist $1 \leq p \leq q \leq h$ such that $\phi_p, \dots, \phi_q$ are precisely the curves such that the type of their labels contains $A$. 
    This is because $\omega_t$ is the gate word of $P$ and the letters containing~$A$ occur consecutively in $\omega_t$ as ensured by one of the algorithmic checks. 
    
    Now we choose the points $b_1, \dots, b_{p-1}$ on the gate of $s$ all to the left of $L_t$ such that $b_i$ lies to the left of $b_{i+1}$ for every $i \in [p-2]$.
    Now for every $i \in [p-1]$ we add a straight-line segment between $b_i$ and the gate endpoint of $\phi_i$.
    Note that this does not produce any crossings: 
    First, the new segments do not cross the old ones since the new ones lie to the left of $\psi$.
    Second, the choice of $b_1, \dots, b_{p-1}$ ensures that the new segments are pairwise non-crossing.
    Note that the gate word yield by curves ending on the gate of $s$ so far is now $\omega_t^1 \omega^A$.
    Similarly, we choose the points $b_{q+1}, \dots, b_h$ on the gate of $s$ all to the right of $R_t$ and such that $b_i$ is to the left of $b_{i+1}$ for every $i \in [q+1, h-1]$.
    And for every $i \in [q+1, h]$ we now add a straight-line segment between $b_i$ and the gate endpoint of $\phi_i$.
    Again, observe that this does not create any crossings: 
    First, the new segments do not cross the old ones since the new ones lie to the right of $\psi$.
    Second, the choice of $b_{q+1}, \dots, b_h$ ensures that the new segments are pairwise non-crossing.
    This concludes the construction of $Q$ in this case.
    Observe that the gate word of $Q$ is given by $\omega_t^1 \omega^A \omega_t^2$.
    One of the checks of the algorithm ensures that we have $\omega_t^1 = \omega^1$ and $\omega_t^2 = \omega^2$ and therefore the gate word of $Q$ is $\omega^1 \omega^A \omega^2 = \omega$.
    
    \begin{figure}[t]
        \centering
        \includegraphics{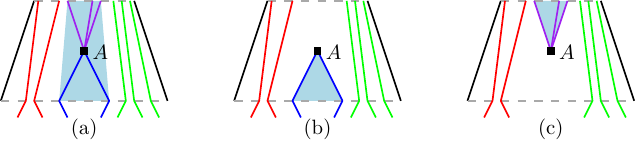}
        \caption{The drawing $Q$ restricted to the strip $s$. The curves corresponding to the subword $\omega_t^1$ resp.\ $\omega^A_t$ resp.\ $\omega^2_t$ are red resp.\ blue resp.\ green. The curves starting in $A$ are purple, the region $\psi$ is light blue. (a) Both $\omega_t^A$ and $\omega^A$ are non-empty, (b) $\omega^A$ is empty, (c) $\omega^A_t$ is empty.}
        \label{fig:polar-3-cases}
    \end{figure}
    
    \subparagraph{Case 2: $\omega^A_t$ is non-empty and $\omega^A$ is empty} (see \cref{fig:polar-3-cases}~(b) for an illustration)
    In this case we let $\psi$ denote the triangle spanned by $L_b$, $R_b$, and $A$.
    We may assume that the $x$-coordinate of $A$ is between the $x$-coordinates of $L_b$ and $R_b$ by an argument analogous to Case 1.
    Furthermore, we may assume that the $y$-coordinate of $A$ is $\varepsilon$-close to the bottom gate of $s$ for a sufficiently small $\varepsilon$.
    Similarly as in the previous case, there exist $1 \leq p \leq q \leq t$ such that $\phi_p, \dots, \phi_q$ are precisely the curves such that the type of their labels contains $A$. 
    We choose the points $b_1, \dots, b_{p-1}$ on the gate of $s$ such that $b_i$ lies to the left of $b_{i+1}$ for every $i \in [p-2]$.
    We also choose the points $b_{q+1}, b_{q+2}, \dots, b_h$ on the gate of $s$ all lying to the right of $b_{p-1}$ such that $b_i$ lies to the left of $b_{i+1}$.
    Now for every $i \in [1, p-1] \cup [q+1,h]$ we add a straight-line segment between $b_i$ and the gate endpoint of $\phi_i$.
    The sufficiently small choice of $\varepsilon$ ensures that these new segments do not cross $\psi$ and therefore they cross none of the existing segments.
    The choice of $b_1, \dots, b_{p-1}$ and $b_{q+1}, b_{q+2}, \dots, b_h$ ensures that the new segments are pairwise non-crossing.
    This concludes the construction of $Q$ in this case.
    Observe that the gate word of $Q$ is obtained from $\omega_t^1 \omega_t^2$ by, if the last letter of $\omega_t^1$ and the first letter of $\omega_t^2$ coincide, removing the first letter of $\omega_t^2$.
    A check of the algorithm then ensures that this gate word is then equal to $\omega$.
    
    \subparagraph{Case 3: $\omega^A_t$ is empty and $\omega_A$ is non-empty} (see \cref{fig:polar-3-cases}~(c) for an illustration)
    In this case we let $\psi$ denote the triangle spanned by $L_t$, $R_t$, and $A$.
    We may assume that the $x$-coordinate of $A$ is between the $x$-coordinates of $L_t$ and $R_t$ by an argument analogous to the previous cases.
    Similarly, we may assume that the $y$-coordinate of $A$ is $\varepsilon$-close to the top gate of $s$ for a sufficiently small $\varepsilon$.
    We make a case distinction depending on whether the last letter of $\omega^1$ and the first letter of $\omega^2$ coincide to define a certain value $p$.
    If they do not coincide, the algorithm ensures that we have $\omega_t = \omega^1 \omega^2$.
    We then define the (unique) value $p$ such that $\phi_1, \dots, \phi_p$ yield the gate word $\omega^1$ and the curves $\phi_{p+1}, \dots, \phi_h$ yield the gate word~$\omega^2$.
    If the last letter of $\omega^1$ and the first letter of $\omega^2$ coincide the algorithm ensures that $\omega_t$ is obtained from $\omega^1 \omega^2$ by removing the first letter of $\omega^2$.
    We proceed as follows.
    We may assume that $P$ contains at least two curves corresponding to the last letter of $\omega^1$ in~$\omega_t$ for the following reason.
    Note that this letter is from $V_M$ (otherwise, $\omega$ would contain the same letter from ${V_M \choose 2}$ at least twice, contradicting shortness).
    So if such a curve, say $\phi$, is unique, we can insert a curve $\phi'$ into $P$ that has the same label as $\phi$ and follows $\phi$ very closely from its non-gate endpoint to the gate. 
    The arising object still satisfies all properties of a partial solution of $t$, it is still clean, and it still has the record $R_t$.
    With this assumption in hand, we can now choose a value $p$ (not necessarily unique) in such a way that $\phi_1, \dots, \phi_p$ yield the gate word $\omega^1$ and the curves $\phi_{p+1}, \dots, \phi_h$ yield the gate word $\omega^2$.
    Now given the choice of $p$ in both cases we proceed as follows.
    First, we choose the points $b_1, \dots, b_p$ on the gate of $s$, all to the left of $L_t$ and such that $b_i$ is to the left of $b_{i+1}$ for every $i \in [p-1]$.
    And for every $i \in [p]$ we add we straight-line segment between $b_i$ and the gate endpoint of $\phi_i$.
    Similarly, we choose the points $b_{p+1}, \dots, b_h$ on the gate of $s$, all to the right of $R_t$ and such that $b_i$ is to the left of $b_{i+1}$ for every $i \in [p+1,h-1]$.
    And for every $i \in [p+1,h]$ we add the straight-line segment between $b_i$ and the gate endpoint of $\phi_i$.
    As before, no crossings are produced this way: the choice of $\varepsilon$ ensures that $\psi$ is not crossed while the choice of $b_1, \dots, b_h$ ensures that the new segments are pairwise non-crossing.
    This concludes the construction of~$Q$ in this case.
    Observe that the gate word of $Q$ is given by $\omega^1 \omega^A \omega^2 = \omega$.
    
    \medskip
    
    Let us remark that in the above construction, all curves extended by a straight-line segment and ending on the gate of $s$ preserve their label from $P$.
    Now we argue that $Q$ is indeed a partial solution.
    The properties 1-\ref{def:ps-label} hold since they hold in $P$.
    The property \cref{def:ps-placed-vertices} holds for any vertex $B \in X_t$ since it holds for $P$.
    We now argue that it also holds for the vertex $A$.
    The algorithm checks that $\{vA \mid vA \in E_M^1\} \subseteq E_t^1$ is true.
    Since $A \notin X_t$, this implies that for every $vA \in E_M^1$, the partial solution $P$ contains a curve labeled $vA$ starting in $v$ and ending in some point, say $p$ on the gate of $t$.
    In the construction of $Q$, we connected this curve with $A$ by a straight-line segment and thus, finished this edge.
    On the other hand, the labels of the curves that end on the gate of $s$ in $Q$, by construction, have their tail in $A$ or their head is not $A$.
    Therefore, every edge $e \in \{vA \mid vA \in E_M^1\}$ is finished by $Q$ and there is no curve labeled $e$ ending on the gate of $s$.
    Furthermore, by construction of $Q$, for every edge $Av \in E_M^1$, we added a straight-line segment starting in $A$, ending on the gate of $s$, and labeled $Av$.
    So $Q$ satisfies \cref{def:ps-placed-vertices}.
    
    Similarly, the property \cref{def:ps-E-M-2} is satisfied for the edges $BC \in E_M^2$ with $\{B, C\} \cap A = \emptyset$ by $P$ so by construction of $Q$ it is also satisfied by $Q$.
    Now consider an edge $BA \in E_M^2$ for some $B \in V_M$.
    A check of the algorithm ensures that we have $B \in X_t$.
    Since $P$ is a partial solution, $P$ contains exactly one curve starting in $B$ and having the label $BA$. 
    The construction of $Q$ connects this curve with $A$ using a straight-line segment, and therefore $BA$ is finished.
    Now consider an edge $AB \in E_M^2$ with some $B \in V_M$.
    Note that we have $B \notin X_t$ since that would imply that $P$ finishes the edge $AB$, and therefore also $A \in X_t$ would hold.
    So we have $B \notin X$ due to $A \neq B$.
    Since $R$ is acceptable, $\omega$ (and therefore, also $\omega^A$) contains the letter $AB$.
    So the construction of $Q$ has a straight-line segment labeled $AB$ starting in $A$.
    Furthermore, $\omega$ is short (by acceptability of $R$) so exactly one such segment exists.
    Altogether this implies that \cref{def:ps-E-M-2} is satisfied and $Q$ is indeed a partial solution.
    
    Next we show that $R$ is a record of $Q$.
    We already argued that $X$ and $\omega$ satisfy the desired conditions.
    Now let $E'$ denote the set of edges routed by $P$.
    By construction of $Q$, the set of edges routed by $Q$ is then $E'$ as well.
    We have $E^1_t = E' \setminus \{Bv \mid B \in X_t, v \in \Phi^B\}$ by definition of the record $R$ of $P$.
    The algorithm checks that we have $E^1 = E^1_t$.
    Thus it holds that 
    \[
        E^1 = E^1_t = E' \setminus \{Bv \mid B \in X_t, v \in \Phi^B\} = E' \setminus \{Bv \mid B \in X = X_t \cup \{A\}, v \in \Phi^B\}
    \]
    where the last equality holds since $E'$ does not contain any edge of the form $Av$---this is because $A$ was not placed by $P$.
    Therefore, $E^1$ indeed satisfies \cref{def:record-e1} for $Q$.
    
    Finally, it remains to show that $\Phi$ satisfies the property \cref{def:record-phi}.
    We recall that by \cref{lem:outer-bundles-non-polar}, we have $\Delta_B(f \setminus f^t) = \Delta_B(f \setminus f^s)$ for every $B \in V_M$.
    First, consider $B \in X_t$. 
    Since $R_t$ is the record of $P$, for every $\gamma \in \Delta_A(f \setminus f^t) = \Delta_B(f \setminus f^s)$ and $v \in \gamma$, there is a curve in $P$ labeled $Bv$ if and only if $\gamma \in \Phi^B_t$ holds.
    Now recall that one of the algorithmic checks ensures that we have $\Phi^B_t = \Phi^B$ so by contstruction of $Q$, the above property also holds for $Q$---and \cref{def:record-phi} is satisfied by all $B \in X_t$.
    For the vertex $A$, the algorithm checks that we have~$\Phi^A = \Delta_A(f \setminus f^s)$.
    Furthermore, by construction, for every edge $Av \in E_M^1$, $Q$ contains a straight-line segment labeled $Av$ ending on the gate of $s$.
    So the property \cref{def:record-phi} is satisfied by $A$ as well.
    Therefore, $R$ is indeed a record of $Q$.

    It remains to show that $Q$ is clean.
    For edges incident with a vertex $B \notin X$ this is implied by the fact that the record $R$ is acceptable.
    For edges incident with a vertex $B \in X_t$ this is true because a curve with label $Bv$ for some $v \in V(H)$ ends on the gate of $s$ in $Q$ if and only some curve with the label $Bv$ ends on the gate of $t$ in $P$, $P$ is clean, and we have~$\Delta_A(f \setminus f^t) = \Delta_A(f \setminus f^s)$ for every $V_M$ by \cref{lem:outer-bundles-non-polar}.
    Finally, the partial solution $Q$ has a curve labeled $Av$ for every $Av \in E_M^1$, and therefore edges incident with $A$ also satisfy the cleanness condition.
\end{proof}

\begin{claim}[Non-Polar Claim 4]
    The described procedure computing $\TT_s$ runs in time $k^{\mathcal{O}(k^3)} n^{\mathcal{O}(1)}$.
\end{claim}

\begin{proof}
    First of all, the algorithm computes the set $D_s$ of all acceptable records for $s$---by \cref{thm:relevant-records} this can be done in time~$k^{\bigoh(k^3)} n^{\bigoh(1)}$, and the size of this set is bounded by~$k^{\bigoh(k^3)}$.
    After that it iterates through all records $R_t \in \TT_t$ (whereas the size of this set is bounded by~$k^{\bigoh(k^3)}$ by \cref{cor:size-of-representative}) and all acceptable records $R \in D_s$ for $s$.
    And for a fixed pair~$R,~R_s$ the algorithm carries out a set of checks, each of which can be completed in time polynomial in $n$. 
\end{proof}

\subsection{One Child Above, One Child Below}

Let $s$ be a polar strip with exactly two children one of which lies above $s$ and the other lies below $s$.
Let $t$ denote the child of $s$ lying above $s$ and let $b$ the child of $s$ lying below $s$.
As always, without loss of generality we may assume that the parent, say $p$, of $s$ lies above $s$.
Without loss of generality we may also assume that the gate shared by $s$ and $t$ lies to the left of the gate shared by $s$ and its parent. 
Recall that by \cref{lem:streamlining}, the boundary of $s$ contains at most one vertex, say $u$, adjacent to $V_M$.
Furthermore, it is precisely the vertex separating the gate shared by $s$ and $t$ from the gate shared by $s$ and its parent.
We refer to \cref{{fig:children-above-below-illustration}} for an illustration.

\begin{figure}[t]
    \centering
    \includegraphics{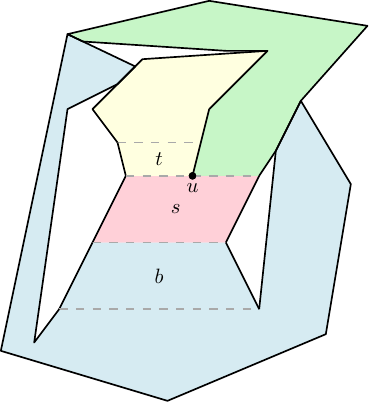}
    \caption{A sketch for a strip $s$ that has one child above and one child below. The strip $s$ is pink. 
    The one-sided regions $f^t$, $f^b$, and $f \setminus f^s$ are sketched in yellow, blue and green, respectively. 
    Note that the boundaries of any two of these one-sided regions share vertices---for edges incident with such vertices, in general, there are options inside which regions they are routed.
    }
    \label{fig:children-above-below-illustration}
\end{figure}

As in the previous case, 
the algorithm first computes the set $D_s$ of all records acceptable for $s$ (see \cref{thm:relevant-records}).
Then it considers all pairs $R_t = (\omega_t, X_t, E^1_t, \Phi_t) \in \mathcal{R}_t$, $R_b = (\omega_b, X_b, E^1_b, \Phi_b) \in \mathcal{R}_b$, and all acceptable records $R = (\omega, X, E^1, \Phi) \in D_s$ to check if we can ``combine'' the records $R_b$ and $R_t$ with a drawing inside the strip $s$ in order to obtain the record $R$.
Below, we describe this procedure for a fixed but arbitrary choice of $R_t$, $R_b$ and~$R$.
If at least one of the upcoming checks fails (we also call them \emph{algorithmic checks}), we discard this triple and continue with the next one; Otherwise, in the very end, the record $R$ will be added to $\mathcal{R}_s$.

We check that $X_b \cap X_t = \emptyset$ and $X_b \cup X_t = X$.
For every vertex $A \in X_t$ we first check that~$\Phi^A = \emptyset$ and $\omega$ does not contain any letter containing $A$.
Second, for every letter $\alpha$ occurring in $\omega_t$, we require $\alpha \subseteq X$.
Third, we check that for every $BA \in E_M^2$, we have~$B \in X$. 
Fourth, we check for every $A \in X_b$, every $\gamma \in \Phi^A$, and every $v \in \gamma$ that the following holds:
if there exists an outer $A$-bundle $\gamma' \in \Delta_A(f \setminus f^b)$ with $v \in \gamma'$, then we have $\gamma' \in \Phi_b^A$. 

Assuming all of these checks succeed, we iterate through all words $\omega_u$ over the alphabet~$V_M \cup {V_M \choose 2}$ with the following properties:
\begin{itemize}
    \item no letter from ${V_M \choose 2} \cup (V_M \setminus X_b)$ occurs in $\omega_u$,
    \item $\omega_u$ is short, 
    \item for every letter $A$ occurring in $\omega_u$, the graph $G$ contains the edge $Au$,
    \item for every letter $A$ occurring in $\omega_u$, if $u \in V^A(f \setminus f^b)$, then we have $u \in \gamma$ for some $\gamma \in \Phi_b^A$.
\end{itemize}
We define the set $E_u^1 = \{Au \mid A \text{ occurs in } \omega_u \}$.
And we check that we have
\[
    E^1 = E^1_b \cup E^1_t \cup E_u^1 \setminus \{Av \mid A \in X_b, \exists \gamma \in \Phi^A \colon v \in \gamma\}.
\]

Finally, we ensure that the following holds. 
Let $\overleftarrow{\omega_t}$ denote the reverse of $\omega_t$---at this point let us remark that since the gate of $t$ is a bottom gate of $t$, the word $\omega_t$ reflects the right-to-left order of the edges crossing the gate of $t$, so we use its reverse $\overleftarrow{\omega_t}$ to ``fit'' the left-to right orderings in $\omega_b$ and $\omega$.
Now let the word $\omega'$ be obtained from the concatenation~$\overleftarrow{\omega_t} \omega_u \omega$ by first, removing the first letter of $\omega_u$ if it coincides with the last letter of $\overleftarrow{\omega_t}$ and second, removing the first letter of $\omega$ from it if it coincides with the last letter of $\omega_u$---less formally speaking we eliminate at most two ``repetitive'' letters so that in the arising word no two consecutive letters coincide.
Then we check that 
\begin{equation}\label{eq:word-check-opposite-sided}
    \omega_b = \omega'
\end{equation} 
holds.
If none of the listed checks failed, we add the record $R$ to $\mathcal{R}_t$, and move on to the next candidate tuple.
We emphasize that for each acceptable record $R$ for $s$ we check \textbf{all} pairs~$R_b \in \mathcal{R}_b$ and $R_t \in \mathcal{R}_t$.

\begin{claim}[Opposite-Sided Polar Claim 1]
    For each everywhere clean solution $S$, the table~$\mathcal{R}_s$ contains the record of $S$ at~$s$.
\end{claim}

\begin{proof}
    Let $S$ be an everywhere clean solution.
    For $x \in \{b, t, s\}$, let $S^x$ denote the restriction of~$S$ to $f^x$ 
    and let $R_x = (\omega_x, X_x, E^1_x, \Phi_x)$ denote its record.
    We also define $R = (\omega, X, E^1, \Phi) = R_s$ to stay consistent with the notation of the algorithm.
    Since $\TT_t$ and $\TT_b$ are representatives of $t$ and $b$, respectively, it holds that $R_t \in \TT_t$ and $R_b \in \TT_b$.
    We will now show that there exists a word $\omega_u$ such that in the iteration where $R_b$, $R_t$, $R$, and $\omega_u$ are considered by the algorithm, no check fails and $R$ is indeed added to the table $\TT_s$ as desired. 
    
    First of all, we have $X_t \cap X_b = \emptyset$ since every missing vertex is placed exactly once by $S$.
    We also have $X = X_t \cup X_b$ since $f^s = f^t \cup f^b \cup s$ holds and the solution $S$ places no missing vertex inside the polar strip $s$.
    
    For an arbitrary vertex $A \in X_t$, a curve starting in $A$, fully contained in $f^s$, and ending on the gate shared by $s$ and its parent 
    would necessarily cross the gate shared by $t$ and $s$ from top to bottom---therefore, such a curve cannot be upward.
    Thus, we have $\Phi_A = \emptyset$, the word $\omega$ contains no letter containing $A$, and this check does not fail.
    Also if some letter, say~$\alpha$, in $\omega_t$ would have a non-empty intersection with $V_M \setminus X$, this would mean that there is an upward or downward curve that starts inside $f \setminus f^s$ and ends inside $f^t$---thus it would necessarily cross the gate shared by $s$ and its parent from top to bottom and the gate shared by $s$ and $t$ from bottom to top---contradicting both upwardness. 
    This check is therefore satisfied.
    Furthermore, for every edge $BA \in E_M^2$ with $A \in X_t$, if $S$ would place $B$ outside~$f^s$, then a drawing of $BA$ would necessarily cross the gate shared by $s$ and its parent from top to bottom, and therefore it would be not an upward curve contradicting the fact that $S$ is a solution---therefore, we have $B \in X$ and this check is satisfied.
    
    Next, consider a vertex $A \in X_b$ and bundle $\gamma \in \Phi^A$.
    By definition of a record of $S^s$, there exists a curve in $S^s$ starting in $A$, ending on the gate of $s$, and labeled $Av'$ for some $v' \in \gamma$.
    Since $S$ is everywhere clean, for every vertex $v \in \gamma$, the partial solution $S^s$ contains a curve labeled $Av$ starting in $A$ and ending on the gate of $s$.
    So let $v \in \gamma$ be arbitrary and let $\phi$ be a curve labeled $Av$ starting in $A$ and ending on the gate of $s$.
    Since the gate shared by $s$ and~$b$ is the only bottom gate of $s$, the curve crosses this gate.
    Therefore, the partial solution~$S^b$ contains a curve labeled $Av$.
    So if $v \in \gamma'$ for some $\gamma' \in \Delta_A(f \setminus f^b)$, then by definition of the record of $S^b$ we have $\gamma' \in \Phi_b^A$, and this check is satisfied.
    
    Now let $\phi^u_1, \dots, \phi^u_r$ be the curves in the solution that leave $u$ inside the strip $s$ in the left-to-right ordering they are leaving $u$, and let $\omega_u$ the word over $V_M$ corresponding to these ordering of curves:
    More formally, for $i \in [r]$ let $A_i \in V_M$ denote the end-vertex of $\phi^u_i$ other than $u$---let us remark that $A_i = A_j$ for $i \neq j$ is, in general, possible since multiple drawings of the same edge are possible in $S$.
    Then there exist unique $\ell \in \mathbb{N}_0$,~$i_1, \dots, i_\ell \in \mathbb{N}^+$, and~$B_1, \dots, B_\ell \in V_M$ such that $A_1 A_2 \dots A_r = B_1^{i_1} B_2^{i_2} \dots B_\ell^{i_\ell}$ and for every $j \in [\ell-1]$ we have~$B_j \neq B_{j+1}$.
    We then define $\omega_u = B_1 \dots B_\ell$.
    Since $u$ on the top-side of $s$, all curves~$\phi^u_1, \dots, \phi^u_r$ enter $u$ from below.
    In other words $\phi_i^u$ is an upward curve from $A_i$ to $q$ and therefore $A_i u \in E_M^1$ holds for every $i \in [r]$.
    Furthermore, since $s$ is polar, no missing vertex is placed inside $s$ by $S$ and thus $\phi_i^u$ crosses the unique bottom gate of $s$, and therefore starts inside $f^b$.
    So we have $A_i \in X_b$.
    In particular, we then have $B_i u \in E_M^1$ and $B_i \in X_b$ for every $i \in [\ell]$.
    Furthermore, for $i \in [\ell]$, if we have $u \in \Delta_{B_i}(f \setminus f^b)$, the fact that there is a curve labeled $B_i u$ in $S^b$ ending on the gate of $s$ implies that we have $u \in \Phi^{B_i}_b$.
    Note that by construction, no two consecutive letters coincide in $\omega_u$.
    We will later show that $\omega_u$ is short.
    Let $E^1_u = \{Au \mid A \text{ occurs in } \omega_u \}$ as defined by the algorithm.
    
    To verify that the checks of the gate words satisfy the desired condition, we inspect how the solution $S$ looks inside the strip $s$.
    Let $z \in \mathbb{N}$ and let $\phi^s_1, \dots, \phi^s_z$ be the curves that constitute the solution $S$ restricted to the strip $s$ in the left-to-right ordering.
    Further, let~$p^b_1, \dots, p^b_z$ denote the endpoints of $\phi^s_1, \dots, \phi^s_z$ on the bottom of $s$, respectively.
    Note that $p^b_i$ lies to the left of $p^b_{i+1}$ for every $i \in [z-1]$.
    Similarly, let $p^t_1, \dots, p^t_z$ denote the endpoints of~$\phi^s_1, \dots, \phi^s_z$ on the top of $s$, respectively.
    Note that $p^t_i$ lies to the left of $p^t_{i+1}$ for every $i \in [z-1]$.
    Recall that the gate shared by $s$ and $t$ lies to the left of the vertex $u$ which, in turn, lies to the left of the gate shared by $s$ and its parent.
    Thus, there exist two indices $h$ and $g$ such that $p^t_1, \dots, p^t_h$ lie on the gate shared by $s$ and $t$, $p^t_{h+1} = \dots = p^t_{g} = u$, and $p^t_{g+1}, \dots, p^t_z$ lie on the gate shared by $s$ and its parent.
    In the solution $S$, the curves~$\phi^s_1, \dots, \phi^s_h$ then ``continue'' inside $f^t$, and therefore yield the subword $\overleftarrow{\omega_t}$ on the gate of $b$.\footnote{We remark that this is $\overleftarrow{\omega_t}$ as opposed to $\omega_t$ because the gate of $t$ lies below $t$, and therefore $\omega_t$ represents the right-to-left ordering of the curves in $S^t$ ending on this gate.}
    The curves $\phi^t_{h+1}, \dots, \phi^t_{g}$, i.e., those ending in $u$, form precisely the intersection of the curves~$\phi^u_1, \dots, \phi^u_r$ with the strip $s$, and therefore yield the subword $\omega_u$ on the gate of $b$.
    Finally, in the solution~$S$, the curves $\phi^t_{g+1}, \dots, \phi^t_{z}$ ``continue'' in $f \setminus f^s$, and therefore yield the subword $\omega$ on the gate of $b$.
    Altogether, this implies that $\omega_b$ is obtained from $\overleftarrow{\omega_t} \omega_u \omega$ by (at most twice) removing a letter that coincides with the letter following it---so this check does not fail as well.
    Note that in particular this implies that $\omega_u$ is a subword of a short (by acceptability of $R_b$, see \cref{def:acceptable-short}) word $\omega_b$, and therefore $\omega_u$ is also short.

    Finally, we will show that the check of $E^1$ does not fail, i.e., that 
    \begin{equation}\label{eq:desired-edge-equality-above-below}
        E^1 = \bigl(E^1_b \cup E^1_t \cup E_u^1\bigr) \setminus \{Av \mid A \in X_b, \exists \gamma \in \Phi^A \colon v \in \gamma\}
    \end{equation}
    holds.
    For $x \in \{s, t, b\}$, let $E'_x \subseteq E_M^1$ be the set of edges routed by $S^x$, and we let $E' = E'_s$.
    Then it holds that 
    \begin{equation}\label{eq:routed-edges-children-different-sides}
        E' = E'_b \cup E'_t \cup E^1_u.
    \end{equation}
    Further, by definition of a record $R_x$ of $S^x$ for $x \in \{s, t, b\}$, we have 
    \[
        E^1 = E' \setminus \{Av \mid A \in X, \exists \gamma \in \Phi^A \colon v \in \gamma\},
    \]
    \[
        E^1_b = E'_b \setminus \{Av \mid A \in X_b, \exists \gamma \in \Phi^A_b \colon v \in \gamma\},  
    \]
    and
    \[
        E^1_t = E'_t \setminus \{vA \mid A \in X_t, \exists \gamma \in \Phi^A_t \colon v \in \gamma\}
    \]
    (note that in the last equality $A$ is the head of the excluded edges because the gate of~$t$ is a bottom-gate).
    So the desired equality~\eqref{eq:desired-edge-equality-above-below} is equivalent to
    \begin{align*}
        &E' \setminus \{Av \mid A \in X, \exists \gamma \in \Phi^A \colon v \in \gamma\} = \\
        &\Bigl(\bigl(E'_b \setminus \{Av \mid A \in X_b, \exists \gamma \in \Phi^A_b \colon v \in \gamma\}\bigr) \cup \bigl(E'_t \setminus \{vA \mid A \in X_t, \exists \gamma \in \Phi^A_t \colon v \in \gamma\}\bigr) \cup E^1_u\Bigr) \setminus \\ 
        &\{Av \mid A \in X_b, \exists \gamma \in \Phi^A \colon v \in \gamma\}.
    \end{align*}
    In the remainder of the proof we will rewrite the sides of the equality until we reach the desired form.
    First of all, consider a vertex $B \in X_t$ and an incoming edge $wB \in E_M^1$.
    Recall that no drawing of $wB$ can leave $w$ outside $f^s$ as it would then cross the gate shared by $s$ and its parent from top to bottom, and therefore would be not an upward curve.
    Also it cannot leave the vertex $w$ inside $s$ since that would imply $w = u$ but $u$ is at the top of $s$ so only edges incoming into $u$ can leave $u$ inside $s$.
    Thus, $wB$ is routed by one of $S^t$ and $S^b$, i.e., we have $wB \in E'_t \cup E'_b$.
    If $wB$ is routed by $S^b$ we have 
    \[
        wB \in E'_b \setminus \{Av \mid A \in X_b, \exists \gamma \in \Phi^A_b \colon v \in \gamma\}.
    \]
    Otherwise, the edge $wB$ is not routed by $S^b$ and only routed by $S^t$.
    So any drawing of $wB$ in $S$ does not cross the gate of $t$ (as the first such crossing would cross this gate from top to bottom contradicting upwardneess of the corresponding curve).
    Therefore, by definition of~$R_t$ being a record of $S^t$ we have $w \notin \gamma$ for all $\gamma \in \Phi^A_t$.
    Thus it holds that 
    \begin{align*}
        &\Bigl(\bigl(E'_b \setminus \{Av \mid A \in X_b, \exists \gamma \in \Phi^A_b \colon v \in \gamma\}\bigr) \cup \bigl(E'_t \setminus \{vA \mid A \in X_t, \exists \gamma \in \Phi^A_t \colon v \in \gamma\}\bigr) \cup E^1_u\Bigr) \setminus \\ 
        &\{Av \mid A \in X_b, \exists \gamma \in \Phi^A \colon v \in \gamma\} = \\
        &\Bigl(\bigl(E'_b \setminus \{Av \mid A \in X_b, \exists \gamma \in \Phi^A_b \colon v \in \gamma\}\bigr) \cup E'_t \cup E^1_u\Bigr) \setminus \{Av \mid A \in X_b, \exists \gamma \in \Phi^A \colon v \in \gamma\}.
    \end{align*}
    
    Now consider an edge $Av \in E_M^1$ with $A \in X_b$ and $v \notin \gamma$ for any $\gamma \in \Phi^A$ and $v \in \gamma'$ for some $\gamma' \in \Phi_b^A$.
    First, by definition of $R_b$ being the record of $S^b$, there is a curve labeled $Av'$ in $S^b$ ending on the gate of $b$ for some $v' \in \gamma'$.
    Since $S$ is everywhere clean, $S^b$ is clean so there is also a curve, say $\alpha$, labeled $Av$ in $S^b$ ending on the gate of $b$.
    The property $v \notin \gamma$ for any $\gamma \in \Phi^A$ implies by \cref{lem:outer-bundles-behaviour} that exactly one of the following two properties hold. 
    Either there exists no downward curve from $v$ to the gate of $s$ inside $f \setminus f^s$, or we have $v \in \gamma^*$ for some $\gamma^* \in \Delta_A(f \setminus f^s) \setminus \Phi^A$.
    In the former case the curve $\alpha$ cannot continue by crossing the gate of $s$ as it would then never be able to reach $v$ while being an upward curve---and therefore, it continues inside $f^t$ and reaches $v$ there, i.e., we have $Av \in E_t^1$.
    In the latter case, by definition of the record $R$ of $S^s$, there is no curve labeled $Av$ in $S^s$ ending on the gate of~$s$, and therefore the curve $\alpha$ continues in $f^t$ and we also have $Av \in E_t^1$.
    Hence, we have
    \begin{align*}
        &\Bigl(\bigl(E'_b \setminus \{Av \mid A \in X_b, \exists \gamma \in \Phi^A_b \colon v \in \gamma\}\bigr) \cup E'_t \cup E^1_u\Bigr) \setminus \{Av \mid A \in X_b, \exists \gamma \in \Phi^A \colon v \in \gamma\} = \\
        &\Bigl(\bigl(E'_b \setminus \{Av \mid A \in X_b, \exists \gamma \in \Phi^A_b \colon v \in \gamma, \exists \gamma' \in \Phi^A \colon v \in \gamma'\}\bigr) \cup E'_t \cup E^1_u\Bigr) \setminus \\ 
        &\{Av \mid A \in X_b, \exists \gamma \in \Phi^A \colon v \in \gamma\} = \\
        &\Bigl(E'_b \cup E'_t \cup E^1_u\Bigr) \setminus \{Av \mid A \in X_b, \exists \gamma \in \Phi^A \colon v \in \gamma\} \stackrel{\eqref{eq:routed-edges-children-different-sides}}{=} \\
        &E' \setminus \{Av \mid A \in X_b, \exists \gamma \in \Phi^A \colon v \in \gamma\}
    \end{align*}
    ---here, the second equality holds since the set excluded by the inner set difference operator is a subset of the set excluded by the outer set difference operator, i.e., due to 
    \[
        \{Av \mid A \in X_b, \exists \gamma \in \Phi^A_b \colon v \in \gamma, \exists \gamma' \in \Phi^A \colon v \in \gamma'\} \subseteq \{Av \mid A \in X_b, \exists \gamma \in \Phi^A \colon v \in \gamma\}.
    \]
    Finally, one check of the algorithm ensures that we have $\Phi^A = \emptyset$ for every $A \in X_t$.
    Together with $X_b \cup X_t = X$ this implies that 
    \[
        E' \setminus \{Av \mid A \in X, \exists \gamma \in \Phi^A \colon v \in \gamma\} = E' \setminus \{Av \mid A \in X_b, \exists \gamma \in \Phi^A \colon v \in \gamma\}.
    \]
    Altogether, this implies that both sides of the desired equality~\eqref{eq:routed-edges-children-different-sides} are equal to the set~$E' \setminus \{Av \mid A \in X_b, \exists \gamma \in \Phi^A \colon v \in \gamma\}$, and therefore this equality holds.
    So none of the checks fails and~$R$ is indeed added to $\TT_s$ by the algorithm.
\end{proof}

\begin{claim}[Opposite-Sided Polar Claim 2]\label{clm: opposite-sided}
    For every record $R$ in $\mathcal{R}_s$ there exists a clean partial solution with this record.
\end{claim}

\begin{proof}
    Consider an arbitrary record $R = (\omega, X, E^1, \Phi) \in \TT_s$.
    We show that there exists a clean partial solution with this record.
    Let $R_t = (\omega_t, X_t, E^1_t, \Phi_t) \in \TT_t$, $R_b = (\omega_b, X_b, E^1_b, \Phi_b) \in \TT_b$, and $\omega_u$ be the word such that $R$ was added to $\TT_t$ in the iteration of the algorithm where~$R_t$,~$R_b$, and $\omega_u$ were considered. 
    Since $\TT_t$ is a representative of $t$, there exists a clean partial solution, say $P_t$ of $t$ with record $R_t$.
    Similarly, since $\TT_b$ is a representative of $b$, there exists a clean partial solution, say $P_b$ of $b$ with record $R_b$.
    We use $E'_b$ and $E'_t$ to denote the set of edges routed by $P_b$ and $P_t$, respectively.
    We will now construct a labeled drawing $Q$, then show that it is a clean partial solution of $s$, and finally, argue that $R$ is the record of $Q$.

    \begin{figure}[t]
    \centering
    \includegraphics{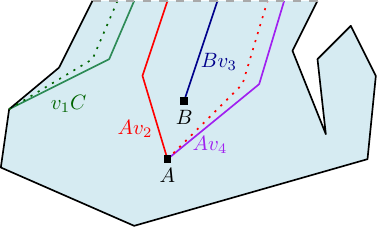}
    \caption{Illustration for the copying procedure. Solid lines represent curves existing in the partial solution. Dotted curves are copied, the color of the copied curve reflects its label. The green solid curve labeled $v_1 C$ starts in the pre-drawn vertex so its copy corresponds to the same occurrence of the letter $C$ in the gate word. The red dotted copy of the curve labeled $Av_2$ closely follows the curve labeled $Av_4$, and therefore corresponds to a different occurrence of the letter $A$ in the gate word than the red solid curve labeled $Av_2$---this is possible since the non-gate endpoint is the missing vertex $A$.
    }
    \label{fig:copying}
    \end{figure}

    In our following arguments we will use the so-called \emph{copying} operation as illustrated in~\cref{fig:copying}.
    To define it, let~$P'$ be a partial solution of some strip $s'$ and let $\phi$ be a curve in~$P'$ labeled $e$ for some $e \in E_M^1$ and ending on the gate of $s'$.
    Further, let $A$ be the type of~$e$, i.e., $\{A\} = e \cap V_M$ and let $\phi'$ be an arbitrary curve of type $A$ in $P'$.
    Then we can carry out one of the following operations.
    First, 
    we can extend $P$ by a new curve $\phi^*$ that is labeled $e$ 
    and follows $\phi$ $\varepsilon$-closely (for a sufficiently small $\varepsilon$) without producing any crossings.
    Second, if $A$ is placed by $P'$ (i.e., if the non-gate endpoint of $\phi$ and $\phi'$ is $A$), we can extend~$P'$ by a new curve $\phi^*$ that is labeled $A v'$ for an arbitrary $v' \in V(H)$ with $Av' \in E_M^1$ 
    and follows~$\phi'$~$\varepsilon$-closely (for sufficiently small $\varepsilon$) without producing any crossings.
    Observe that in both cases the arising object is still a partial solution of $s'$, it yields the same gate word and routes the same set of edges.
    We will refer to the above operations as \emph{copying} the edge~$e$.
    Let us emphasize that $e \notin E_M^2$ is crucial---otherwise copying $e$ would not yield a partial solution.

    Recall that by the algorithmic check \eqref{eq:word-check-opposite-sided}
   it holds that $\omega_b$ is equal to the concatenation~$\overleftarrow{\omega_t} \omega_u \omega$ after removing at most two repetitions.
    We may assume that for every letter of the gate word $\omega_b$ that belongs to $V_M$, there exist at least three curves corresponding to this letter: this can easily be achieved by 
    copying a curve corresponding to this letter twice. 
    Thus, if we consider the curves in $P_b$ ending on the gate of $b$ in the left-to-right ordering of their endpoints, we first see the curves yielding the subword $\overleftarrow{\omega_t}$, then the curves yielding $\omega_u$, and finally the curves for $\omega$.

    Now we will define labeled drawings $Q_b$ and $Q_t$ obtained from $P_b$ and $P_t$, respectively.
    We will then add some straight-line segments inside $s$ to combine $Q_b$ and $Q_t$ and obtain the desired partial solution $Q$ of $s$.
    Let us remark at this point that we will not prove that $Q_b$ and $Q_t$ are partial solutions of $b$ and $t$, respectively, since this will not the necessary for the proof that $Q$ is a partial solution.
    Instead, we will only argue that $Q_b$ and $Q_t$ have some useful properties sufficient for the proof.
    We start with $Q_t \coloneqq P_t$ and $Q_b \coloneqq P_b$ and adapt them as follows.

    First, for every letter $A \in X_b$ such that $A$ occurs in $\omega$ we apply the copying operation until the following property holds: for every occurrence $o$ of $A$ in $\omega$ and every edge $Av \in E_M^1$ 
    there is a curve in $Q_b$ labeled $Av$, ending on the gate of $b$, and corresponding to the occurrence~$o$ of~$A$.
    This does not change the set of edges routed by $Q_b$.
    After that we discard from~$Q_b$ all curves end on the gate of $b$ and are labeled with edges $Av$ with $A \in X_b$ and $v \in \gamma$ for some~$\gamma \in \Delta_A(f \setminus f^s) \setminus \Phi^A$.
    We claim that this discarding does not change the gate word of~$Q_b$.
    Indeed, suppose for a contradiction that it changes the gate word of $Q_b$.
    In particular, this implies that then $A$ occurs in $\omega$.
    By acceptability of $R$ (see \cref{def:acceptable-A-letter}) there exists a vertex~$v$ with $Av \in E_M^1$ and $v \notin \gamma$ for any $\gamma \in \Delta_A(f \setminus f^s) \setminus \Phi^A$.
    At the beginning of this paragraph, a curve labeled $Av$ was copied in $Q_b$ to every occurrence of $A$ in $\omega$.
    Since~$v \notin \gamma$ for any~$\gamma \in \Delta_A(f \setminus f^s) \setminus \Phi^A$, this curve was not discarded, and therefore the gate word remains~$\omega_b$.
    Also note that $Q_b$ still routes the same set of edges as $P_b$ since we only manipulated some curves that start in a missing (and not pre-drawn) vertex and end on the gate of $b$.
    
    Now let $A_1, \dots, A_{|\omega_u|} \in X_b$ be the letters such that $\omega_u = A_1 A_2 \dots A_{|\omega_u|}$---such values exist due to the properties satisfied by $\omega_u$ considered by the algorithm.
    Next we ensure that there are precisely $|\omega_u|$ curves in $Q_b$ corresponding to the subword $\omega_u$ and they are labeled as $A_1 u, A_2 u, \dots, A_{|\omega_u|}$ as follows.
    We achieve this by iterating through $i \in [|\omega_u|]$, and discarding all curves but an arbitrary one corresponding to $A_i$.
    Note that since $A_i \in X_b$, the only curve preserved, in particular, starts in $A_i$; we label this curve $A_i u$.
    Note that this also does not change the set of edges routed by $Q_b$ since we only changed some curves starting in missing vertices and ending on the gate of $b$.

    We refer to \cref{fig:above-below-relabeling} for an illustration of the step described in the next two paragraphs.
    Now let $C_1, \dots, C_{|\omega_t|}$ be such that $\overleftarrow{\omega_t} = C_1, \dots, C_{|\omega_b|}$ and recall that $\omega_b$ is a prefix of $\omega_t$.
    First, we by copying some curves we ensure that for every $i \in |\omega_t|$, the number of curves corresponding to $C_i$ in $Q_b$ and $Q_t$ are equal: we pick one curve on the ``smaller'' side and copy the curves there.
    Again, this does not change the set of edges routed by any of~$Q_b$ and~$Q_t$.
    Note that now we can, for every $i \in [|\omega_t|]$, naturally ``pair'' the curves in $Q_b$ and~$Q_t$ in the left-to-right ordering---the paired curves will later be connected by a straight-line segment to yield a finished edge in $Q$.
    Now we will describe a certain careful relabeling of the curves so that first, two paired curves ``from'' $C_i$ have the same label and this label is from the side ($t$ or $b$) where $C_i$ is not placed.
    We formalize this process below. 
    
    Recall that by the algorithmic check, for every letter $\alpha$ occurring in $\omega_b$ we have $\alpha \subseteq X_t \cup X_b = X$.
    So for every $i \in |\omega_b|$ we proceed as follows.
    Let $\phi^b_{i,1}, \dots, \phi^b_{i,z_i}$ be the curves corresponding to $C_i$ in $Q_b$ and let $\phi^t_{i,1}, \dots, \phi^t_{i,z_i}$ be the curves corresponding to $C_i$ in $Q_t$ where $z_i$ denotes their number.
    If $C_i \in {V_M \choose 2}$ holds, then we have $z_i = 1$---this is since both~$P_b$ and~$P_t$ are partial solutions, and the previous copying step did not increase the number of these curves because both were equal to 1---and both $\phi^b_{i,1}$ and $\phi^t_{i,1}$ are labeled with the edge $C_i$ (more formally, it is labeled with the unique directed version of $C_i$ in $E_M^2$).
    So the labels of $\phi^b_{i,1}$ and $\phi^t_{i,1}$ already coincide and no change needed here.
    Otherwise, for $j \in [z_i]$ we proceed as follows.
    If $C_i \in X_b$ holds, let $C_i u_j$ resp.\ $C_i w_j$ be the label of $\phi^b_{i,j}$ and $\phi^t_{i,j}$.
    We relabel the curve $\phi^b_{i,j}$ to have the label $C_i w_j$---importantly this does not change the set of edges routed by $Q_b$ since due to $C_i \in X_b$ the curve $\phi^b_{i,j}$ starts in $C_i$, and this clearly does not change the set of edges routed by $Q_t$.
    Furthermore, observe that since $C_i \notin X_t$, the curve $\phi^t_{i,j}$ starts in $w_j$.
    Later we will connect the endpoints of $\phi^b_{i,j}$ and $\phi^t_{i,j}$ by a straight-line segment that will therefore finish the edge $C_i w_j$. 
    If $C_i \in X_t$ holds, let $u_j C_i$ resp.\ $w_j C_i$ be the label of $\phi^b_{i,j}$ and $\phi^t_{i,j}$.
    We similarly relabel the curve $\phi^t_{i,j}$ to have the label $u_j C_i$. 

    \begin{figure}[t]
        \centering
        \includegraphics[width=0.8\textwidth]{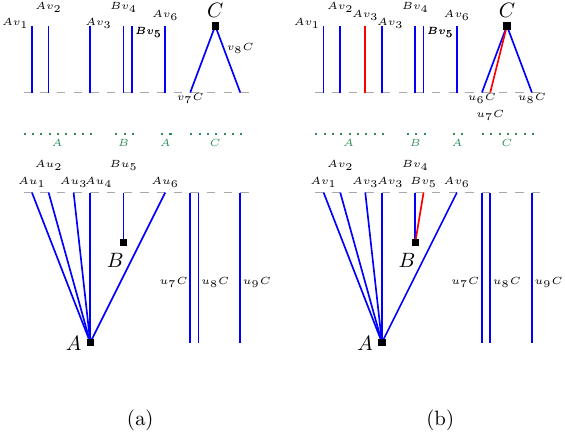}
        \caption{Illustration for the copying and relabeling of the curves to obtain $Q_t$ and $Q_b$. In blue we have the curves from $P_t$ and $P_b$, red curves are the copied ones. The ``boundaries'' of each letter of the gate word $ABAC$ are depicted in green. (a) Labeled curves from $P_t$ and $P_b$ (b) The arising relabeled drawing, in particular, paired curves now have the same label.}
        \label{fig:above-below-relabeling}
    \end{figure}

    This concludes the construction of labeled drawings $Q_b$ and $Q_t$.
    Crucially, by construction,~$Q_b$~finishes the same set $X_b$ of vertices, routes the same set $E'_b$ of edges, and yields the same gate word $\omega_b$ as $P_b$ does (we will elaborate on outer bundles and $\Phi^A_b$ in detail later on).
    The analogous statement holds for $Q_t$ and $P_t$.

    Now we can finally define the labeled drawing $Q$.
    It is obtained by starting with the union of $Q_b$ and $Q_t$ and adding some straight-line segments inside $s$ as follows.
    First, for every~$i \in [|\omega_t|]$ and every $j \in [z_i]$ we connect the gate endpoints of $\phi^b_{i,j}$ and $\phi^t_{i,j}$ by a straight-line segment---recall that this curves are labeled with the same edge and by construction of the label, the arising curve is now the drawing of this edge.
    Second, for every $i \in [|\omega_u|]$ there is precisely one curve in $Q_b$ corresponding to the letter $A_i$ in the subword $\omega_u$ of $\omega_b$: furthermore, this curve is labeled $A_i u$ and starts in $A_i$.
    We connect the gate endpoint of this curve with $u$ by a straight-line segment---thus finishing the edge $A_i u$.
    Finally, let $r \in \mathbb{N}^+$ and let $\psi_1^b, \dots, \psi_r^b$ be the curves in $Q_b$ that were not yet extended by straight-line segments.
    Recall that these curves correspond to the suffix $\omega$ of $\omega_b$ and these are the curves in $Q_b$ with the rightmost endpoints on the gate of $b$.
    We choose points $p_1, \dots, p_r$ on the gate shared by~$s$ and its parent such that $p_i$ lies to the left of $p_{i+1}$ for every $i \in [r-1]$.
    And for every~$i \in [r]$, we add a straight-line segment between the gate endpoint of $\psi_i^b$ and $p_i$, the arising curve then preserves the label of $\psi_i^b$.
    This concludes the construction of $Q$.
    First, note that the gate word of $Q$ is $\omega$.
    Also observe that we add the straight-line segments from left to right so no crossings are produced.
    
    We will now argue that $Q$ is indeed a clean partial solution of $s$ with the record $R$.
    First, the vertices placed by $Q$ are precisely $X_b \cup X_t$.
    One check of the algorithm ensures that~$X_b \cap X_t = \emptyset$ and~$X_b \cup X_t = X$---so $Q$ places the missing vertices from $X$ and each of them is placed precisely once.
    And we argued above that the gate word of $Q$ is $\omega$.
    Let $E'$ denote the set of edges routed by $Q$, then it holds that 
    \[
        E' = E'_b \cup E'_t \cup E_u
    \]
    by construction---this is because $Q_b$ resp.\ $Q_t$ routes the same set $E'_b$ resp.\ $E'_t$ of edges as $P_b$ resp.\ $P_t$ does.
    The first 6 conditions (\cref{def:ps-no-crossing-gamma}-\cref{def:ps-label}) from the definition of a partial solution are trivially satisfied by construction.
    
    Now we consider the remaining properties (\cref{def:ps-placed-vertices} and \cref{def:ps-E-M-2}).
    We start with edges incoming into missing vertices.
    For a vertex $A \in X_b$ it is satisfied because all edges incoming into $A$ are finished by $P_b$ by definition of a partial solution and our transformations of the drawings never removed the finished edges, nor they changed labels of curves starting in pre-drawn vertices. 
    Now consider a vertex $A \in X_t$.
    We know that $R$ is acceptable so for every incoming edge $vA \in E_M^1$ we have $vA \in E^1$.
    The algorithmic check also ensures that we have 
    \[
        E^1 = E^1_b \cup E^1_t \cup E_u^1 \setminus \{Av \mid A \in X_b, \exists \gamma \in \Phi^A \colon v \in \gamma\}.
    \]
    Moreover, recall that we have $vA \notin E_u^1$ since this set only contains edges incoming into~$u$.
    Thus, $vA \in E^1_b \cup E^1_t$.
    If we have $vA \in E^1_t$, then the finished drawing of $vA$ from $P_t$ remains in~$Q$ and so $vA$ is finished by~$Q$ as well.
    And if $vA \in E^1_b$ holds, then there is a curve, say~$\phi^b_{i,j}$ for some $i, j \in \mathbb{N}$, in $P_b$ labeled $vA$ and ending on the gate of $b$.
    Recall that by the algorithmic check, the letter $A$ occurs in none of $\omega_u$ and $\omega$.
    So this curve corresponds to some letter $A$ in $\overleftarrow{\omega_t}$.
    Also recall that by construction of $Q_b$, since $A \in X_t$, we preserved the label of $\phi^b_{i,j}$.
    After that we connected $\phi^b_{i,j}$ by a straight-line segment with the gate endpoint of $\phi^t_{i,j}$, which is a downward curve of type $A$ so it starts in $A$ and ends on the gate of $t$.
    Hence the edge~$vA$ is finished and, once again, since $A$ does not occur in $\omega$, there is no curve in $Q$ labeled~$vA$ and ending on the gate of $s$.
    Now we consider an edge $BA \in E_M^2$ (for the still fixed vertex~$A \in X_t$).
    Again, $A$ does not occur in any letter of $\omega$ so $Q$ contains no curve ending on the gate of $s$ labeled $BA$.
    Further, the algorithmic check ensures that we have $B \in X$.
    If $B \in X_t$, then the edge is already finished by $P_t$, we did not change the finished edges so the edge is also finished by $Q$.
    And if $B \in X_b$, then by definition of partial solutions, both $\omega_b$ and $\overleftarrow{\omega_t}$ contain the letter $\{A,B\}$.
    The construction of $Q$ connects the unique corresponding curves by a straight-line segment and thus, finishes this edge $AB$.

    Now we move on to the edges outgoing from the missing vertices. 
    For a vertex $A \in X_t$, the definition of a partial solution ensures that all edges outgoing from $A$ are finished by~$P_t$, and they are still finished by $Q$ by construction.
    Furthermore, $\omega$ does not contain any letter containing $A$ so the edges outgoing from $A$ satisfy the desired property.
    Now consider a vertex~$A \in X_b$.
    We start with the simpler case and consider an edge $AB \in E_M^2$.
    If $B \in X_b$, then $P_b$ finished this edge and the letter $\{A, B\}$ does not occur in any of $\omega, \omega_u, \omega_t, \omega_b$.
    If~$B \in X_t$, then we already argued why this edges satisfies the desired property when considering the edges incoming into $B \in X_t$.
    So we may assume $B \in V_M \setminus X$.
    The record $R$ is acceptable so the word $\omega$ contains exactly one letter $\{A, B\}$.
    Since $\omega$ is a subword of $\omega_b$, the partial solution~$P_b$ contains a curve labeled $AB$ and ending on the gate of $s$ in the part corresponding to the subword $\omega$ of $\omega_b$.
    The construction of $Q$ continues this curve to the gate of $s$ as desired.
    
    Now consider an edge $Av \in E_M^1$ with $A \in X_b$.
    If the edge is finished by $P_b$, it remains finished by $Q$---so we may now assume $Av \notin E'_b$, and in particular this implies $Av \notin E^1_b$.
    If the edge is routed by $P_t$, i.e., there is a downward curve, say $\alpha$, from $v$ to the gate of $t$ labeled $Av$, then the construction of $Q$ (recall $A \in X_b$) ensures that $\alpha$ is connected using a straight-line segment to a curve in $Q_t$ starting in $A$ and ending on the gate of $b$---so the edge $Av$ is finished.
    Hence, we may now assume that the edge is not routed by $P_t$ either, i.e.,~$Av \notin E'_t$.
    Further, if $v = u$ and $Av \in E_u^1$, the edge $Av$ is finished by construction as well.
    The assumptions we made so far imply that we have $Av \notin E^1_b \cup E^1_t \cup E^1_u \supseteq E^1$ where the last subset-relation holds due to the check of the algorithm.
    Note that since $R$ is acceptable, $v$ is not internal to $f^s$ (see \cref{def:acceptable-internal}). 
    Now recall that since $Av$ is not finished by $P_b$, by definition of a partial solution (see \cref{def:ps-placed-vertices}) it contains a curve, say $\phi$, labeled $Av$ and ending on the gate of $b$.
    Now we make a case distinction.
    If we have $v \notin V^A(f \setminus f^s)$, then $A$ occurs in $\omega$ (see \cref{def:acceptable-outgoing-edges})---now recall that in the first step of the construction of $Q_b$ the existence of $\phi$ ensured that there is a curve labeled $Av$ ending in the suffix $\omega$ of $\omega_b$, and thus by construction of $Q$, this curve was continued by a straight-line segment to add on the gate of $s$ with the same label. 
    If we would have $v \in \gamma$ for some $\gamma \in \Delta_A(f \setminus f^s) \setminus \Phi^A$, then the acceptability of~$R$ (see \cref{def:acceptable-outer-bundles-routed}) would imply $Av \in E^1$, contradicting the above assumption.
    Hence, we have~$v \in \gamma$ for some~$\gamma \in \Phi^A$.
    By acceptability we know that $A$ occurs in $\omega$ (see \cref{def:acceptable-outer-bundles-letter}).
    Then as in the case above, by construction of $Q_b$ there is a curve in the suffix $\omega$ of $\omega_b$ that is labeled~$Av$ and continues in $Q$.
    Altogether, we obtain that either $Av$ is finished in $Q$ or there is a curve in $Q$ labeled $Av$ and ending on the gate of $s$.
    Therefore, $Q$ is indeed a partial solution.
    We will argue its cleanness later on.
    
    Next we show that $\Phi$ satisfies the property of a record of $Q$.
    First, recall that in the construction of $Q_b$   
    we discarded all curves labeled $Av$ with $A \in X_b$ and $v \in \gamma$ for some~$\gamma \in \Delta_A(f \setminus f^s) \setminus \Phi^A$ so no such curve ends on the gate of $s$ in $Q$.
    We will now show that for every $A \in X$, every $\gamma \in \Phi^A$, and every $v \in \gamma$, there is a curve in $Q$ ending on the gate of $s$.
    If $\Phi^A = \emptyset$, the claim trivially holds.
    Otherwise, by acceptability (see \cref{def:acceptable-outer-bundles-letter}), the word $\omega$ contains the letter $A$.
    The construction of $Q_b$ starts by ensuring that for this letter~$A$ there exists a curve labeled $Av$ for every $Av \in E_M^1$.
    After that we only discard curves labeled $Av'$ for $v' \in \gamma'$ for some $\gamma' \in \Delta_A(f \setminus f^s) \setminus \Phi^A$.
    Thus, for every $\gamma \in \Phi^A$ and every~$v \in \gamma$, the drawing $Q$ contains the curve labeled $Av$.
    Note that in particular, this implies that $Q$ satisfies the condition of being clean for all vertices $A \in X$.
    Finally, we now show that $E^1$ indeed satisfies the properties of the record of $Q$---note that the acceptability of $R$ then implies that $Q$ is also clean with respect to $A$-bundles for $A \in V_M \setminus X$ implying that $Q$ is clean.

    Given the above, it remains to show that for the set $E' = E'_b \cup E'_t \cup E^1_u$ of edges routed by $Q$ it holds that
    \[
        E^1 = E' \setminus \{Av \mid A \in X, \exists \gamma \in \Phi^A \colon v \in \gamma\}.
    \]
    Recall that the algorithm checks that $\Phi^A = \emptyset$ for every $A \in X_t$ so the above equality is equivalent to 
    \begin{equation}\label{eq:desired-eq-partial-solution}
        E^1 = (E'_b \cup E'_t \cup E^1_u) \setminus \{Av \mid A \in X_b, \exists \gamma \in \Phi^A \colon v \in \gamma\}.
    \end{equation}
    Recall that the algorithm also  checks that 
    \begin{equation}\label{eq:algorithmic-check}
        E^1 = (E^1_b \cup E^1_t \cup E_u^1) \setminus \{Av \mid A \in X_b, \exists \gamma \in \Phi^A \colon v \in \gamma\}
    \end{equation}
    holds.
    Since $P_t$ ant $P_b$ are partial solutions with records $R_t$ and $R_b$, respectively, we have
    \begin{equation}\label{eq:useful-equation}
        E^1_b = E'_b \setminus \{Av \mid A \in X_b, \exists \gamma \in \Phi^A_b \colon v \in \gamma\}, E^1_t = E'_t \setminus \{vA \mid A \in X_t, \exists \gamma \in \Phi^A_t \colon v \in \gamma\}.
    \end{equation}
    Therefore, the left-hand side of \eqref{eq:desired-eq-partial-solution} is trivially a subset of the right-hand side.
    So it remains to show the other direction.
    For this we consider an edge 
    \[
        e \in (E'_b \cup E'_t \cup E^1_u) \setminus \{Av \mid A \in X_b, \exists \gamma \in \Phi^A \colon v \in \gamma\}
    \]
    and will show that it belongs to $E^1$.
    Note that in particular, $e \in E'_b \cup E'_t \cup E^1_u$ implies that the edge is routed by $Q$.

    First, suppose that $e = vA$ for some $v \in V(H)$ and $A \in V_M$.
    If $A \in X$, then $vA \in E^1$ holds by acceptability of $R$ (see \cref{def:acceptable-incoming edges}).
    If $A \in V_M \setminus X$, then the following holds.
    The edge $vA$ is routed by $Q$ but it is not finished by $Q$ since $A$ was not placed by $Q$.
    Therefore, there is a curve labeled $vA$ in $Q$ and ending on the gate of $s$.
    By construction, this curve was obtained from a curve labeled $vA$ in $P_t$ by adding a straight-line segment to it, i.e., we have~$vA \in E'_b$ (for this recall that when transforming $P_b$ into $Q_b$ we did not discard any curve whose non-gate endpoint is pre-drawn).
    By definition of $E^1_b$ (see \eqref{eq:useful-equation}) we then have~$vA \in E^1_b$, i.e., $vA$ also belongs to the left-hand side of \eqref{eq:desired-eq-partial-solution}.
    So in both cases we have $e \in E^1$.

    Now we may assume that $e = Av$ for some $v \in V(H)$ and $A \in V_M$.
    First, consider the case that $Av \in E'_t$, by definition of $E^1_t$ we then also have $Av \in E^1_t$ (see \eqref{eq:useful-equation}).
    If $Av \in E^1_u$, then by \eqref{eq:algorithmic-check}, we also have $Av \in E^1$.
    So we may now assume that $Av \in E'_b$.
    By definition of a partial solution $P_b$, we then have $A \in X_b$.
    If $Av \in E^1_b$, then we are done as $Av \in E^1$ by~\eqref{eq:algorithmic-check}.
    So it holds that $v \in \gamma$ for some $\gamma \in \Phi^A_b$---in particular, we have $v \in V^A(f \setminus f^b)$.
    Recall that we are considering the edge $Av$ from the right-hand side of \eqref{eq:desired-eq-partial-solution}, and therefore we have $v \notin \gamma'$ for every $\gamma' \in \Phi^A$.
    Thus, there are two possibilities: either we have $v \in \gamma^*$ for some $\gamma^* \in \Delta_A(f \setminus f^s) \setminus \Phi^A$---in which case we have $Av \in E^1$ since $R$ is an acceptable record (see \cref{def:acceptable-outer-bundles-routed}). 
    Otherwise, we have $v \notin V^A(f \setminus f^s)$.
    By \cref{lem:outer-bundles-behaviour}, the vertex $v$ is then either internal to $f^s$ or not interesting with respect to $f \setminus f^s$---in both cases, the acceptability of $R$ (see \cref{def:acceptable-internal} and \cref{def:acceptable-non-interesting}), we have $Av \in E^1$.
    This concludes the proof of \eqref{eq:desired-eq-partial-solution}.
    Altogether, we obtain that $Q$ is a clean partial solution with record $R$.
\end{proof}

\begin{claim}[Opposite-Sided Polar Claim 3]
    The described procedure computing $\TT_s$ runs in time $k^{\mathcal{O}(k^3)} n^{\mathcal{O}(1)}$.
\end{claim}
 
\begin{proof}
    First of all, the algorithm computes the set $D_s$ of all acceptable records for $s$---by \cref{thm:relevant-records} this can be done in time $k^{\bigoh(k^3)} n^{\bigoh(1)}$, and the size of this set is bounded by~$k^{\bigoh(k^3)}$.
    After that it iterates through all records $R_t \in \TT_t$ and $R_b \in \TT_b$ (the size of each of these sets is bounded by $k^{\bigoh(k^3)}$ \cref{cor:size-of-representative}), all acceptable records $R \in D_s$ for $s$, and the short words~$\omega_u$ (whose number is bounded by $k^{\mathcal{O}(k^3)}$ by \cref{obs:short-words-cubic}).
    And for a fixed combination of these objects, the algorithm carries out a set of checks each of which runs in time polynomial in $n$.
\end{proof}

\subsection{Both Children on the Same Side}
Let $s$ be a polar strip with exactly two children that both lie on the same side of $s$---without loss of generality we assume that these both children lie below $s$.
Then the parent of $s$ lies above $s$.
Let $l$ and $r$ denote the left and right child of $s$ correspondingly. 

The algorithm first computes the set $D_s$ of all records acceptable for $s$ (see \cref{thm:relevant-records}).
Then it considers all pairs $R_l = (\omega_l, X_l, E^1_l, \Phi_l) \in \mathcal{R}_l$, $R_r = (\omega_r, X_r, E^1_r, \Phi_r) \in \mathcal{R}_r$, and all acceptable records $R = (\omega, X, E^1, \Phi) \in D_s$ to check if it is possible to ``combine'' the records~$R_l$ and $R_r$ to obtain the record $R$.
So let now $R_l$, $R_r$, and $R$ be fixed.

First of all, we ensure that $X$ is a disjoint union of $X_l$ and $X_r$, otherwise we skip this triple of records. 
Next, for every $A\in X$, and every  $v \in \bigcup \Phi_l^A \cup \bigcup \Phi_r^A$, we check the following. If $v$ belongs to some $A$-bundle $\gamma \in \Delta_A(f\setminus f^s)$, then $\gamma \in \Phi^A$. We also ensure that there is no missing vertex $A \in X_l$ such that there is a letter in $\omega_r$ containing $A$. Symmetrically, we check that there is no missing vertex $A \in X_r$ such that there is a letter in $\omega_r$ containing $A$.

By construction of the strip tree, the common gate of $l$ and $s$ and the common gate of $r$ and $s$ share a vertex, say $q$.  
To ensure that the gate words of $s$, $l$ and $r$ agree, we branch over the words $\omega_q$ with the following properties:
\begin{itemize}
    \item $\omega_q$ contains only letters from $V_M \setminus X$,
    \item $\omega_q$ is short,
    \item for every letter $A$ occurring in $\omega_q$ we have $qA \in E_M^1$.
\end{itemize}
We then define the set $E_q^1 = \{A \in V_M \mid A \text{ occurs in } \omega_q\}$.
For a fixed choice of $\omega_q$, let $\omega_*$ be the concatenation of $\omega_l$, $\omega_q$ and $\omega_r$. By a concatenation we mean that we glue the words together and in case there are consecutive appearances of the same letter, we keep only one of them. We check that $\omega_*=\omega$ holds.
Finally, we set $E^1_*=(E^1_l \cup E^1_r \cup E^1_q) \setminus \{Av \mid A \in X, \exists \gamma \in \Phi^A \colon v \in \gamma \}$. If $E^1= E^1_*$, we add the record $R$ to $\mathcal{R}_s$, and otherwise we continue with the next choice of~$\omega_q$.

To establish correctness, it suffices to prove the following two claims. 
\begin{claim}[Same-Sided Polar Claim 1]
    For every everywhere clean solution $S$, the table~$\mathcal{R}_s$ contains the record of $S$ at~$s$.
\end{claim}
\begin{proof}
Let $S_s$ be a partial solution in $s$, obtained by restricting everywhere clean solution~$S$ to $f^s$. Let $S_l$ and $S_r$ be the restrictions of $S$ to $f^l$ and $f^r$ respectively. 
The table $\TT_l$ (resp.~$\TT_r$) is a representative of $l$ (resp.\ $r$) so it 
contains the record $R_l=(X_l, \omega_l, E^1_l, \Phi_l)$ (resp.~$R_r=(X_r, \omega_r, E^1_r, \Phi_r)$) of $S_l$ (resp.\ $S_r$). Let $R=(X, \omega, E^1, \Phi)$ be the record of $S_s$, then it is acceptable by \cref{thm:relevant-records}. We will show that there exists a word $\omega_q$ such that algorithm adds $R$ to $\mathcal T_s$ when it considers the triple of records $R$, $R_l$, and $R_r$ together with the word $\omega_q$, i.e., no check fails.

 The interiors of $f^l$ and $f^r$ are disjoint, so we have $X_l \cap X_r =\emptyset$. Moreover, since $S$ is a solution, the partial solution $S_s$ does not place any vertices in the non-polar strip $s$, so the set $X$ of missing vertices placed by $S_s$ is precisely $X_l \cup X_r$. 
 Moreover, consider some $A \in X_l$ and suppose that $\omega_r$ contains a letter containing $A$. 
 This would imply that $S$ contains an upward curve that starts inside $f^r$ and ends inside $f^l$---such an upward curve from would need to first cross the gate shared by $l$ and $s$ from bottom to top and then cross the gate shared by $r$ and $s$ from top to bottom---contradicting upwardness.
  A symmetric argument shows that we cannot have $A \in X_r$ and $A$ occurring in some letter of $\omega_l$---so these checks are also satisfied.

 To see that the check for the bundles in $\Phi$ is passed, fix $A\in X$. Let $v$ be some vertex of~$\bigcup \Phi_l^A$ that also belongs to some $A$-bundle $\gamma \in \Delta_A(f\setminus f^s)$.  
 By definition of the record of a partial solution, $S_l$ contains a curve labeled by $Av'$ where $v'$ belongs to the same $A$-bundle as~$v$ w.r.t. $f\setminus f^l$. Since $S_l$ is clean as a restriction of everywhere clean solution, it also contains a curve labeled with $\{A,v\}$.
  Since the gate shared by $s$ and its parent is the only top-gate of, this curve continues up to this gate in $S_s$.
  Since $R$ is a record of $S_s$, existence of such a curve in $S_s$ implies that $\gamma \in \Phi^A$. The case when $v$ is a vertex of $\bigcup \Phi_r^A$ that also belongs to some $A$-bundle $\gamma \in \Delta_A(f\setminus f^s)$ is symmetric.

 Let $\phi_1^q, \dots, \phi_z^q$ (for some $z \geq 0$ be the curves in $S_s$ that leave $q$ inside $s$ in the left-to-right ordering.
 And let $\omega_q$ be the word obtained from the types of these curves after leaving precisely one letter from each sequence of consecutive equal letters.
 For every letter, say~$\alpha$, occuring in $\gamma$ we have $\alpha \in V_M \setminus X$---this is because no missing vertex is placed by $S$ inside the non-polar strip $s$ so every upward curve in $S$ that leaves $q$ inside $s$ necessarily crosses the unique top-gate of $s$, namely the gate shared by $s$ and its parent.
 Furthermore, we have~$q \alpha \in E_M^1$ since this curve is a drawing of some edge with tail $q$.
 Now consider the branch of the algorithm corresponding to $\omega_q$.
 First, observe that if we traverse the curves in $S_s$ ending on the gate of $s$ from left to right, we first see the curves from $S_l$ that end on the gate of $l$ and the continue to the gate of $s$ in $S_s$; then the curves $\phi_1^q, \dots, \phi_z^q$ starting in~$q$; and finally, the curves from $S_r$ that end on the gate of $r$ and the continue to the gate of $s$ in $S_s$.
 Therefore, the word $\omega$ is obtained from the concatenation $\omega_l \omega_q \omega_r$ by removing repetitions---and this check does not fail.
 Furthermore, this implies that $\omega_q$ is a subword of the short word $\omega$ (by acceptability of $R$) so $\omega_q$ is short as well.
 
 It remains to show that the check on $E^1$ is satisfied, i.e., that we have 
 \begin{equation}\label{eq:desired-same-side-everywhere-clean}
    E^1 = E_*^1 \text{ where } E_*^1 = (E^1_l \cup E^1_r \cup E^1_q) \setminus \{Av \mid A \in X, v \in \bigcup \Phi^A \}.
 \end{equation}
 Let $E'$ denote the set of edges routed by $S_s$, and let $E'_l$ and $E'_r$ denote the sets of edges routed by $S_l$ and $S_r$.
 Then by definition of the records, we have
 \[
    E^1 = E' \setminus \{Av \mid A \in X, \exists \gamma \in \Phi^A \colon v \in \gamma\},
 \]
 \[
    E^1_l = E'_l \setminus \{Av \mid A \in X_l, \exists \gamma \in \Phi^A_l \colon v \in \gamma\},
 \]
 \[
    E^1_r = E'_r \setminus \{Av \mid A \in X_r, \exists \gamma \in \Phi^A_r \colon v \in \gamma\}.
 \]
 By definition of $S_l$, $S_r$, and $E_q^1$, we then have
 \[
    E' = E'_l \cup E'_r \cup E_q^1 \text{, and therefore also } E^1 = (E'_l \cup E'_r \cup E_q^1) \setminus \{Av \mid A \in X, v \in \bigcup \Phi^A \}.
 \]
 It is therefore straight-forward to verify that the sets $E^1$ and $E_*^1$ agree on the edges incoming into missing vertices as well as edges with an endpoint in $V_M \setminus X$---so in the remainder we focus on the edges outgoing from $X$.
 Furthermore, it is immediate to see from these equalities that we have $E_*^1 \subseteq E^1$.
 
 For the other direction, consider an edge $e = Av$ on the left-hand side of \eqref{eq:desired-same-side-everywhere-clean} with $A \in X$.
 In particular, we have $e \in E'$, i.e., the edge $e$ is routed by $S_s$.
 Any routed by $S_s$ edge either belongs to $E^1_q$ or is routed by at least one of $S_l$ and $S_r$. 
 If it belongs to $E_1^q$, then it also belongs to to the right-hand side of \eqref{eq:desired-same-side-everywhere-clean}.
 So we may assume that the edge $e\in E^1$ is routed by~$S_l$ (the case when $e$ is routed by $S_r$ is symmetric), i.e., we have $e \in E'_l$.
 If the edge $e$ belongs to $E^1_l$, we are done, as it then belongs to the right-hand side of \eqref{eq:desired-same-side-everywhere-clean}.
 So let $e$ not belong to $E^1_l$. 
 Then there exists a vertex $A\in X$ and some $A$-bundle $\gamma\in \Phi^A_l$ of~$f^l$ such that $e=Av$ and~$v \in \gamma$. In particular, $v$ is interesting with respect to $l$, so there exists an upward curve, say $\phi$, from the common gate of $l$ and $s$ to $v$ that does not intersect the pre-drawn part. 
 Since the gate shared by $s$ and its parent is the only top-gate of $s$, 
 $\phi$ also intersects the gate $g$ of $s$.
 Hence, $\phi$ witnesses that $v$ is interesting with respect to $s$. Moreover, note that every strip visited by $\phi$ after it intersects the gate of $s$ lies outside of~$f^s$. In particular, $v$ is not internal to $s$. 
 By Lemma~\ref{lem:outer-bundles-behaviour} we conclude that $v$ is $A$-bundled with respect to $f\setminus f^s$. Let $\gamma' \in \Delta_A(f \setminus f^s)$ be the bundle containing $v$. Since $v$ also belongs to $\gamma\in \Phi^A_l$, on of the algorithmic checks ensures that $\gamma' \in \Phi^A$. But then $e=Av$ does not belong to $E^1$---contradicting the choice of $e \in E^1$. 
 So we obtain $E^1 \subseteq E_*^1$.
 Altogether, this implies that this last check does not fail either.
 Therefore, the record $R=(\omega, X, E^1, \Phi)$ was successfully added to $\mathcal T_s$ by the algorithm. 
\end{proof}
\begin{claim}[Same-Sided Polar Claim 2]
    For every record $R$ in $\mathcal{R}_s$ there exists a clean partial solution with this record.
\end{claim}
\begin{proof}
Let~$R=(X, \omega, E^1, \Phi)$ be a record that was added to $\TT_s$ when some triple~$R$,~$R_l=(X_l, \omega_l, E^1_l, \Phi_l) \in \TT_l$, and~$R_r=(X_r, \omega_r, E^1_r, \Phi_r) \in \TT_r$ was processed.
And let $\omega_q$ be the word for which this happened.
Since $\TT_l$ and $\TT_r$ are representatives of $l$ and $r$, respectively, there exist clean partial solutions $S_l$ and $S_r$ of $l$ and $r$ with records $R_l$ and $R_r$, respectively. 
Let~$Y_l$ and~$Y_r$ denote the sets of edges finished by $S_l$ and $S_r$ respectively. We construct $S$ as follows. First, we take an internally disjoint union of $S_l$ and $S_r$. 
Then we extend every curve in $S_l$ (resp.\ $S_r$) ending on the gate of $l$ (resp. $r$) by a straight-line segment inside $s$ ending on the gate of $s$---we do so 
so that all of these end on the gate of $s$ in different internal points, and no two straight-line segments intersect. We preserve the labels of these curves. 
Further, let~$\omega_q = B_1 \dots B_z$ for some $z \geq 0$ and $B_1, \dots, B_z \in V_M \setminus X$. 
For every $i \in [z]$ we add a new curve labeled $q B_i$ in such a way that on the gate of $s$ we first see the prolonged curves from~$S_l$, then the curves labeled $q B_1, \dots, q B_z$ (in this ordering), and then the prolonged curves from $S_r$.
Note that the gate word on the gate of $s$ yield by the drawing created so far is obtained from $\omega_l \omega_q \omega_r$ by repeating at most two repetitions of consecutive letters, i.e., it is precisely $\omega$ due to an algorithmic check.

Next, for every letter $A\in X$, every occurrence $o$ of $A$ in the word $\omega$, and every edge~$Av \in E^1_M$ 
we copy some curve corresponding to $o$ and label the copied curve by $Av$. The operation of copying was formally defined in the proof of \cref{clm: opposite-sided}. 
Next, for each $A$-bundle $\gamma \in \Delta_A(f \setminus f) \setminus \Phi^A$ we discard every curve labeled $Av$ with $v \in \gamma$.
Note that this does not change the set of edges routed by $S$ since we only discard some curves starting in a missing vertex and ending on the gate of $s$.
More importantly, this does not change the gate word of $S$.
Suppose for the sake of contradiction, the gate word changed by this discarding.
In particular, this implies that $\omega$ contains the letter $A$.
Since $R$ is acceptable (see \cref{def:acceptable-A-letter}), there exists a vertex $v$ with $Av \in E_M^1$ and $v \notin \gamma$ for every $\gamma \in \Delta_A(f \setminus f^s) \setminus \Phi^A$. 
But then the curve labeled $Av$ was not discarded and therefore, the gate word did not change.
This concludes the construction of $S$.
We now show that $S$ is a clean partial solution with record~$R$.

Let $X=X_l \cup X_r$ and $Y=Y_l\cup Y_r \cup E^1_q$, then $X$ and $Y$ are the sets of vertices and edges finished by $S$. 
To see that $S$ is a partial solution in $s$, first note that Item~1-\cref{def:ps-label} of \cref{def:partial-solution} follow directly from our construction and the fact that $S_l$ and $S_r$ are partial solutions.

To check \cref{def:ps-placed-vertices}, fix a vertex $A\in X$ and without loss of generality assume that $A \in X_l$ holds.
First, consider the edges outgoing from $A$.
Since \cref{def:ps-placed-vertices} holds for $S_l$, every missing edge of form $vA \in E_M^1$ is finished by $S_l$ so it is by $S$ as well.
Furthermore, $S$ contains no curve labeled $vA$ ending on the gate of $s$, this is true for the following reason. 
First, $S_l$ does not contain such a curve due to satisfying \cref{def:ps-placed-vertices}.
Second, $S_r$ does not since $A$ was not placed by $S_r$.
Further, $\omega_q$ does not contain $A$ due to $A \in X_l$.
Finally, our copying only copies some curves with labels being edges outgoing from missing vertices.
Now we consider edges outgoing from $A$.
First of all, since $S_l$ is a partial solution, for every edge~$Av \in E_M^1$, this edge is either finished by $S_l$ (and it remains finished by $S$ by construction), or there is a curve labeled $Av$ ending on the gate of $l$.
If we did not discard this curve, then it is continued in $S$.
So we may assume that $v \in \gamma$ for some $\gamma \in \Delta_A(f \setminus f^s) \setminus \Phi^A$.
Since $R$ is acceptable, we have~$Av \in E^1$ (see \cref{def:acceptable-outer-bundles-routed}).
By our algorithmic check we, in particular, have~$Av \in E^1_l \cup E^1_r \cup E_q^1$.
We cannot have $Av \in E_q^1$ since this set only contains edges outgoing from $q$.
If we have~$Av \in E^1_l$, then the edge is finished by $S_l$ and therefore, by~$S$ as well.
Finally, we cannot have~$Av \in E^1_r$ since $A$ was not placed by $S_r$ and the gate of $r$ is the top-gate of $r$.
Therefore, $S$ satisfies \cref{def:ps-placed-vertices}.

Finally, to see that \cref{def:ps-E-M-2} holds, consider an edge $e=AB$ whose both endpoints are missing, i.e., $e\in E^2_M$. 
First, if $A$ belongs to $X$ and $B$ does not belong to $X$, we may assume without loss of generality that $A\in X_l$. By \cref{def:ps-E-M-2} for $S_l$, there is precisely one curve labeled by $\{A,B\}$ that intersects the gate of $l$. By construction, this curve is continued to the gate of $s$ in $S$. We show that in this case there can be no curve labeled with $\{A,B\}$ in $S_r$. Recall that the record $R_r$ of $S_r$ is acceptable. If $S_r$ would contain such a curve, then its gate word~$\omega_r$ would contain letter $\{A,B\}$, and hence by \cref{def:acceptable-E-M-2} $A$ must belong to $X_r$. But~$A\in X_l$ and our algorithm ensures that $X_l$ and $X_r$ are pairwise disjoint. Hence, $S$ contains precisely one curve labeled by $\{A,B\}$. It remains to consider the case when both $A$ and $B$ belong to~$X$.
It follows from our algorithmic check on the gate words that in this case they either both belong to $X_l$ or both to $X_r$. 
Then \cref{def:acceptable-E-M-2} for $S_l$ (resp.~$S_r$) implies that $AB \in Y_l$ (resp.~$AB \in Y_r$) and hence $AB \in Y$. Therefore, $S$ is indeed a partial solution in $s$.

Let us show that $R$ is a record of $S$. The first two items in the definition of the record of a partial solution (corresponding to $\omega$ and $X$) follow directly from our construction. 
To check \cref{def:record-phi}, consider any $A\in X$ and any bundle $\gamma \in \Delta_A(f\setminus f^s)$. If $\gamma \in \Phi^A$, then, first, the acceptability of $R$ (see \cref{def:acceptable-outer-bundles-letter}) implies that $\omega$ contains $A$.
Now recall that while constructing~$S$ we ensured, by copying, that it contains a curve labeled by $\{A,v\}$ for every~$v \in \gamma$ with~$\gamma \in \Phi^A$ and such curves were not discarded. 
On the other hand, we discarded all curves labeled $Av$ for $v \in \gamma$ with $\gamma \in \Delta(f\setminus f^s) \setminus \Phi^A$ so $S$ contains no such curve. 

Finally, to check \cref{def:record-e1}, we will show that for the set $E'$ of edges routed by $S$ we have
\begin{equation}\label{eq:desired-eq-same-side-clean-ps}
    E^1 = E' \setminus \{Av \mid A \in X, v \in \bigcup \Phi^A\}.
\end{equation}
Let $E'_l$ and $E'_r$ denote the set of edges routed by $S_l$ and $S_r$, respectively.
By construction of~$S$ we have
\[
    E' = E'_l \cup E'_r \cup E^1_q.
\]
Since $R_l$ and $R_r$ are the records of $S_l$ and $S_r$, respectively, we have
\[
    E^1_l = E'_l \setminus \{Av \mid A \in X_l, v \in \bigcup \Phi^A_l\},
\]
\[
    E^1_r = E'_r \setminus \{Av \mid A \in X_r, v \in \bigcup \Phi^A_r\}.
\]
And one of the algorithmic checks ensures that we have
\begin{equation}\label{eq:alg-check-same-side-proof}
    E^1 = (E^1_l \cup E^1_r \cup E^1_q) \setminus \{Av \mid A \in X, v \in \bigcup \Phi^A\}.
\end{equation}
Thus we have to show that we have
\[
    (E^1_l \cup E^1_r \cup E^1_q) \setminus \{Av \mid A \in X, v \in \bigcup \Phi^A\} = (E'_l \cup E'_r \cup E^1_q) \setminus \{Av \mid A \in X, v \in \bigcup \Phi^A\}.
\]
The above equalities imply that the left-hand side is a subset of the right-hand side and it remains to show the other direction.
So let $e \in E'$ be an edge, i.e., an edge routed by $S$, that does not belong to $\{Av \mid A \in X, v \in \bigcup \Phi^A\}$.
We will show that $e \in E^1$ holds.
First, if we have $e \in E_1^q$, then $e$ also belongs to $E^1$ by \eqref{eq:alg-check-same-side-proof}.
So we may assume without loss of generality that $e \in E'_l$ holds, i.e., $e$ is routed in $S_l$. 
Then it belongs to $E^1_l \cup \{Av \mid A\in X_l,v\in \bigcup \Phi^A_l\}$. 
If $e\in E^1_l$, then it is contained in $E^1$ as well by \eqref{eq:alg-check-same-side-proof}.
So it remains to consider the case when~$e=Av$ for $A\in X_l$ and $v\in \bigcup \Phi_l^A$. 
In particular, $v$ is interesting with respect to $l$, so there exists an upward curve, say $\phi$, from the common gate of $l$ and $s$ to $v$ that does not intersect the pre-drawn part. 
 Since the gate shared by $s$ and its parent is the only top-gate of $s$, 
 $\phi$ also intersects the gate $g$ of $s$.
 Hence, $\phi$ witnesses that $v$ is interesting with respect to~$s$. Moreover, note that every strip visited by $\phi$ after it intersects the gate of $s$ lies outside of $f^s$. In particular, $v$ is not internal to $s$. 
 By Lemma~\ref{lem:outer-bundles-behaviour} we conclude that $v$ is $A$-bundled with respect to $f\setminus f^s$.
Hence, one of the algorithmic checks ensures that $v\in \bigcup \Phi^A$. But this contradicts our assumption that $e$ does not belong to $\{Av \mid A \in X, v \in \bigcup \Phi^A\}$. 
Therefore, the desired equality \eqref{eq:desired-eq-same-side-clean-ps} indeed holds.

Finally, it remains to show that $S$ is clean.
For the outer bundles, it follows from the construction of $S$ that for any $A \in X$ and any $A$-bundle $\gamma\in \Delta_A(f\setminus f^s)$, $S$ contains curves either for every missing edge $Av$ with $v \in \gamma$ (which happens if $\gamma \in \Phi^A$), or for none of them (if $\gamma \not \in \Phi^A$).  
For the inner bundles, we use the fact that $R$ is acceptable record. In particular, for any $A\in V_M \setminus X$, any $A$-bundle $\gamma\in \Delta_A(f^s)$, and any pair $v,v'$ of vertices of $\gamma$ it follows from \cref{def:acceptable-inner-bundles} of \cref{def:acceptable} that whenever $vA$ belongs to $E^1$, $v'A$ belong to $E^1$ as well. 
\end{proof}
\begin{claim}[Same-Sided Polar Claim 3]
    The described procedure computing $\TT_s$ runs in time $k^{\mathcal{O}(k^3)} n^{\mathcal{O}(1)}$.
\end{claim}

\begin{proof}
    First of all, the algorithm computes the set $D_s$ of all acceptable records for $s$---by \cref{thm:relevant-records} this can be done in time $k^{\bigoh(k^3)} n^{\bigoh(1)}$, and the size of this set is bounded by $k^{\bigoh(k^3)}$.
    After that it iterates through all records $R_l \in \TT_l$ and $R_r \in \TT_r$ (the size of each of these sets is bounded by $k^{\bigoh(k^3)}$ \cref{cor:size-of-representative}), all acceptable records $R \in D_s$ for $s$, all subsets $E^1_q$ of missing out-neighbors of $q$ and all words $\omega_q$ containing one letter for each missing endpoint of an edge from $E^1_q$. 
    And for a fixed combination of these objects, the algorithm carries out a set of checks each of which runs in time polynomial in $n$.
\end{proof}

\subsection{Checking the Root}

 Recall that $r'$ denotes the root of $T_f$ (which is a leaf of $T_f$) and $r$ denotes its unique neighbor.
 We can invoke the algorithm $\mathcal{A}$ for each node of a strip tree $T_f$ in a leaf-to-root fashion to compute the a representative for every node of the strip tree other than $r'$, in particular for the unique neighbor $r$ of $r'$.
 It remains to show how the information about the existence of the solution to our instance $\mathcal{I} = (H,G,\Gamma(H))$ can be extracted from the set of records of $r$. 

 \begin{lemma}
 \label{lem:root_correct}
 The instance $\mathcal I$ has a solution if and only if the representative $\mathcal T_{r}$ of $r$ 
 contains the record $(\omega_r, V_M, E^1_M, (\Phi^A)_{A\in X})$, where $\omega_r$ is the empty word 
 and $\Phi^A = \emptyset$ for every $A \in X$. 
 \end{lemma}
 \begin{proof}
 As always, without loss of generality we assume that the strip $r'$ lies above the strip~$r$.
 
 First, suppose that the record $(\omega_r, V_M, E^1_M, \Phi)$ is contained in the set $\mathcal{R}_{r}$. 
 By \cref{thm:dp-algorithm}, the set $\TT_r$ is a representative of $R$
 so in particular, there exists a partial solution, say $P$, of $r$ with the record $(\omega_r, V_M, E^1_M, \Phi)$.
 First, all missing vertices are placed by $P$ inside $f^r$. 
 Second, the definition of a partial solution requires that for any two adjacent vertices $A, B \in V_M$, the edge $\{A, B\}$ is drawn. 
 Therefore, all edges from $E_M^2$ are finished by $P$, i.e., for every missing edge $e$ in $E_M^2$, there is an upward curve in $P$ from its tail to its head.
 Let $E' \subseteq E_M^1$ denote the set of edges routed by $P$.
 By definition of the record (see \cref{def:record-e1}) we have $E^1_M \subseteq E'$.
 Then
 we have $E' = E_M^1$. 
 Since $\omega_r$ is the empty word, no curve in $P$ ends on the gate of $P$, thus every routed edge is finished by $P$, i.e., for every edge in $E_M^1$, there is an upward curve from its tail to its end.
 The remaining properties of a partial solution then imply that $P$ is a solution of $\mathcal{I}$.

 For the other direction, let $S$ be a solution of $\mathcal{I}$.
 By \cref{lem:everywhere-clean-solutions} we may assume that $S$ is everywhere clean.
 By definition of a solution, $S$, in particular, does not place anything inside the leaf $r'$, and therefore $S$ is fully contained in $f^r$.
 By \cref{thm:dp-algorithm}, the set $\TT_r$ is a representative of $R$
 so in particular, it contains the record, say $\tilde R = (\tilde X, \tilde \omega, \tilde E^1, \tilde \Phi)$ of $S$ at $r$.
 Since $S$ is fully contained in $f^r$, the following properties hold.
 First, all vertices are placed inside $f^r$, i.e., we have $\tilde X = V_M$.
 Second, no drawing crosses the gate shared by $r'$ and $r$ and therefore $\omega$ is the empty-word and we have $\Phi^A = \emptyset$ for every $A \in V_M$.
 Finally, for every edge $e \in E_M^1$, there exists a drawing of this edge in $S$ and it is fully contained in $f^r$, i.e.,~$e$ is routed by $S$ inside $f^r$.
 Together with $\Phi^A = \emptyset$ for every $A \in V_M$ this implies that we have~$E^1 = E_M^1$ concluding the proof.
 \end{proof}

Finally, recall that by \cref{lem:streamlining} the strip tree contains $\mathcal{O}(kn)$ nodes.
So by combining~\cref{lem:root_correct} with~\cref{thm:dp-algorithm}, applied to each internal node of the strip tree in a leaf-to-root fashion, we directly obtain the main theorem of the section:

\begin{theorem}\label{thm:main_face}
    \upefshort\ can be solved in time $k^{\mathcal{O}(k^3)}\cdot n^{\mathcal{O}(1)}$.
\end{theorem}

\section{Handling the Outer Face}\label{sec:outer}
Having settled the fixed-parameter tractability of \upef, all that remains before we can solve \textsc{UPE} by applying Lemma~\ref{lem:setup} is to solve \textsc{UPEF-Outer}.
For this, we design a Turing reduction that reduces an instance~$\III=(G,H,\Gamma(H))$ of \textsc{UPEF-Outer} with the outer face~$f$ to at most $2^{k^2+2k}k!$ instances of \textsc{UPEF}.
We proceed as follows.
Intuitively, our approach is as follows. First, we make a box around the outer face $f$, and the inner region of the box will then become the new inner face $f'$. Then, to preserve the connectivity of $H$, we will add a line segment connecting the boundaries of $f$ and $f'$. The crucial difficulty is that the line segment may need to cross a hypothetical solution to $\III$; in fact, it may even need to cross every possible solution to $\III$.
Formally, let us choose an arbitrary (and ``sufficiently large'') rectangle $R$ that completely contains the drawing $\Gamma(H)$ in its interior. Intuitively, we would like to use $R$ to form a new empty outer face that will contain $\Gamma(H)$; however, to add such an outer face into $\Gamma(H)$ we need it to not contain horizontal line segments. Towards this, we choose an arbitrary convex polygon $P$ without horizontal segments which fully contains $R$---we naturally identify~$P$ with an upward planar drawing of a directed graph that contains one vertex per bend in~$P$ and the edges are directed upward, and extend $H$ and $\Gamma(H)$ with this polygon.
Let~$f'$ be the intersection of $f$ with the inner of $P$. It is not hard to see that if our instance admits a solution, it also admits a solution where all the new vertices are inserted inside the rectangle~$R$, from now on we will only restrict ourselves to such solutions. Now $f'$ is an inner face, but $H$ is not connected. 
Towards making $H$ connected, we consider
the lowermost vertex $t$ of $\Gamma(H)$ and let the coordinates of $t$ be $(x_t,y_t)$. We enhance $\Gamma(H)$ by adding a vertical line segment going down from $t$ until it hits the boundary of $P$. Let $b=(x_b,y_b)$ be the point at which the line segment hits the boundary of $P$, 
and let us use $[b,t]$ to denote this line segment. 
We may assume that $b$ is a vertex of $P$ by subdividing the corresponding edge if necessary. 
Below, we show one can assume that, w.l.o.g., $[b,t]$ only intersects missing edges at their endpoints. 

\begin{lemma}\label{lem:vertical}
If there is a solution 
$\Gamma(G)$ of $\III$, then there is also a solution 
$\Gamma'(G)$
placing all missing vertices inside $f'$ such that no internal point of any missing edge intersects the line segment $[b,t]$. 
\end{lemma}

\begin{proof}
Consider any solution $\Gamma(G)\supseteq \Gamma(H)$ where the missing vertices are placed inside~$f'$. We firstly establish the existence of a downward polyline $\phi$, i.e., a polyline whose vertical coordinate monotonically decreases when traversed from $t$ to $b$, not crossing $\Gamma(H)$ in any points other than $t$ and $b$ and such that no internal point of any missing edge is crossed by $\phi$.
To construct such a polyline, we start at $t$ and draw an arbitrary downward polyline from~$t$ to $b$ until encountering the first edge $e$ of $\Gamma(G)$. If the first point encountered on $e$ is an endpoint, we cross it and continue our downward line from it. Otherwise, if we encounter an inner point $q$ of $e$ of $G\setminus H$, we stop our polyline slightly above $q$ and continue downward by following the edge $e$ to its lowermost endpoint. We then cross this endpoint and proceed similarly until the point $b$ on the boundary of $P$ is reached. Without loss of generality, we may assume that $x_t=x_b=0$. It follows from our construction that $\phi$ does not intersect internal points of edges of $\Gamma(G)$. 

Our next step is to turn $\phi$ into a vertical line segment. For this, we modify the drawing~$\Gamma(G)$ by ``flattening'' every curve and point whose $y$-coordinates lie between $y_b$ and $y_t$ onto the vertical line segment connecting $t$ and $b$; more precisely, we rescale and place all such points and curves onto a sufficiently thin rectangular area $R$ surrounding the vertical line segment connecting $t$ and $b$. We note that we may assume w.l.o.g.\ that no part of $\Gamma(G)$ is placed near the point $b$ and that all curves exiting $R$ do so by crossing one of its vertical sides; the latter means that each such curve may be extended in an upward way to reach its position at $y_t$ without crossing any other curve or $P$.
\end{proof}
Given Lemma~\ref{lem:vertical}, it suffices to restrict our attention to solutions $\Gamma(G)$ with the following property: the vertical line segment $[b,t]$ between the lowermost point of $\Gamma(H)$ and the boundary of the outer face intersects $\Gamma(G)$ only in the missing vertices. 
In particular, there are at most $k$ such intersections. Moreover, we may assume w.l.o.g.\ that these intersection points subdivide $[b,t]$ into intervals of equal length. 

\begin{figure}[t]
\centering
\includegraphics[width=\textwidth]{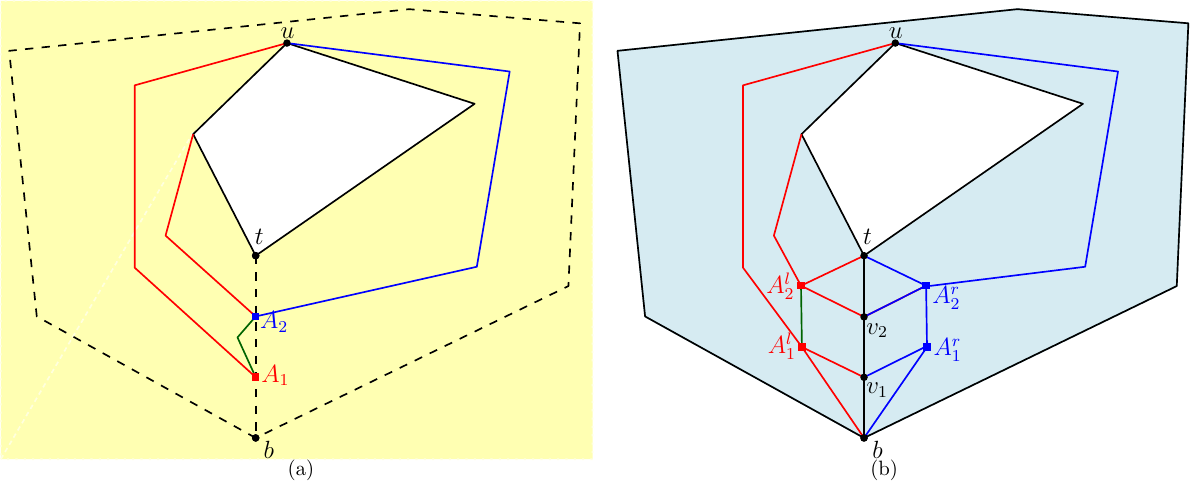}
\caption{Instances of \textsc{UPEF} and \textsc{UPEF-Outer} along with solutions, the pre-drawn part is solid black. (a) Instance $\III$ of \textsc{UPEF-Outer}, the outer face is light-yellow, the new parts $P$ and $[b, t]$ of the outer boundary are dashed. 
(b) An instance
of \textsc{UPEF} corresponding to the branch where the vertices intersecting $[b,t]$ are $A_1$ and $A_2$, and the edges $A_1u$ and $A_2u$ leave the segment $[b,t]$ on the left and right, respectively. The vertex $A_1$ is split into two vertices $A_1^l$ and $A_1^r$, and similarly for $A_2$. The vertices $v_1$ and $v_2$ are new pre-drawn vertices ensuring the correct placement of $A_1^l$, $A_1^r$, $A_2^l$, $A_2^r$.}
\label{fig:outer-face}
\end{figure}

Let $u$ denote the topmost vertex in~$\Gamma(H)$.
We now branch into at most $k!\cdot 2^k$ choices to determine the subset of vertices of $V_M$ that occur on the line segment $[b,t]$ and their ordering on it. 
For every such ordered tuple $\tau=(A_1,\dots, A_q)$ (with pairwise distinct $A_1, \dots, A_q \in V_M$), we additionally branch into at most $2^{k(k+1)}$ binary matrices of size $q \times (k+1)$, with the rows indexed by $[q]$ and the columns indexed by $V_M \cup \{u\}$. Intuitively, the matrix $\mu$ determines for each edge of $G$ between some $A_i$ and some vertex $B \in V_M \cup \{u\}$ 
whether the edge $\{A_i, B\}$ (if it exists) leaves $A_i$ to the left or to the right from the line segment $[b,t]$. 
For every choice of $\tau$ and $\mu$, create a new instance $\III_{\tau}^{\mu}=(G_{\tau}^{\mu},H_{\tau}^{\mu},\Gamma(H_{\tau}^{\mu}))$ obtained by altering $\Gamma(H)$ to incorporate the information from $\tau$ and $\mu$ as 
follows (see \cref{fig:outer-face} for an illustration).
Formally, we proceed as follows.
To construct $G_{\tau}^{\mu}$, we begin by deleting the vertices $\{A_1, \dots, A_q\}$ from~$G$. 
Instead, for each $i\in [q]$, $G_{\tau}^{\mu}$ contains the vertices $v_i$, $A_i^l$, and $A_i^r$---the vertex $v_i$ is pre-drawn and the vertices $A_i^l$ and $A_i^r$ are missing.
The drawing $\Gamma(H)$ is extended so that the vertices $v_1, \dots, v_q$ are placed on the vertical straight-line segment $bt$ in the bottom-up ordering $v_1, \dots, v_q$ so that the straight-line segments $b v_1, v_1 v_2, \dots, v_{q-1} v_q, v_q t$ have the same length.
Furthermore, $G_{\tau}^{\mu}$ contains the pre-drawn edges $b v_1, v_1 v_2, \dots, v_{q-1} v_q, v_q t$ which are drawn as vertical straight-line segments.
For simplicity of notation, let $A_0^l = A_0^r = b$ and~$A_{q+1}^l = A_{q+1}^r = t$.
For every $i \in [q]$ the graph $G_{\tau}^{\mu}$ also contains the missing edges $v_i A_i^l$ and $v_i A_i^r$.
Furthermore, for every $i \in [q]_0$ the graph contains the missing edges~$A_i^l A_{i+1}^l$ and~$A_i^r A_{i+1}^r$.

Then for each edge $A_ia$ of $G$ with $a \neq u \in V(H)$, 
we add to $G_{\tau}^{\mu}$ the missing edge $A_i^l a$ in case $a$ is reachable in $\Gamma(H)$ from $v_i$ via an upward curve that leaves the vertex $v_i$ to the left of the straight-line segment $bt$, and we add the missing edge $A_i^r a$ otherwise. 
Let us remark that if $G$ has an edge from some pre-drawn vertex, say $v$, of $G$ to some $A_i$ ($i \in [q])$, no solution of~$\mathcal{I}$ can place $A$ below $v$ and in particular, no solution can place $A_i$ on $bt$ (this is because~$t$ is the bottommost point of $\Gamma(H)$)---so in this case we just ``discard'' the candidate pair $(\mu, \tau)$ (formally speaking, we let $\mathcal{I^\mu_\tau}$ be a trivial no-instance) and proceed otherwise.
Now for every~$i \in [q]$, if $G$ has an edge from $A_i$ to $u$ and $\mu_{i,u}=0$, we add to $G_{\tau}^{\mu}$ a new edge $A_i^lu$.
And if $G$ has an edge from $A_i$ to $u$ and $\mu_{i,u}=1$, we add to $G_{\tau}^{\mu}$ a new edge $A_i^ru$.
Moreover, for every edge of $G$ of the form $\{A_i,A\}$, where $A \in V_M \setminus \{A_i, \dots, A_q\}$, $i \in [q]$, and $\mu_{i, A}=0$ (resp.~$\mu_{i, A}=1$), we add to $G_{\tau}^{\mu}$ a new edge $\{A_i^l,A\}$ (resp.\ $\{A_i^r,A\}$), in the same direction (i.e., $A$ is the tail of $\{A_i^r,A\}$ if and only if it is the tail of $\{A_i,A\}$). 
For each edge of $G$ that has form $A_iA_j$ for some $i,j\in[q]$, we add to $G_{\tau}^{\mu}$ a new edge $A_i^lA_j^l$ in case $\mu_{i, A_j}=1$ and a new edge $A_i^r A_j^r$ otherwise. 

To solve each of the instances $\III_{\tau}^{\mu}$, we 
can invoke our dynamic programming procedure for the inner face,
as now $f'$ is an inner face 
and the addition of the line segment $[b,t]$ makes its boundary connected. 
Hence, it remains to prove the following two lemmas: 
\begin{lemma}
If $\III$ is a \textup{YES}-instance, then there exists a choice of $\tau$ and $\mu$ such that $\III_{\tau}^{\mu}$ is a \textup{YES}-instance. 
\end{lemma}
\begin{proof}
Let $\Gamma(G)\supseteq \Gamma(H)$ be a solution to the original instance $\III$. By \cref{lem:vertical}, we may assume that $\Gamma(G)$ lies inside the polygon $P$ described above, and the vertical line segment~$[b,t]$ connecting the lowermost point $t$ of~$\Gamma(H)$ to the boundary of $P$ does not contain internal points of edges of $\Gamma(G)$. Let $A_1, \dots, A_q$ be the vertices of $G \setminus H$ that lie on $[b,t]$, from the bottom to the top. We modify $\Gamma(G)$ by compressing some of its parts along $[b,t]$ to achieve that $A_i$, $i\in[q]$, subdivide $[b,t]$ into intervals of equal length. If we now add $[b,t]$ and $P$ to~$\Gamma(H)$ and replace each $A_i$ by $v_i$, $i\in [q]$, we obtain precisely the partial drawing $\Gamma(H_{\tau}^{\mu})$ of the instance $\III_{\tau}^{\mu}$, where $\tau=(A_1, \dots, A_q)$, and the binary matrix $\mu$ is defined as follows. For every edge $\{A_i,A\}$ of $G$, where $A \in V_M \cup \{u\}$, we set $\mu_{i,A}=0$ if there is an edge $\{A_i,A\}$ and it leaves $A_i$ to the left from $[b,t]$, and we set $\mu_{i,A}=1$ otherwise. 

First, we copy from $\Gamma(G)$ the placement of all vertices and edges of the subgraph of $G\setminus H$ induced by the vertices other than $\{A_1, \dots, A_q\}$. Since $\Gamma(G)$ is a solution to $\III$, it contains every edge of $G$, in particular the edges of the form $A_ia$ where $a$ is a non-topmost vertex of $\Gamma(H)$. For every $i\in [q]$, there are two types of such edges in $\Gamma(G)$: the edges leaving $A_i$ to the left of the vertical segment $bt$ and the edges leaving $A_i$ to the right of $bt$ (recall that all such edges intersect $[b,t]$ only in $A_i$). Moreover, it is determined by the drawing $\Gamma(H)$ whether the edge $A_ia$ goes to the left or to the right from $A_i$: only the topmost vertex $u$ of~$\Gamma(H)$ can be reachable from $[b,t]$ via upward curve not intersecting~$\Gamma(H)$ both from the left and from the right. 

We place two vertices $A_i^l$ and $A_i^r$ of $G_{\tau}^{\mu}\setminus H_{\tau}^{\mu}$ in sufficiently small neighborhood of $v_i$, such that both are slightly higher than~$v_i$, $A_i^l$ lies to the left from~$v_i$ and $A_i^r$ lies to the right, and add straight-line segments $v_iA_i^l$ and $v_iA_i^r$. Since $A_i^l$ and $A_i^r$ are placed close to $v_i$, and~$v_i$ has no other neighbors outside $[b,t]$, we may assume without loss of generality that for each~$i\in [q-1]$, $v_iA_i^lA_{i+1}^lv_{i+1}$ and $v_iA_i^rA_{i+1}^rv_{i+1}$ are parallelograms that do not intersect rest of the drawing created so far. Therefore, we can draw the edges $A_i^l A_{i+1}^l$ and $A_i^r A_{i+1}^r$ (in case~$i<q$) as straight-line segments. Similarly, we draw the edges $bA_1^l$, $bA_1^r$, $A_1^lt$ and $A_1^rt$.

Next, we follow every edge of $\Gamma(G)$ of the form $\{A_i,a\}$ that leaves $A_i$ to the left (resp.\ to the right) of $bt$ to draw the edge $\{A_i^l,a\}$ (resp.\ $\{A_i^r,a\}$) of ($G_{\tau}^{\mu}$), where $a$ is either missing or pre-drawn vertex of $\III$. The resulting drawing provides solution to $\III_{\tau}^{\mu}$.  
\end{proof}

\begin{lemma}
If for some choice of $\tau$, $\III_{\tau}^{\mu}$ is a \textup{YES}-instance, then so is the original instance~$\III$. 
\end{lemma}
\begin{proof}
Let $\tau= (A_1, \dots, A_q)$ and $\mu=(\mu_{i,j}:i\in[q], j\in V_M \cup \{u\})$ be such that $\III_{\tau}^{\mu}=(G_{\tau}^{\mu},H_{\tau}^{\mu},\Gamma(H_{\tau}^{\mu}))$ is a YES-instance, and let $\Gamma(G_{\tau}^{\mu}) \supseteq \Gamma(H_{\tau}^{\mu})$ be the solution to $\III_{\tau}^{\mu}$. We modify it to obtain a solution $\Gamma(G)$ to $\III$ as follows. 
Note that if $q = 0$, then no $A_i$ is placed on the vertical segment $bt$ and therefore, $\Gamma(G_{\tau}^{\mu})$ immediately yields a solution for $\mathcal{I}$ and we remain with the case $q \geq 1$.

First, note that 
for each $i \in [q]$, $A_i^l$ and $A_i^r$ must be drawn on different sides from $[b,t]$: otherwise the drawing $\Gamma(G_{\tau}^{\mu})$ would not be planar. Without loss of generality, assume that all $A_i^l$ are attached from the left and all $A_i^r$ are attached from the right. 
If some quadrilateral~$v_iA_i^lA_{i+1}^lv_{i+1}$ contains some missing vertices in the interior, then these vertices belong to~$V_M \setminus \{A_1, \dots, A_q\}$, so they can be attached only to $A_i^l$ and $A_{i+1}^l$. Hence, the part of the drawing that lies inside $v_iA_i^lA_{i+1}^lv_{i+1}$ can be shifted very close to the edge~$A_i^lA_{i+1}^l$ and then sent to the other side of this edge (outside of the quadrilateral). In a similar way, we may achieve that all quadrilaterals $v_iA_i^rA_{i+1}^rv_{i+1}$, $i\in[q-1]$ as well as triangles~$bA_1^lv_1$,~$bA_1^rv_1$,~$A_q^lv_qt$, and~$A_q^rv_qt$, have empty interiors.

Next, we modify only the part of the drawing that lies to the left from $[b,t]$ as follows. For each $i \in [q]$, the horizontal projection of $A_i^l$ into $[b,t]$ is denoted by $C_i$. Note that the presence of the edge $A_i^lA_{i+1}^l$ guarantees that $C_i$ lies below $C_{i+1}$ for each $i \in[q-1]$. We delete all the edges $v_iA_i^l$ and $A_i^lA_{i+1}^l$, $i \in [q-1]$, as well as the edges $v_qA_q^l$, $bA_1^l$ and $A_q^lt$. After this, we modify each remaining edge of the form $\{a,A_i^l\}$ as follows. We keep the edge unchanged outside the small neighborhood of $A_i^l$, then we add an almost horizontal line segment that ends in $C_i$. Once such a redrawing is applied for all edges incident with $A_i^l$, the vertex $A_i^l$ is removed. When all the vertices $A_i^l$, $i \in [q]$, are removed, we subdivide the left part of the obtained drawing into horizontal strips defined by the vertices $C_i$, $i \in [q]$, and one by one compress or stretch these strips along the vertical axis to make them equally long. In particular, each $C_i$ becomes $v_i$. We then perform the same procedure for the right part of the drawing, by defining horizontal projections $D_i$ of $A_i^r$ into $[b,t]$, $i \in [q]$. Finally, we replace each $C_i=v_i=D_i$ by $A_i$. The resulting drawing $\Gamma(G)$ is a solution to $\III$. 
\end{proof} 
Therefore, it suffices to solve at most $2^{k^2+2k} k!$ instances $\III_{\tau}^{\mu}$ of \textsc {UPEF}.
Note that every instance $\III_{\tau}^{\mu}$ has at most twice as many missing vertices as $\mathcal{I}$.
Further, the size of the instance~$\III_{\tau}^{\mu}$ is polynomial in the size of~$\mathcal{I}$.
Therefore we obtain the following result:
\begin{corollary}
\label{cor:outer}
Let $\mathbb{A}$ be a fixed-parameter algorithm solving \upefshort in time at most $\bigoh^*(f(k))$ where $f$ is a non-decreasing function. Then there exists fixed-parameter algorithm $\mathbb{B}$ solving \upefshorto in time at most $\bigoh^*(2^{k^2+2k} k! \cdot f(2k))$.
\end{corollary}

By combining Corollary~\ref{cor:outer} with Theorem~\ref{thm:main_face} and Lemma~\ref{lem:setup}, we finally obtain:

\begin{theorem}
\label{thm:main}
\upeshort can be solved in time at most $\bigoh^*(k^{\mathcal{O}(k^3)})$, where $k$ is the vertex+edge deletion distance.
\end{theorem}

\section{Concluding Remarks}
We believe that the insights obtained on our way towards solving \textsc{Upward Planar Drawing Extension} not only reveal previously unknown fundamental combinatorial properties of upward planar drawings, but can also pave the way towards a broader understanding of drawing extension problems. In particular, we believe that the bundling framework employed in Section~\ref{sec:properties} can be lifted to different settings, such as 
extending partial level-planar~\cite{KlemzS24}, 1-planar~\cite{EibenGHKN20, mfcs/EibenGHKN20} or straight-line planar drawings~\cite{Patrignani06}.
Other interesting questions for future work include the development (or exclusion) of subexponential or kernelization algorithms for \textsc{Upward Planar Drawing Extension}; in fact, very little is currently known about the existence of such algorithms even for other drawing extension problems.
It is also open whether the fixed-parameter tractability still holds if the pre-drawn graph $H$ is disconnected.

\bibliography{biblio}

\end{document}